\newif\ifjrssb
\jrssbfalse

\ifjrssb
\documentclass{statsoc}
\usepackage[a4paper, top=1in, bottom=1in, left=1in, right=1in]{geometry}

\else
\documentclass[letterpaper,11pt]{article}
\usepackage{geometry}
\usepackage{authblk}
\fi

\usepackage{amsmath}
\usepackage{commath}
\usepackage{amsthm}
\usepackage{amssymb}
\usepackage{setspace,mathrsfs}
\usepackage{algorithm}
\usepackage{algorithmic}
\usepackage{url,amsmath,amsfonts,amssymb}
\usepackage{graphicx}
\usepackage{listings,enumerate,color,relsize,multirow,rotating}
\usepackage{bbm}
\usepackage{natbib}
\usepackage{hyperref}
\hypersetup{colorlinks,citecolor=blue,urlcolor=blue,linkcolor=blue}
\usepackage{bm}
\usepackage{lipsum}
\usepackage[english]{babel}
\usepackage{enumitem}
\usepackage{footmisc}
\usepackage{makecell}

\definecolor{purple}{rgb}{0.84, 0.17, 0.89}
\definecolor{red2}{rgb}{0.7, 0, 0.1}

\usepackage{color}
\newcommand{\darkred}{\color{black}}

\usepackage{my_tex}

\def\ep{\varepsilon}
\def\bep{\boldsymbol{\ep}}

\def\bs{\boldsymbol{s}}

\newcommand\blfootnote[1]{%
  \begingroup
  \renewcommand\thefootnote{}\footnote{#1}%
  \addtocounter{footnote}{-1}%
  \endgroup
}

\ifjrssb
\title[Maxway CRT]{Maxway CRT: Improving the Robustness of the Model-X Inference}
\author{Shuangning Li}
\address{Department of Statistics, Harvard University}
\author[Shuangning Li and Molei Liu]{Molei Liu \footnote{The two authors have equal contributions to this work. We thank Emmanuel Cand\`es, Lucas Janson, Lihua Lei, and Zhimei Ren for insightful comments and helpful discussions, and Tianxi Cai for helpful suggestions on the real case study.}}
\address{Department of Biostatistics, Columbia Mailman School of Public Health}

\else
\title{Maxway CRT: Improving the Robustness of the Model-X Inference \blfootnote{The two authors have equal contributions to this work. We thank Emmanuel Cand\`es, Lucas Janson, Lihua Lei, and Zhimei Ren for insightful comments and helpful discussions, and Tianxi Cai for helpful suggestions on the real case study.
}}
\author[1]{Shuangning Li}
\author[2]{Molei Liu}

\affil[1]{Department of Statistics, Harvard University}
\affil[2]{Department of Biostatistics, Columbia Mailman School of Public Health}
\fi

\begin{document}

\maketitle
\date{}

\begin{abstract}
\noindent The model-X conditional randomization test (CRT) proposed by \cite{candes2018panning} is known as a flexible and powerful testing procedure for conditional independence: $X\indp Y\mid Z$. Though having many attractive properties, it relies on the model-X assumption that we have access to perfect knowledge of the distribution of $X$ conditional on $Z$. If there is an error in modeling the distribution of $X$ conditional on $Z$, this approach may lose its validity. This problem is even more severe when the adjustment covariates $Z$ are of high dimensionality, in which situation precise modeling of $X$ against $Z$ can be hard.
In response to this, we propose the Maxway ({\bf M}odel and {\bf A}djust {\bf X} {\bf W}ith the {\bf A}ssistance of {\bf Y}) CRT, a robust inference framework for conditional independence when the conditional distribution of $X$ is unknown and needs to be estimated from the data. The Maxway CRT learns the distribution of $Y\mid Z$, using it to calibrate the resampling distribution of $X$ to gain robustness to the error in modeling $X$. 
We prove that the type-I error inflation of the Maxway CRT can be controlled by the learning error for the low-dimensional adjusting model plus the product of learning errors for $X\mid Z$ and $Y\mid Z$, which could be interpreted as an ``almost doubly robust" property. {\darkred Based on this, we develop implementing algorithms of the Maxway CRT in practical scenarios including (surrogate-assisted) semi-supervised learning and transfer learning where valid information about $Y\mid Z$ can be potentially provided by some auxiliary or external data.} Through extensive simulation studies under different scenarios, we demonstrate that the Maxway CRT achieves significantly better type-I error control than existing model-X inference approaches while preserving similar powers. {\darkred Finally, we apply our methodology to two real examples, including (1) studying obesity paradox with electronic health record (EHR) data assisted by surrogate variables; (2) inferring the side effect of statins among the ethnic minority group via transferring knowledge from the majority group.}

\end{abstract}


\section{Introduction}
In many fields such as biology, biomedical science, economics, and political science, it is often of great importance to understand the causal or association relationship between some response $Y$ and some explanatory variable $X$ conditioning on a large number of confounding variables $Z\in \mathbb{R}^p$. For example, geneticists may want to know whether a particular genetic variant is related to the risk of a disease conditional on other genetic variants. 
These problems are often handled by modeling $Y$ against $X$ and $Z$ through some parametric or semiparametric model. Researchers may consider a gaussian linear model, and encode the relevance between $Y$ and $X$ conditional on $Z$ in some key model parameter, e.g. the regression coefficient of $X$. Methodology and theory for the statistical inference of such parametric or semiparametric models have been well studied and widely applied in practice; see e.g. \cite{chernozhukov2016double}. However, this relatively standard and classic strategy has been criticized as being invalid and powerless in certain cases due to potential model misspecification and relatively limited observations of $Y$.

As an alternative strategy, the model-X framework and conditional randomization test (CRT) proposed by \cite{candes2018panning} formulate this fundamental and central problem as testing for the general conditional independence hypothesis $H_0: X\indp Y\mid Z$ free of any specific effect parameters. Instead of imposing model assumptions for $Y\mid (X,Z)$ and testing for the effect of $X$ based on its estimate, the model-X CRT assumes the distribution of $X\mid Z$ to be known. With perfect knowledge of the distribution of $X\mid Z$, it controls the type-I error exactly (non-asymptotically) and allows for the choice of any test statistic.
This strategy can be particularly useful when there is either strong and reliable scientific knowledge of the distribution of $X\mid Z$ \citep[e.g.,][]{sesia2020multi,bates2020causal} or an auxiliary dataset of $(X,Z)$ of potentially large sample size, known as the semi-supervised setting. Though either case enables more precise than the usual characterization of the conditional distribution of $X$ given $Z$, the model-X CRT in these practical situations is still far from being perfect; it suffers from estimation errors, raising concerns regarding its robustness under imperfect knowledge of $X\mid Z$. 

In this paper, our goal is to develop an approach that improves the robustness of the model-X CRT, i.e., lessens type-I error inflation arising from the specification error of $X\mid Z$, while at the same time preserving its advantages over parametric (or semiparametric) inference, e.g., generality, flexibility, and powerfulness.

\subsection{Background}\label{sec:background}

Suppose there are $n$ i.i.d. samples of $(Y,X,Z)$ denoted as $(Y_i,X_i,Z_{i \cdot})$, and let $\by=(Y_1,Y_2,\ldots,Y_n)\trans\in\mathbb{R}^n$, $\bx=(X_1,X_2,\ldots,X_n)\trans\in\mathbb{R}^n$, and $\Z=(Z_{1\cdot},Z_{2\cdot},\ldots,Z_{n\cdot})\trans\in\mathbb{R}^{n\times p}$. The goal is to test whether
\begin{equation}
\label{eqn:indep_test}
X \indp Y \mid Z,
\end{equation}
i.e., whether $X$ provides extra information about $Y$ beyond what is already provided by $Z$. The model-X CRT presented in Algorithm \ref{alg:crt} is a general framework for testing \eqref{eqn:indep_test}. 
Proposition \ref{thm:crt} establishes that with perfect knowledge of the distribution of $\x\mid\Z$, this approach controls the type-I error rate exactly. 
\begin{algorithm}[htbp]
\caption{\label{alg:crt} The model-X conditional randomization test (CRT).}
{\bf Input:} Knowledge of the distribution of $\X\mid\bZ$, data $\bD=(\by,\bx,\bZ)$, test statistic (i.e., importance measure) function $T$, and number of randomizations $M$.\\
\vspace{0.2cm}
{\bf For} $m=1,2,...,M$: Sample $\bx\supm$ from the distribution of $\bx\mid\bZ$ independently of $(\bx,\by)$.\label{algline:cr}\\
\vspace{0.2cm}
{\bf Output:} CRT $p$-value $p_{\operatorname{mx}}(\bD)=\frac{1}{M+1}\left(1+\sum_{m=1}^M \One{T(\by,\bx\supm,\bZ)\geq T(\by,\bx,\bZ)}\right)$.
\end{algorithm}

\begin{prop}[\cite{candes2018panning}]
\label{thm:crt}
The $\pval$ $p_{\operatorname{mx}}(\bD)$ satisfies $\mathbb{P}_{H_0}(p_{\operatorname{mx}}(\bD)\le \alpha)\le\alpha$ for any $\alpha\in [0,1].$
\end{prop}

Now suppose that we do not know the distribution of $\x\mid\Z$, instead, we learn it through some learning algorithms based on some auxiliary information and dataset; see Section \ref{sec:in_practice} for some specific examples.
Let $f(\x \mid \Z)$ be the probability (density or mass) function of $\x$ given $\Z$, and let $\fhat(\x\mid\Z)$ be the estimate of $f(\x \mid \Z)$ learned from the auxiliary data. Then we use the fitted $\fhat(\x\mid\Z)$ as the (imperfect) knowledge of the distribution of $\bx\mid\bZ$ in Algorithm \ref{alg:crt} to implement the CRT and obtain the $p$-value denoted by $p_{\operatorname{mx}}(\bD;\fhat)$. As shown by \cite{berrett2018conditional}, for certain test statistic, the type-I error inflation of the model-X CRT with $p_{\operatorname{mx}}(\bD;\fhat)$ can be as large as the expectation of $d\subTV(f,\fhat)$, where $d\subTV(a,b)$ represents the total variation distance between two distributions $a$ and $b$. Meanwhile, for semiparametric inference approaches like post-double-selection \citep{belloni2014inference} and double machine learning \citep{chernozhukov2016double}, their type-I error inflation can generally be expressed as $n^{-1/2}$ (built upon the regular central limit theorem) plus the product of the learning errors for $X\mid Z$ and $Y\mid Z$, the latter known as ``double robustness". Nevertheless, the model-X CRT, though also used for detecting conditional independence, does not pursue such double robustness through learning and adjusting for both $X\mid Z$ and $Y\mid Z$ like post-double-selection or DML. Thus, it is substantially more sensitive to the learning error of $X\mid Z$. As illustrated in Figure 1 of \cite{chernozhukov2016double}, in a different context but in a similar vein, such loss in robustness can incur severe invalidity.


\subsection{Our contribution}\label{sec:intro:cont}

To improve the robustness of the model-X CRT to the specification error of the distribution of $X\mid Z$, we propose a new approach for conditional independence testing named Maxway (Model and Adjust $X$ with the assistance of $Y$) CRT. As a special class of the CRT procedure, the Maxway CRT is more elaborate in terms of both specifying the resampling distribution of $X$ and choosing the test statistic $T$, achieving the double robustness in type-I error control that cannot be readily achieved by the model-X CRT. In addition to modeling $X$ through some $h(\Z)$ sufficient to characterize $\x\mid\Z$, our approach learns and extracts some low-dimensional $g(\Z)$ sufficient for characterizing the dependence of $\by$ on $\Z$. Next, it adjusts the estimated distribution of $\x\mid\Z$ against $g(\Z)$ and resamples from the adjusted distribution for randomization testing with the test statistic restricted to the form $T(\by,\x,g(\Z),h(\Z))$.

As we show, the gain of the additional adjustment is to substantially reduce the type-I error inflation of the model-X CRT (See equation (\ref{eqn:original_CRT_bound_var})) to $\mathbb{E}[\Delta_x\Delta_y]+\mathbb{E}[\Delta_{x|g,h}]$ of the Maxway CRT, where the terms $\Delta_x$, $\Delta_y$ and $\Delta_{x|g,h}$ represent the estimation errors (measured by the total variation distance) of the distributions of $\x\mid \Z$, $\by\mid \Z$ and $\x\mid \{g(\Z),h(\Z)\}$ respectively; see Section \ref{sec:theory:robust:1} for more details. As mentioned in Section \ref{sec:background}, the term $\Delta_x\Delta_y$ captures the double robustness in the sense that either accurate specification of $\x\mid \Z$ or $\by\mid \Z$ means that this term will be small. Meanwhile, $\Delta_{x|g,h}$ is the learning error under a much lower dimensionality than that of $\Z$, and thus can be substantially smaller than $\Delta_x$ with the most typical learning tools used in practice. In Section \ref{section:specific_stats}, we also establish sharper bounds for the type-I error inflation when the Maxway CRT is applied with specific test statistics. 

{\darkred
Based on this Maxway framework, we design novel robustified CRT algorithms under three scenarios frequently appearing in contemporary data-driven studies, including typical semi-supervised learning (SSL), surrogate-assisted SSL (SA-SSL), and transfer learning (TL). In both the SA-SSL and TL settings, large auxiliary samples are incorporated to learn side information about the above-introduced low-dimensional $g(\Z)$ characterizing $Y\mid \Z$. Consistent with our theoretical findings, such procedures are shown to enhance the robustness of the CRT effectively. We demonstrate the finite-sample utility and practical versatility of the proposed Maxway methods under the three scenarios through comprehensive simulation studies and two real examples.

}



\subsection{Related work}\label{sec:intro:com}

Our work builds upon the model-X framework proposed by \cite{candes2018panning}. The robustness of the commonly used model-X approaches under imperfect knowledge of the conditional distribution of $X$ is a crucial problem and has been frequently studied. Specifically, \cite{barber2020robust} showed for the knockoffs \citep{candes2018panning} that the inflation of its false discovery rate is proportional to the estimation error of every single feature's conditional distribution given the remaining features. Similarly, \cite{berrett2018conditional} proved that type-I error inflation of the model-X CRT conducted with an estimated conditional distribution of $X$ can be bounded by the total variation distance between the estimated distribution and the true distribution. 

Recent progress has enhanced the robustness of model-X inference. \cite{huang2020relaxing} proposed a robust knockoff inference procedure that is conditional on the observed sufficient statistic of the model of the features. Their approach is shown to be exactly valid when the features' distribution is not known but a parametric model is correctly specified for the conditional distribution of the response given the features. 
\cite{berrett2018conditional} proposed conditional permutation test (CPT), a variant of the CRT that fixes $\x$'s values and randomizes only on its permutation given $\Z$ to conduct the randomization test. They showed that the proposed CPT has no larger type-I error inflation than the CRT. \cite{sudarshan2021contra} enhanced robustness of the holdout randomization test \citep{tansey2018holdout} using an equal mixture of two contrarian models for importance measure. When $X\mid Z$ is misspecified, their importance measure fits better on the resampled $X$ than the observed $X$ and their procedure alleviates the inflation of FDR. 

To the best of our knowledge, \cite{berrett2018conditional} is perhaps the most relevant to our approach in terms of the overall setup. However, our approach mainly aims to handle the failure or inadequacy in adjusting for the effect of $Z$ confounding $X$ and $Y$, whereas the CPT generally handles all types of specification errors, e.g., the estimation error in $\sigma^2$ when $X\mid Z\sim\mathcal{N}(0,\sigma^2)$. Both theoretically and numerically, we show that our approach is more effective than the CPT in removing the confounding effects of $Z$ to control the type-I error inflation. In addition, the idea of \cite{berrett2018conditional} is to condition on more information than just $\Z$ for randomization tests, i.e. the order statistics of $\x$ given $\Z$. This strategy tends to make the test conservative and to reduce the power compared with the original CRT, no matter whether $X\mid Z$ is estimated well or not. In comparison, as discussed in Section \ref{sec:sim}, our approach does not necessarily sacrifice power due to such overall conservativeness. Finally, our approach can be naturally combined with \cite{berrett2018conditional} for potentially more robust inference than either of them alone, which means ours is a tool-kit not to replace but to improve existing robust model-X inference procedures.

In our approach, the test statistic depends on $\Z$ only through some low-dimensional $g(\Z)$ and $h(\Z)$ given $\x$ and $\by$. This restriction is technically similar to the distilled CRT (dCRT) proposed by \cite{liu2020fast}. Their purpose is to reduce the computation burden of the CRT. dCRT first distills the information about $\x\mid\Z$ and $\y\mid\Z$ into some low-dimensional functions of $\Z$. Then, it constructs the test statistic with the distilled data (and $\by,\x$) to avoid refitting any high dimensional models and thus speed up the resampling procedure. {\darkred As was shown by extensive simulation studies in \cite{liu2020fast} and asymptotic power analysis in \cite{katsevich2020theoretical} and \cite{wang2020power}, restricting to test statistic of the form $T(\by,\x,g(\Z),h(\Z))$ in the CRT preserves the high power of the original CRT in most common setups like partially linear and hierarchical interaction relationships between $X$ and $Y$ given $Z$. In Section \ref{sec:discuss}, we discuss on the relationship between our approach and the dCRT, as well as the implication of the above-mentioned power studies on the performance of our approach.
}


Our work is related to the semiparametric (asymptotic) inference approaches that have been studied for a long time. Among them, the post-double-selection approach \citep{belloni2014inference} and the double machine learning (DML) framework \citep{chernozhukov2016double} are perhaps the most relevant to us, since they also rely on complex machine learning algorithms to estimate $X\mid Z$ and $Y\mid Z$. One recent manuscript, \cite{dukes2021doubly} proposed a calibration method to further improve the robustness of the DML inference to inconsistent or misspecified machine learning (ML) estimators of the nuisance models. Their idea of calibration seems similar to ours from a high-level viewpoint. Nonetheless, these approaches have substantially different contexts and utility from our work. Specifically, it only covers some special cases of the testing models for $Y$, e.g., the partially linear or generalized linear model \citep{denis2021regularized,liu2021double}. For more complex test statistics accommodated by our approach, e.g., the variable importance of random forest, DML cannot provide a valid implementation. Also, DML draws inference based on the asymptotic normality of the test statistics, while our approach is non-asymptotic and thus more robust to data with small sample sizes or heavy tail variables; see Section 4.2 and D.5 of \cite{liu2020fast} for demonstration. 


{\darkred
In addition, we consider the implementation of our new framework in several prevalent practical scenarios frequently studied in recent literature, including semi-supervised learning (SSL), surrogate-assisted SSL (SA-SSL), and transfer learning (TL). In the model-X context, \cite{berrett2018conditional} also studied the SSL setting where the unlabeled (without $Y$) data have a large sample size to enable relatively accurate (but not perfect) estimation of $\x\mid \Z$. Our strategy of using some prior estimators about $\y\mid \Z$ learned from learned auxiliary data to guide the CRT, is conceptually related to lots of recent work in SA-SSL \citep[e.g][]{hong2019semi,zhang2022prior} and TL \citep[e.g.][]{li2022transfer,tian2022transfer,cai2022semi}, which has shown the great potential of application in fields like biomedical studies. While the main focus of this track is usually estimating $\y\sim \Z$ itself instead of adjusting for it in conditional independence testing, it also highlights adaptivity and robustness to the poor quality of the external knowledge about $\y\mid \Z$, which are pursued in our proposed framework as well.

Finally, we notice a recent paper by \citet{niu2022reconciling}, where the authors show that the model-X CRT with the d$_0$ statistic is doubly robust under weak assumptions. 
We want to point out that the results in \cite{niu2022reconciling} do not contradict with ours; rather, the two sets of results are largely orthogonal. In essence, the Maxway CRT achieves (almost) double robustness through sampling $\x^{(m)}$ conditioning additionally on $g(\Z)$, which captures the relationship between $\y$ and $\Z$. On the other hand, the d$_0$CRT achieves double robustness through specific construction of statistics, in this case, the d$_0$ statistic. Our Maxway CRT framework can be applied with more general test statistics, including ones involving random forest (See Implementation example \ref{example:2}). Moreover, we can combine the Maxway CRT framework with the d$_0$ statistic to increase the procedure's robustness even further. We will demonstrate this through simulation studies in Section \ref{sec:sim}.}

\section{Method}\label{sec:method}
\subsection{Motivation: a simple model}
\label{subsection:motivation}

Suppose $Z$ is a high dimensional random variable, but $X$ and $Y$ depend only on a small subset of the $Z_j$'s. Specifically, assume that $Z=(Z_1,\ldots,Z_p)$, and that conditioning on $Z$, $X$ and $Y$ follow the following simple model under the null hypothesis: 
\begin{equation}
\label{eqn:mov_linear}
X = \phi(Z_1, Z_2) + \varepsilon,
\quad 
Y = \psi(Z_2, Z_3) + \eta,
\end{equation}
where $\varepsilon, \eta \sim \mathcal{N}(0,1)$ independently of $Z$. 
Let $\mcS_{x} = \cb{1,2}$ and $\mcS_{y} = \cb{2,3}$. Correspondingly, 
let $Z_{\mcS_{x}} = \cb{Z_1, Z_2}$ be the set of $Z_j$'s that $X$ depends on and $Z_{\mcS_{y}} = \cb{Z_2, Z_3}$ be the set of $Z_j$'s that $Y$ depends on. 
In such a setting, in order to implement the original model-X CRT, we need to know the set $\mcS_x$ and the distribution of $X \mid Z_{\mcS_x}$. 
Assume for now that given a set $\mcS$ whose cardinality is not huge, we are able to learn the distribution of $X \mid Z_S$ accurately. 
In practice, this can be achieved when one is willing to make some assumptions on the function $\phi$; for example, linear regression could be used if $\phi$ is assumed to be linear and kernel regression can be used if $\phi$ is assumed to be smooth. 
Now if we are given a set of indices $\mcS$ and the test statistic $T(\cdot)$ is taken to be a function of $\x$, $\y$ and $\Z_{\cdot \mcS}$, then the procedure of the model-X CRT simplifies to the following:
\begin{enumerate}
\item For $m=1,2,...,M$: 

\qquad Sample $\bx\supm$ from the distribution of $\bx\mid\bZ_{\cdot \mcS}$ independently of $(\bx,\by)$.
\item Output the model-X CRT $p$-value
\begin{equation}
p_{\operatorname{mx}}(\bD)=\frac{1}{M+1}\left(1+\sum_{m=1}^M \One{T(\by,\bx\supm,\bZ_{\cdot \mcS})\geq T(\by,\bx,\bZ_{\cdot \mcS})}\right). 
\end{equation}
\end{enumerate}

On the one hand, we can immediately verify that the above procedure produces a valid $\pval$ if the set $\mcS$ contains $\cb{1,2}$. In this case, $Z_{\mcS}$ contains all the $Z_j$'s that $X$ depends on, and thus under the null hypothesis, the distribution of $\x \mid \Z_{\cdot \mcS}$ is the same as the distribution of $\x \mid \Z$. Therefore, the validity of the procedure follows from the standard analysis of the model-X CRT. 

On the other hand, the $\pval$ is also valid if the set $\mcS$ contains $\cb{2,3}$. To see this, note that when $\mcS$ contains all the $Z_j$'s that $Y$ depends on, $Y \indp Z \mid Z_{\mcS}$. Under the null hypothesis, we have $X \indp Y \mid Z$, therefore, $X \indp Y \mid Z_{\mcS}$. A close look at the above procedure shows that it can be treated alternatively as a CRT for $(X, Y, Z_{\mcS})$, i.e., the procedure is a valid testing procedure for whether $X \indp Y  \mid Z_{\mcS} $. These imply that the $\pval$ is valid if the set $\mcS$ is a superset of $\mcS_y$.
                                                                           
At a high level, the above observation shows that some knowledge of how $Y$ depends on $Z$ can be useful in enhancing robustness in CRT. More specifically,  to achieve validity in this motivating example, the best we can hope for in the absence of any information or prior knowledge on $Y$ is the set $\mcS$ to contain the index set $\mcS_x$. Extra information on the distribution of $Y$ relaxes the condition of the validity of the $\pval$; the procedure is valid when the set $\mcS$ contains either $\mcS_x$ or $\mcS_y$. 

\subsection{Maxway: model and adjust X with the assistance of Y}

More generally, assume that under the null hypothesis, there exist low-dimensional functions $h(\cdot)$ and $g(\cdot)$ that are sufficient for characterizing the distribution of $X\mid Z$ and $Y\mid Z$ in the sense that 
\begin{equation}
X\indp Z\mid h(Z),\quad Y\indp Z\mid g(Z).
\end{equation}
For the example we considered in Section \ref{subsection:motivation}, one possible choice of $h$ and $g$ is $h(Z) = (Z_1, Z_2)$ and $g(Z) = (Z_2, Z_3)$. Figure \ref{fig:diagram} shows an illustration of the relationship of $X, Y, Z$ in this example. In other settings, the form of $h$ and $g$ could be different. For example, if we consider a sparse linear model, i.e. let $X = \theta\trans Z_{\mcS_x} + \varepsilon$, where $\mcS_x \subset \cb{1, \dots, p}$ and $\varepsilon \sim \mathcal{N}(0,1)$ independently of $\Z$, then $h(Z)$ could be the vector of random variables $Z_{\mcS_x}$, or it could also be the mean function $\theta\trans Z_{\mcS_x}$. In particular, $X \indp Z \mid Z_{\mcS_x}$ and $X \indp Z \mid \theta\trans Z_{\mcS_x}$. 
In practice, the choice of functions $h$ and $g$ may not be ``perfect". In other words, $X \indp Z \mid h(Z)$ and $Y \indp Z \mid g(Z)$ may not hold exactly for our specified $h$ and $g$. 
Oftentimes, we learn $h$ and $g$ from some auxiliary or external data; see Section \ref{sec:in_practice} for more details. 

\begin{figure}[t]
    \centering
    \includegraphics[width=0.7\textwidth]{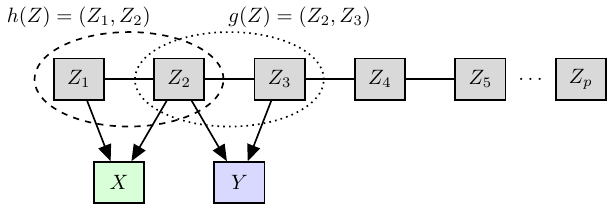}
    \caption{An illustration of a possible choice of the functions $h$ and $g$ in the simple model}
    \label{fig:diagram}
\end{figure}

Throughout this section, 
our goal is to study whether a phenomenon similar to that described in Section \ref{subsection:motivation} can be seen here and whether it is possible to develop a more general procedure that uses information of $Y$ to enhance the robustness of the model-X CRT. Algorithm \ref{alg:maxway} describes an algorithm that makes use of such information and is more robust to misspecification of the distribution of $X\mid Z$. The name ``Model and adjust $X$ with the assistance of $Y$" emphasizes the importance of X-modeling itself and the helpfulness of Y-modeling in enhancing robustness. 

Before proceeding, we introduce a few more notations here. For any function $a$ of $Z$, we often overload the notation and write $a(\Z) = (a(Z_{\cdot 1}), \dots, a(Z_{\cdot n}))\trans$. Let $\rhos(\cdot \mid g(Z), h(Z))$ be the conditional distribution of $X$ given $g(Z)$ and $h(Z)$, and 
$\rho^{\star n}(\x \mid g(\Z), h(\Z)) = \prod_{i = 1}^n \rhos(X_i \mid g(Z_{i \cdot}), h(Z_{i \cdot}))$. Let $\rho$ be an estimate of $\rhos$ and $\rho^{n}(\x \mid g(\Z), h(\Z)) = \prod_{i = 1}^n \rho(X_i \mid g(Z_{i \cdot}), h(Z_{i \cdot}))$. We call $\rho$ the \textit{Maxway distribution}.



\begin{algorithm}
\caption{The model and adjust $X$ with the assistance of $Y$ (Maxway) CRT}
\label{alg:maxway}
{\bf Input:} functions $h$, $g$; Maxway distribution $\rho$, which is an estimate of $\rhos$; data $\bD=(\by,\bx,\bZ)$; test statistic (i.e. importance measure) function $T$; and number of randomizations $M$.

\vspace{3pt}

\textbf{For} $m=1,2,...,M$:

\qquad Sample $\x\supm$ from the Maxway distribution $\rho^n (\cdot \mid g(\Z), h(\Z))$ independently of $(\bx,\by)$.

\vspace{3pt}

{\bf Output:} The Maxway CRT $p$-value 
\begin{equation}
\label{eqn:maxway_CRT}
p_{\operatorname{maxway}}(\bD)= \frac{1}{M+1}\p{1+\sum_{m=1}^M \One{T(\y,\x\supm,g(\Z), h(\Z))\geq T(\y,\x, g(\Z), h(\Z) )}}. 
\end{equation}
\end{algorithm}

While the model-X CRT samples $\x\supm$ from an estimated distribution of $\x \mid \Z$, the Maxway CRT samples $\x\supm$ from an estimated distribution of $\x \mid g(\Z), h(\Z)$. In order to have a more intuitive comparison of the two, consider 
the model-X CRT where $\x\supm$'s are sampled from $\x \mid h(\Z)$. Compared to this, the Maxway CRT additionally conditions $\x$ on $g(\Z)$ after conditioning on $h(\Z)$. The name of the proposed approach, ``Maxway", captures exactly this point. ``Model-X" corresponds to conditioning $\x$ on $h$, and ``adjust X with the assistance of Y" corresponds to further adjusting for $g$. 
After conditioning on $h$, there may still be information about $\Z$ left in the residual of $\x$. The ``adjusting for $g$" step aims at removing this information in the residual, especially the part which is dependent on $g$ and $\y$. 



{We shall point out that taking $\rho^{\star}$ to be the conditional distribution of $\x$ given $g(\Z)$ and $h(\Z)$ is neither the only way nor the most effective way to construct the Maxway distribution adjusting for $g(\Z)$ based on $\x\mid h(\Z)$. In the last paragraph of Section \ref{sec:discuss}, we propose a more general definition of the Maxway distribution to be used in Algorithm \ref{alg:maxway} that includes the current definition as a special case.} We also note that the Maxway CRT (Algorithm \ref{alg:maxway}) requires the test statistic $T$ to be a function of $\y$, $\x$ and $g(\Z)$ and $h(\Z)$. This requirement is not as general as that of the model-X CRT, where $T$ can be taken to be a function of the full $\Z$. Nonetheless, this specific form of the test statistic has been shown to have a computational advantage by \citet{liu2020fast}, and more importantly, this form of the test is typically powerful \citep{katsevich2020theoretical}. 

\subsection{A transformed Maxway CRT}

In implementing the Maxway CRT, a natural question to ask is whether there always exists a low-dimensional function $h(Z)$ responsible for all the dependence of $X$ given $Z$, i.e., whether there exists a low-dimensional $h$ such that $X \indp Z \mid h(Z)$, as well as the same question for $g(Z)$ and $Y$. The answer appears to be negative at first sight. Think of a case where the distribution of $X\mid Z$ is learned with a generative adversarial network, and tens of thousands of parameters are used to describe the conditional distribution of $X$ given $Z$; thus the existence of a low-dimensional $h$ does not appear to be plausible. Nonetheless, in this example, there still exists a proper {\em residual-transformation} function $R$ such that $R(X,Z) \indp Z \mid h(Z)$ where $h(Z)=\emptyset$. 
In specific, one can set $R(X, Z) = F_X(X \mid Z)$ to be the cumulative distribution function (CDF) of $X$ conditional on $Z$. With this transformation, if $X$ has a continuous distribution conditional on $Z$, then $R(X, Z) = F_X(X \mid Z)$ follows a $\operatorname{Unif}[0,1]$ distribution and is independent of $Z$, which can be viewed as a ``residual" of $X$ on $Z$.  Hence, $h$ can simply be taken as the null set. For a discrete $X$, we can always add an arbitrarily small continuous random noise to it and use the perturbed $X$ for such CDF transformation. Since the artificial noises can be arbitrarily small, the power of the Maxway CRT will not be impacted by the perturbation. 
We again emphasize that the above discussion shows that there exists a \textit{perfect} transformation $R$ and a \textit{perfect} low-dimensional $h$, such that $R(X,Z) \indp Z \mid h(Z)$. In practice, however, such functions $R$ are estimated from data, and thus $R(X,Z) \indp Z \mid h(Z)$ may not hold exactly. 

With this transformation step, we then proceed with testing whether $R(X,Z)$ is independent of $Y$ conditional on $Z$. We note that this transformation step would not invalidate the testing procedure. Under the null hypothesis, i.e., when $X \indp Y \mid Z$, $R(X,Z)$ is also independent of $Y$ given $Z$. So any valid testing procedure will still be valid if $X$ is replaced by $R(X,Z)$. Conceptually, this transformation step is essentially extracting the residual of $X$ after removing the influence of $Z$. {This residual is low-dimensional but contains useful (sometimes even complete) information to characterize the effect of $X$ on $Y$ conditional on $Z$.} In Algorithm \ref{alg:transform_maxway}, we present this ``transformed" Maxway CRT. For a transformation $R$, let $\r = (R(X_1, Z_{1\cdot}), \dots, R(X_n, Z_{n\cdot}))\trans$, $\rhos_r(\cdot \mid g(Z), h(Z))$ be the conditional distribution of $R(X,Z)$ given $g(Z)$ and $h(Z)$, and 
$\rho_r^{\star n}(\r \mid g(\Z), h(\Z)) = \prod_{i = 1}^n \rhos_r(R(X_i, Z_{i\cdot}) \mid g(Z_{i \cdot}), h(Z_{i \cdot}))$. Let $\rho_r$ be an estimate of $\rhos_r$ and $\rho_r^{n}(\r \mid g(\Z), h(\Z)) = \prod_{i = 1}^n \rho_r(R(X_i, Z_{i\cdot}) \mid g(Z_{i \cdot}), h(Z_{i \cdot}))$. We call $\rho_r$ the \textit{transformed Maxway distribution}. 

\begin{algorithm}
\caption{The Transformed Maxway CRT}
\label{alg:transform_maxway}
{\bf Input:} a transformation function $R$; functions $h$, $g$; transformed Maxway distribution $\rho_r$, which is an estimate of $\rhos_r$; data $\bD=(\by,\bx,\bZ)$; test statistic function $T$; and number of randomizations $M$.

\vspace{3pt}

Compute $\r = (R(X_1, Z_{1\cdot}),R(X_2, Z_{2\cdot}), \dots, R(X_n, Z_{n\cdot}))\trans$. 

\textbf{For} $m=1,2,...,M$:

\qquad Sample $\r\supm$ from the (transformed) Maxway distribution $\rho_r^n (\cdot \mid g(\Z), h(\Z))$ independently of $(\bx,\by)$.

\vspace{3pt}

{\bf Output:} The transformed Maxway CRT $p$-value 
\begin{equation}
\label{eqn:transform_maxway_CRT}
p_{\operatorname{t-maxway}}(\bD)= \frac{1}{M+1}\p{1+\sum_{m=1}^M \One{T(\y,\r\supm,g(\Z), h(\Z))\geq T(\y,\r, g(\Z), h(\Z) )}}. 
\end{equation}
\end{algorithm}

This transformation step is not only useful for allowing the existence of a low-dimensional $h$, but also helpful in reducing the dimensionality of $h$, thus making the (transformed) Maxway distribution easier to learn. 
For example, in an additive model where $X = f(Z) + \eta$ and $\eta$ is independent of $Z$, we can run the Maxway CRT (Algorithm \ref{alg:maxway}), or we can run the transformed Maxway CRT (Algorithm \ref{alg:transform_maxway}) with $R(X,Z) = X - f(Z)$ and $h$ being the null set.
For the original Maxway CRT, the Maxway distribution is the conditional distribution of $X$ given $g(Z)$ and $h(Z)$; for the transformed Maxway CRT, the (transformed) Maxway distribution is the conditional distribution of $X - f(Z)$ given $g(Z)$. Typically, the latter can be easier to learn due to its lower dimensionality of the predictors. 

\section{Robustness of the Maxway CRT}\label{sec:theory:robust}

In this section, we study the robustness of the Maxway CRT. Unless mentioned otherwise, we focus on Algorithm \ref{alg:maxway}. Results for the transformed Maxway CRT (Algorithm \ref{alg:transform_maxway}) follow similarly. See Appendix \ref{section:robust_transformed} for more details. 

\subsection{Conditions to achieve exact validity}
\label{section:exact_inference}
We start with analyzing the Maxway CRT when we have perfect information on a subset of the distributions of $X \mid Z$, $Y \mid Z$, and $X \mid g(Z), h(Z)$. Our first theoretical result (Theorem \ref{theo:exact_inference}) establishes a sufficient condition for the Maxway CRT $p$-value to be exactly valid, which is indeed strictly weaker than that of the model-X CRT. 

\begin{theorem}
\label{theo:exact_inference}
Suppose that either of the following conditions holds: (i) the vectors $\x,  \x^{(1)}, \dots, \x^{(M)}$ are exchangeable conditioning on $\Z$; (ii) the vectors $\x,  \x^{(1)}, \dots, \x^{(M)}$ are exchangeable conditioning on $\{g(\Z),h(\Z)\}$, and $\Z \indp \y \mid g(\Z)$. Then the Maxway CRT $p$-value defined in \eqref{eqn:maxway_CRT} is valid, i.e., $\PP{p_{\operatorname{maxway}}(\bD) \leq \alpha} \leq \alpha$ for any $\alpha \in [0,1]$ under the null hypothesis \eqref{eqn:indep_test}. 
\end{theorem}

Compared to the conditions required for the model-X CRT to be exactly valid, those stated in Theorem \ref{theo:exact_inference} are strictly weaker. More specifically, condition (i) is necessary for the model-X CRT to be valid, while the Maxway CRT is valid when either (i) or (ii) holds. 
In words, condition (i) requires that the knowledge of the distribution of $X$ given $Z$ is perfect, and condition (ii) requires that $g(Z)$ contains all the information about $Y$ that $Z$ can possibly provide and that the distribution of $X$ given the low-dimensional $g(Z)$ and $h(Z)$ is known. 
Therefore, 
when knowledge of the distribution of $Y$ given $Z$, and knowledge of the distribution of $X$ given low-dimensional objects is available, the Maxway CRT can be more robust to model misspecification of the distribution of $X$ given $Z$ than the model-X CRT. 

We also note that even though the problem description appears to be symmetric in $X$ and $Y$, $X$ and $Y$ do not play the same role in Theorem \ref{theo:exact_inference}. On the one hand, the steps in the Maxway CRT algorithm are asymmetric in $X$ and $Y$. Specifically, in the random sample generation step in Algorithm \ref{alg:maxway}, the algorithm aims at generating independent samples that are exchangeable with $\x$ instead of $\y$. On the other hand, the algorithm requires more information on $X$ than on $Y$; we need to sample from the distribution of $X$, while, for Y, we need only a sufficient statistic $g(Z)$.

\subsection{Bound on type-I error inflation}
\label{sec:theory:robust:1}

We then move on to study the behavior of the Maxway CRT when none of $g$, $h$ and $\rho$ is guaranteed to be perfect anymore. We establish that when a good estimate of the distribution of $X \mid g(Z), h(Z)$ is available, either a good estimate of the distribution of $X \mid Z$ or a good estimate of the distribution of $Y \mid Z$ will suffice for the approximate validity of the Maxway CRT. 

{\darkred

\begin{prop}[Type-I error bound: arbitrary test statistic]
\label{theo:main0}
Let $\x', \y', \tilde{\x}$ and $\tilde{\y}$ be four random vectors sampled independently conditioning on $\Z$ from the following distributions respectively: 
\begin{equation}
\begin{split}
    &\x' \sim f_{\x \mid \Z}(\cdot \mid \Z), \quad 
    \tilde{\x} \sim \rho^n(\cdot \mid  g(\Z), h(\Z)), \\
    &\y' \sim f_{\y \mid \Z}(\cdot \mid \Z), \quad 
    \tilde{\y} \sim f_{\y \mid g(\Z), h(\Z)}(\cdot \mid  g(\Z), h(\Z)). 
\end{split}
\end{equation}
Under the null hypothesis \eqref{eqn:indep_test}, for any $\alpha \in (0,1)$,  
\begin{equation}
\label{eqn:mayway_bound0}
\begin{split}
&\PP{p_{\operatorname{maxway}}(\bD) \leq \alpha \mid g(\Z), h(\Z)} \\
& \qquad \qquad   \qquad \qquad  \leq \alpha + d_{\operatorname{TV}} \Big(p( \cdot \mid g(\Z), h(\Z)), q( \cdot \mid g(\Z), h(\Z)) \Big) ,
\end{split}
\end{equation} 
where
\begin{equation}
    \begin{split}
        p( \cdot \mid g(\Z), h(\Z)) &= f_{\x', \y' \mid g(\Z), h(\Z)}( \cdot \mid g(\Z), h(\Z)),\\
        q( \cdot \mid g(\Z), h(\Z)) &= f_{\tilde{\x}, \tilde{\y} \mid g(\Z), h(\Z)}( \cdot \mid g(\Z), h(\Z)).
    \end{split}
\end{equation}
Here, recall that $\rho^{n}$ is the distribution from which $\x\supm$ is sampled in Algorithm \ref{alg:maxway}. 
\end{prop}

Our first result demonstrates that the type-I error inflation can be upper bounded by the total variation distance of two distributions $p( \cdot \mid g(\Z), h(\Z))$ and $q( \cdot \mid g(\Z), h(\Z))$. This result is a direct adaptation of Theorem 5 in \citet{berrett2018conditional} to the Maxway CRT.} We further establish an interesting upper bound of this total variation distance in Theorem~\ref{theo:almost_double_robust}. 

\begin{theorem}[Type-I error bound: almost double robustness]
Under the null hypothesis \eqref{eqn:indep_test}, for any $\alpha \in (0,1)$, 
\begin{equation}
\label{eqn:mayway_bound1}
\PP{p_{\operatorname{maxway}}(\bD) \leq \alpha \mid g(\Z), h(\Z)} 
 \leq \alpha + 2\Delta_x\Delta_y + \Delta_{x|g,h},
\end{equation} 
where 
\begin{equation}
    \begin{split}
        \Delta_x  &= d_{\operatorname{TV}} \Big(f_{\x \mid \Z}( \cdot \mid \Z), \rho^{\star n}( \cdot \mid g(\Z), h(\Z))\Big),\\
        \Delta_y &= d_{\operatorname{TV}}\p{f_{\y \mid \Z} ( \cdot \mid \Z), f_{\y \mid g(\Z), h(\Z)} ( \cdot \mid g(\Z), h(\Z))},\\
        \Delta_{x|g,h} &= d_{\rm TV} \p{\rho^{\star n}( \cdot \mid g(\Z) , h(\Z)),\rho^n( \cdot \mid g(\Z) , h(\Z))}.
    \end{split}
\end{equation}
Here, recall that $\rho^{\star n}$ is the distribution of $\x \mid g(\Z), h(\Z)$, and $\rho^{n}$, as an estimate $\rho^{\star n}$, is the distribution from which $\x\supm$ is sampled in Algorithm \ref{alg:maxway}. 
\label{theo:almost_double_robust} 
\end{theorem}

\begin{remark}
We prove Theorem \ref{theo:almost_double_robust} by showing that the total variation distance in \eqref{eqn:mayway_bound0} is upper bounded by $2 \Delta_x \Delta_y + \Delta_{x|g,h}$. 
\end{remark}

The $\Delta_x$ term in \eqref{eqn:mayway_bound1} captures the degree to which $X$ is independent of $Z$ conditional on $h(Z)$. When $X \indp Z \mid h(Z)$, the conditional distribution of $X \mid Z$ will be the same as $X \mid h(Z)$, and thus $\Delta_x=0$; when $h$ is not perfect, the more accurately $h$ characterizes the distribution of $X\mid Z$, the smaller $\Delta_x$ is.  Similarly, the $\Delta_y$ term quantifies the quality of information we have about the distribution of $Y\mid Z$, and becomes zero when $Y \indp Z \mid g(Z)$.
The $\Delta_{x|g,h}$ term quantifies the accuracy of $\rho$, i.e., how accurately we can estimate the distribution of $X$ given the low-dimensional objects $g(Z)$ and $h(Z)$. The bound in \eqref{eqn:mayway_bound1} can be interpreted as an ``almost doubly robust" property of the Maxway CRT because the term $\Delta_x \Delta_y$ will be small if either $\Delta_x$ or $\Delta_y$ is small. The word ``almost" is to emphasize that we have another $\Delta_{x|g,h}$ term in the error bound, although it is typically much smaller than either $\Delta_x$ or $\Delta_y$ because it concerns a much easier task: regressing $X$ against the low-dimensional $g(Z)$ and $h(Z)$. 

We also note that the results provided by Proposition \ref{theo:main0} and Theorem \ref{theo:almost_double_robust} do not depend on the form of the test statistic $T$, as long as $T$ is a function of $\y,\x, g(\Z)$, and $h(\Z)$. The advantage of this is the full flexibility in choosing any complicated statistic (see, for instance, the Gini index statistic computed with random forest introduced in Section \ref{sec:const:semi}). Nevertheless, bound \eqref{eqn:mayway_bound0} and \eqref{eqn:mayway_bound1} may not be sharp given a fixed statistic. In Section \ref{section:specific_stats}, we derive sharper bounds for some specific statistics commonly used in practice. 

Compared with the model-X CRT, Theorem \ref{theo:almost_double_robust} indicates that the type-I error inflation of the Maxway CRT is generally smaller. To see this, we note that for the model-X CRT, \citet{berrett2018conditional} established an upper bound for the type-I error inflation:
\begin{equation}
\label{eqn:original_CRT_bound}
\PP{p_{\operatorname{mx}}(\bD) \leq \alpha} 
 \leq \alpha + \EE{d_{\operatorname{TV}}\p{f_{\x \mid \Z}( \cdot \mid \Z) , \widetilde{\rho}^n_{\x \mid \Z}( \cdot \mid \Z) }},
\end{equation}
where $\widetilde{\rho}^n$ is the distribution from which $\x\supm$ is sampled, which is typically taken as an estimator of $f_{\x \mid h(\Z)}(\cdot \mid \Z)$. 
They also showed that there exists a matching lower bound for some test statistic when the number of randomizations $M \to \infty$. 
In order to compare \eqref{eqn:mayway_bound1} and \eqref{eqn:original_CRT_bound} side by side, 
we further bound $\PP{p_{\operatorname{mx}}(\bD) \leq \alpha} $ by
\begin{equation}
\label{eqn:original_CRT_bound_var}
\begin{split}
\PP{p_{\operatorname{mx}}(\bD) \leq \alpha} 
 &\leq \alpha + \EE{d_{\operatorname{TV}}\p{f_{\x \mid \Z}( \cdot \mid \Z) , f_{\x \mid h(\Z)}( \cdot \mid h(\Z)) }} +\\
& \qquad \qquad \qquad \qquad  \EE{d_{\operatorname{TV}}\p{f_{\x \mid h(\Z)}( \cdot \mid h(\Z)) , \widetilde{\rho}^n_{\x \mid \Z}( \cdot \mid \Z) }}. 
\end{split}
\end{equation}
For notation simplicity, let 
\begin{equation}
\label{eqn:d_x_prime}
\Delta_x' =  d_{\operatorname{TV}}\p{f_{\x \mid \Z}( \cdot \mid \Z) , f_{\x \mid h(\Z)}( \cdot \mid h(\Z)) };\quad \Delta_{x|g,h}' =  d_{\operatorname{TV}}\p{f_{\x \mid h(\Z)}( \cdot \mid h(\Z)) , \widetilde{\rho}^n_{\x \mid \Z}( \cdot \mid \Z) }.
\end{equation}
Similar to $\Delta_x$ and $\Delta_{x|g,h}$, $\Delta_{x}'$ quantifies how good the X-modeling is, while $\Delta_{x|g,h}'$ captures how accurately we can estimate the distribution of $X$ given a low-dimensional object. Therefore, we can normally assume that $\Delta_x \approx \Delta_x'$ and $\Delta_{x|g,h} \approx \Delta_{x|g,h}'$. With the above notations, the type-I error inflation bound of the Maxway CRT can be written as $\EE{2 \Delta_x \Delta_y + \Delta_{x|g,h}}$, while that of the model-X CRT can be written as $\EE{\Delta_x' + \Delta_{x|g,h}'}$. Hence, if $\Delta_y$ is small, then 
\begin{equation}
\EE{2 \Delta_x \Delta_y + \Delta_{x|g,h}} < \EE{\Delta_x +\Delta_{x|g,h}} \approx \EE{\Delta_x' + \Delta_{x|g,h}'}. 
\end{equation} 
We will make the above arguments precise in Section \ref{sec:sa:ssl} Convergence rate example \ref{exam:SIM_theory} and Section \ref{subsection:example} Convergence rate example \ref{exam:gauss_variable_selection}, where we study specific examples and present the convergence rates of the bounds in \eqref{eqn:mayway_bound1} and \eqref{eqn:original_CRT_bound}.

{\darkred
\subsection{A lower bound}

In this section, we establish a lower bound on the type-I error inflation. Proposition \ref{theo:lower_bound} shows that if the number of randomization $M$ is large, and the test statistic $T$ and threshold $\alpha$ are chosen adversarially, then the type-I error inflation can be lower bounded (up to a vanishing error) by the total variation distance between the distributions of $\x, \y \mid g(\Z), h(\Z)$ and  $\x^{(m)}, \y \mid g(\Z), h(\Z)$. Since the lower bound coincides with the upper bound in Proposition \ref{theo:main0}, this result also implies that the bound in Proposition \ref{theo:main0} is tight. 
\begin{prop}[Lower bound]
\label{theo:lower_bound}
Let $\x', \y', \tilde{\x}$ and $\tilde{\y}$ be four random vectors sampled independently conditioning on $\Z$ from the following distributions respectively: 
\begin{equation}
\begin{split}
    &\x' \sim f_{\x \mid \Z}(\cdot \mid \Z), \quad 
    \tilde{\x} \sim \rho^n(\cdot \mid  g(\Z), h(\Z)), \\
    &\y' \sim f_{\y \mid \Z}(\cdot \mid \Z), \quad 
    \tilde{\y} \sim f_{\y \mid g(\Z), h(\Z)}(\cdot \mid  g(\Z), h(\Z)). 
\end{split}
\end{equation}
Under the null hypothesis \eqref{eqn:indep_test}, there exists a statistic
$T$ such that
\begin{equation}
\begin{split}
&\EE{\sup _{\alpha \in[0,1]}\p{ \PP{p_{\operatorname{maxway}}(\bD) \leq \alpha \mid \y, g(\Z), h(\Z)} -\alpha} \mid g(\Z), h(\Z)} \\
& \qquad \qquad \qquad  \geq d_{\operatorname{TV}} \Big(p( \cdot \mid g(\Z), h(\Z)), q( \cdot \mid g(\Z), h(\Z)) \Big)   -0.5(1+o(1)) \sqrt{\frac{\log (M)}{M}}, 
\end{split}
\end{equation}
as $M \rightarrow \infty$,
where
\begin{equation}
    \begin{split}
        p( \cdot \mid g(\Z), h(\Z)) &= f_{\x', \y' \mid g(\Z), h(\Z)}( \cdot \mid g(\Z), h(\Z)),\\
        q( \cdot \mid g(\Z), h(\Z)) &= f_{\tilde{\x}, \tilde{\y} \mid g(\Z), h(\Z)}( \cdot \mid g(\Z), h(\Z)).
    \end{split}
\end{equation}
\end{prop}

Finally, we make a remark on the statistic $T$ that achieves the lower bound \eqref{theo:lower_bound}. The specific construction of this statistic can be found in Appendix \ref{appendix:lower_bound_proof}. At a high level, the statistic $T$ aims at distinguishing the two distributions $p$ and $q$, where $p$ is the distribution of $\x', \y' \mid g(\Z), h(\Z)$ and $q$ is the distribution of $\tilde{\x}, \tilde{\y} \mid g(\Z), h(\Z)$. 
We know that $p$ and $q$ are the same if $\rho = \rhos$ and $g(\Z)$ and $h(\Z)$ are sufficient for characterizing the distribution of $\x \mid \Z$ and $\y \mid \Z$; this is because when the above conditions hold, $\x' \indp \y' \mid g(\Z), h(\Z)$ under the null hypothesis. Therefore, the lower bound achieving statistic $T$ focuses on detecting the extent to which $\rho$ differs from $\rhos$ and how well $g(\Z)$ and $h(\Z)$ capture the distribution of $\x \mid \Z$ and $\y \mid \Z$. This choice of statistic, even though allowed for the implementation of the CRT, is not a typical choice in practice. Instead, commonly used statistics aim to determine whether $\x$ and $\y$ are independent given $\Z$. 
}

\subsection{Tighter bounds for specific test statistics}
\label{section:specific_stats}

As discussed above, bound \eqref{eqn:mayway_bound0} and \eqref{eqn:mayway_bound1} do not depend on the test statistic $T$. The bounds are concise, useful for analyzing complicated statistics, and interesting conceptually; yet the bound can be loose for some specific statistics. In particular, the statistic that achieves the matching lower bound of \eqref{eqn:mayway_bound0} does not aim at learning whether $\x$ and $\y$ are independent conditioning on $\Z$, and thus it is not a typical choice of statistics in practice. 

In this section, we study with some commonly used test statistics, whether we can establish sharper bounds on the type-I error inflation of the Maxway CRT. To simplify our analysis, we make the following assumption.
\begin{assumption}
\label{assu:mean_model}
Under the null hypothesis, $Y = \smuy(Z) + \eta$ and $X = \smux(Z) + \varepsilon$, where $\eta$ and $\epsilon$ are mean-zero random variables independent of $Z$ and independent of each other, and $\varepsilon \sim \mathcal{N}(0,1)$. 
\end{assumption} 

Let $\mu_x(Z)$ be an estimator of $\smux(Z)$ and  $\mu_y(Z)$ an estimator of $\smuy(Z)$, such that $\mu_x(Z)$ and $\mu_y(Z)$ are measurable with respect to $g(Z)$ and $h(Z)$ used in Algorithm \ref{alg:maxway}. For example, in a linear model, if $g(Z) = Z_{\mathcal{S}}$ is a subset of $Z$, then one possible choice of $\mu_x(Z)$ is to set $\mu_x(Z) = Z_{\mathcal{S}} \trans \beta$ for some vector $\beta$. 
We study two statistics, the inner-product statistic:
\begin{equation}
\label{eqn:corr_stats}
T(\by,\x,g(\Z),h(\Z)) = \abs{\x \trans \y},
\end{equation}
and the d$_0$ statistic:
\begin{equation}
\label{eqn:d_0_def_bound}
T(\by,\x,g(\Z),h(\Z)) = \abs{\p{\x -  \mu_x(\Z)}\trans \p{\y - \mu_y(\Z)}}. 
\end{equation}
The d$_0$ statistic is proposed by \citet{liu2020fast} as a fast, powerful, and intuitive statistic; see more description in Section \ref{sec:const:semi}. 
The inner-product statistic can be viewed as a simpler version of the d$_0$ statistic. 

\subsubsection{Analysis of the inner-product statistic}
 
Let $\tmuy(Z) = \EE{Y \mid g(Z), h(Z)} = \EE{\smuy(Z) \mid g(Z), h(Z)}$ and $\tmux(Z) = \EE{X \mid g(Z), h(Z)} = \EE{\smux(Z) \mid g(Z), h(Z)}$. Theorem \ref{theo:inner-product_stats} establishes a bound on the type-I error of the Maxway CRT with the inner-product statistic. 

\begin{theorem}[Type-I error bound: inner-product statistic]
\label{theo:inner-product_stats}
Under Assumption \ref{assu:mean_model}, assume further that the Maxway CRT (Algorithm \ref{alg:maxway}) samples $\x\supm$ from a normal distribution, i.e., $\rho(\cdot \mid g(Z), h(Z))$ corresponds to $\mathcal{N}(\mu_x(Z), 1)$.
{\darkred Assume that $\eta$ is a continuous random variable whose density is upper bounded by a constant $C_1$.}
Assume further that there exists a positive constant $C_2$ such that $\EE{\eta^2} \leq C_2$. 
Then there exists a positive constant $C$ such that for any $\alpha \in (0,1)$, the type-I error of the Maxway CRT using the inner-product statistic defined in \eqref{eqn:corr_stats} can be bounded by 
\begin{equation}
\label{eqn:mayway_bound_inner}
\PP{p_{\operatorname{maxway}}(\bD) \leq \alpha}  \leq \alpha  + 
C  \p{\sqrt{n}\Delta_{\rho, \operatorname{mean}}  + \sqrt{n}\Delta_{x, \operatorname{mean}} \Delta_{y, \operatorname{mean}} + (\Delta_{x, \operatorname{mean}} + \tilde{\Delta}_{x, \operatorname{mean}})},
\end{equation} 
where  
\begin{equation}
    \begin{split}
         &\Delta_{x, \operatorname{mean}} = \sqrt{\EE{(\smux(Z) - \tmux(Z))^2}} , \qquad 
        \tilde{\Delta}_{x, \operatorname{mean}} = \sqrt{\EE{((\smux(Z) - \tmux(Z)\tmuy(Z))^2}}\\
        &\Delta_{y, \operatorname{mean}} = \sqrt{\EE{(\smuy(Z) - \tmuy(Z))^2}}, \qquad 
        \Delta_{\rho, \operatorname{mean}} = \sqrt{\EE{(\tmux(Z) - \mu_x(Z))^2}}. 
    \end{split}
\end{equation}
\end{theorem}

Theorem \ref{theo:inner-product_stats} can again be interpreted as an ``almost double robustness" of the Maxway CRT. The terms $\Delta_{x, \operatorname{mean}}$ and $\tilde{\Delta}_{x, \operatorname{mean}}$ quantify the difference between $\EE{X \mid Z}$ and $\EE{X \mid g(Z), h(Z)}$, $\Delta_{y, \operatorname{mean}}$ quantifies the difference between $\EE{Y \mid Z}$ and $\EE{Y \mid g(Z), h(Z)}$, and $\Delta_{\rho, \operatorname{mean}}$ quantifies the estimation error of $\EE{X \mid g(Z), h(Z)}$. When the X-modeling is of high accuracy, $\EE{X \mid Z}$ will be close to $\EE{X \mid g(Z), h(Z)}$; thus $\Delta_{x, \operatorname{mean}}$ and $\tilde{\Delta}_{x, \operatorname{mean}}$ will be small. When the Y-modeling is of high accuracy, $\EE{Y \mid Z}$ will be close to $\EE{Y \mid g(Z), h(Z)}$; thus $\Delta_{y, \operatorname{mean}}$ will be small. The term $\Delta_{\rho, \operatorname{mean}}$ is in general smaller than $\Delta_{x, \operatorname{mean}}$ and $\Delta_{y, \operatorname{mean}}$ because it concerns low-dimensional regression. Finally, in \eqref{eqn:mayway_bound_inner}, the term $\Delta_{x, \operatorname{mean}} + 
 \tilde{\Delta}_{x, \operatorname{mean}}$ is smaller than the other terms since they do not have the multiple of $\sqrt{n}$. To summarize the above discussion,  \eqref{eqn:mayway_bound_inner} can be thought of as 
\begin{equation}
\PP{p_{\operatorname{maxway}}(\bD) \leq \alpha}   \leq \alpha  + 
\textnormal{ small term } + C \p{\Delta_{\rho, \operatorname{mean}} + \sqrt{n} \Delta_{x, \operatorname{mean}} \Delta_{y, \operatorname{mean}}} . 
\end{equation} 
Using the same proof technique, we can obtain a bound for the type-I error of the model-X CRT:
\begin{equation}
\label{eqn:mx_bound_inner}
\begin{split}
 \PP{p_{\operatorname{mx}}(\bD) \leq \alpha} \leq \alpha + 
C \sqrt{n\EE{\p{ \smux(Z) - \mu_{x, \operatorname{mx}}(Z) }^2}},
 \end{split}
\end{equation}
where each $\x\supm_i$ is sampled from $\mathcal{N}(\mu_{x, \operatorname{mx}}(Z_i), 1)$ in the CRT. See Appendix \ref{subsection:proof_mx_bound_inner} for a proof of \eqref{eqn:mx_bound_inner}. Unlike the Maxway CRT, this bound for the model-X CRT does not enjoy any double robustness property. Indeed, if the model for the mean function of $X$ is misspecified, then the $\pval$ $p_{\operatorname{mx}}$ will be far from a $\operatorname{Unif}[0,1]$ under the null hypothesis. 

We then compare the two bounds on the Type-I error of the Maxway CRT obtained from Theorems \ref{theo:almost_double_robust} and \ref{theo:inner-product_stats}, and the corresponding bounds on that of the model-X CRT. 

In particular, we examine a gaussian linear example. 

{\darkred
\begin{rateexample}[Gaussian linear model]
\label{exam:gaussian_linear0}
Assume that $Z_{i \cdot} \sim \mathcal{N}(0, \Sigma_z)$, $Y_i = Z_{i \cdot} \trans \theta^{\star} + \eta_i$, and $X_i = Z_{i \cdot}\trans \beta^{\star} + \varepsilon_i$, where $\eta_i, \varepsilon_i \sim \mathcal{N} (0,1)$ are noise terms independent of $Z_{i \cdot}$. We implement the (transformed) Maxway CRT in the following way: we take $g(Z_{i \cdot}) =  Z_{i \cdot} \trans \theta$ to be an estimate of the conditional mean function of $Y_i$, we take the transformation $R(X_i, Z_{i \cdot}) = X_i - Z_{i \cdot} \trans \beta$ and we take $h(Z_{i \cdot})$ as null.

Assume that $\beta^{\star}$ can be estimated with rate $\norm{\beta - \beta^{\star}} \lesssim \Delta_{x,\lin}$, and $\theta^{\star}$ can be estimated with rate $\norm{\theta - \theta^{\star}} \lesssim \Delta_{y,\lin}$. We get $\rho_r(Z_{i \cdot})$ by running linear regression of $R(X_i, Z_{i \cdot})$ on $Z_{i \cdot}\trans\theta$ (potentially on an external dataset). Assume the linear regression has an error rate $\Delta_{\rho,\lin}$. 

In this setting, we can show that for the terms in Theorem \ref{theo:inner-product_stats},
\begin{equation}
\Delta_{x, \operatorname{mean}} \lesssim  \Delta_{x,\lin}, \quad
\Delta_{y, \operatorname{mean}} \lesssim  \Delta_{y,\lin}, \quad
\Delta_{\rho, \operatorname{mean}} \lesssim  \Delta_{\rho,\lin};
\end{equation}
and for the terms in Theorem \ref{theo:almost_double_robust},
\begin{equation}
\Delta_{x} \lesssim  \sqrt{n}\Delta_{x,\lin}, \quad
\Delta_{y} \lesssim  \sqrt{n}\Delta_{y,\lin}, \quad
\Delta_{x\mid g,h} \lesssim  \sqrt{n}\Delta_{\rho,\lin}. 
\end{equation}

A consequence of the above is the following: 
Bound \eqref{eqn:mayway_bound1} in Theorem \ref{theo:almost_double_robust} implies that with arbitrary test statistic, 
\begin{equation}
\textnormal{Type-I error inflation of the Maxway CRT }
 \lesssim \sqrt{n}\Delta_{\rho,\lin} + n\Delta_{x,\lin}\Delta_{y,\lin},
\label{equ:example_arbitrary_maxway}
\end{equation}
bound \eqref{eqn:mayway_bound_inner} in Theorem \ref{theo:inner-product_stats} implies that with the inner-product statistic, 
\begin{equation}
\textnormal{Type-I error inflation of the Maxway CRT }
 \lesssim \sqrt{n}\Delta_{\rho,\lin} + \sqrt{n}\Delta_{x,\lin}\Delta_{y,\lin},
\label{equ:example_inner_product_maxway}
\end{equation}
and bound \eqref{eqn:original_CRT_bound}/bound \eqref{eqn:mx_bound_inner} implies that
\begin{equation}
\textnormal{Type-I error inflation of the model-X CRT }
 \lesssim \sqrt{n} \Delta_{\rho,\lin} + \sqrt{n}\Delta_{x,\lin}. 
\label{equ:example:_inner_product_CRT}
\end{equation}
We note that the second term of \eqref{equ:example_inner_product_maxway} is smaller than that of \eqref{equ:example_arbitrary_maxway} by a factor of $\sqrt{n}$. This difference comes from focusing on the specific test statistic. 
Comparing bound \eqref{equ:example_inner_product_maxway} with \eqref{equ:example:_inner_product_CRT}, we note that the type-I error inflation of the Maxway CRT is smaller than that of the CRT when
$\Delta_{\rho,\lin} \lesssim \Delta_{x,\lin}$, and $\Delta_{y,\lin} \lesssim 1$. 
We will provide more specific values of the bounds in Section \ref{sec:in_practice} (see Convergence rate example \ref{exam:SIM_theory}), where we
discuss the estimation of $\thetas$ and $\betas$ with semi-supervised learning and surrogate-assisted semi-supervised learning (or transfer learning). We provide details of this example in Appendix \ref{subsection:detail_example_gauss_linear} and additional examples in Appendix \ref{subsection:example}. 
\end{rateexample}
}

\subsubsection{Analysis of the \texorpdfstring{d$_0$}{d0} statistic}

Finally, we turn to the d$_0$ statistic. Theorem \ref{theo:d0_stats} establishes a bound on the type-I error inflation of the Maxway CRT with the d$_0$ statistic. The subsequent discussions following Theorem \ref{theo:inner-product_stats} can also be made for Theorem \ref{theo:d0_stats}. We omit them here for conciseness.  

\begin{theorem}[Type-I error bound: d$_0$ statistic]
\label{theo:d0_stats}
Under the conditions of Theorem \ref{theo:inner-product_stats}, there exists a positive constant $C$ such that for any $\alpha \in (0,1)$, the type-I error of the Maxway CRT using the d$_0$ statistic defined in \eqref{eqn:d_0_def_bound} can be bounded by 
\begin{equation}
\label{eqn:mayway_bound_d0}
\PP{p_{\operatorname{maxway}}(\bD) \leq \alpha}  \leq \alpha  + 
C  \p{\sqrt{n}\Delta_{\rho, \operatorname{mean}}  + \sqrt{n}\Delta_{x, \operatorname{mean}} \Delta_{y, \operatorname{mean}} + (\Delta_{x, \operatorname{mean}} +  \check{\Delta}_{x, \operatorname{mean}})},
\end{equation} 
where  
\begin{equation}
    \begin{split}
         &\Delta_{x, \operatorname{mean}} = \sqrt{\EE{(\smux(Z) - \tmux(Z))^2}} , \quad 
        \check{\Delta}_{x, \operatorname{mean}} = \sqrt{\EE{(\smux(Z) - \tmux(Z))^2(\tmuy(Z) - \mu_y(Z))^2}}\\
        &\Delta_{y, \operatorname{mean}} = \sqrt{\EE{(\smuy(Z) - \tmuy(Z))^2}}, \quad 
        \Delta_{\rho, \operatorname{mean}} = \sqrt{\EE{(\tmux(Z) - \mu_x(Z))^2}}
    \end{split}
\end{equation}
\end{theorem}

{\darkred
Theorem \ref{theo:d0_stats} establishes that the Maxway CRT with the d$_0$ statistic achieves almost double robustness. However, recent research by \cite{niu2022reconciling} shows that the model-X CRT with the d$_0$ statistic is already doubly robust under weak assumptions. Therefore, in terms of the rate of type-I error inflation, using the Maxway CRT framework does not seem to provide further improvement in robustness when using the d$_0$ statistic. However, empirically, as we will show in Section~\ref{sec:sim}, combining the Maxway CRT framework with the d$_0$ statistic can lead to even greater robustness.
}

\section{Construction in practice}
\label{sec:in_practice}

We now consider the implementation of Algorithm \ref{alg:maxway} in data-driven studies. We propose learning strategies for the functions $g$, $h$ and the distribution $\rho$ in three practical scenarios: typical semi-supervised learning (SSL) scenario, where $h$ is learned on an external dataset with observations of $(X,Z)$, a surrogate-assisted SSL (SA-SSL) scenario where $g$ can be learned from an additional surrogate variable $S$ for the unobserved outcome $Y$ in this external dataset, and a transfer learning (TL) scenario where $g$ can be learned leveraging some external source data with the same set of response and covariates.

\subsection{Semi-supervised learning scenario}\label{sec:const:semi}

Consider a semi-supervised learning (SSL) scenario with labeled data $\bD=(\by,\bx,\bZ)$ of sample size $n$ and unlabeled data $\bD^u=(\bx^u,\bZ^u)$ where $\bx^u=(X_1^u,X_2^u,\ldots,X^u_N)\trans\in\mathbb{R}^N$, $\Z^u=(Z_{1\cdot}^u,Z_{2\cdot}^u,\ldots,Z_{N\cdot}^u)\trans\in\mathbb{R}^{N\times p}$, and $N$ represents its sample size typically much larger than $n$. In this scenario, we will implement certain pre-specified learning algorithms on $\bD^u$ to obtain $h(\Z)$ and $\rho^n (\cdot \mid g(\Z), h(\Z))$ and use $\bD=(\by,\bx,\bZ)$ for the CRT. For $g(\Z)$, we consider two cases: (SS.I) to learn $g(\Z)$ with some ``holdout'' training samples $\bD^h=(\by^h,\bZ^h)$ independent from $\bD=(\by,\bx,\bZ)$; (SS.II) to learn $g(\Z)$ simply using $\bD=(\by,\bx,\bZ)$, as presented in Algorithms \ref{alg:outsample:cart} and \ref{alg:insample:cart} respectively.

\begin{algorithm}[H]
\caption{\label{alg:outsample:cart} The holdout training Maxway (Maxway$_{\rm out}$) CRT.}
{\bf Input:} learning algorithms $\Lsc_g$, $\Lsc_h$, and $\Lsc_{\rho}$, as well as the transformation function $R$ when using the transformed Maxway CRT (Algorithm \ref{alg:transform_maxway}); a test statistic function $T$; unlabeled data $\bD^u=(\bx^u,\bZ^u)$; labeled holdout training data $\bD^h=(\by^h,\bZ^h)$; labeled testing data $\bD=(\by,\bx,\bZ)$; and number of randomizations $M$. 

\vspace{3pt}

Implement the learning algorithms to obtain: $g(\Z)=\Lsc_g(\bD^h;\Z)$ and $h(\Z)=\Lsc_h(\bD^u;\Z)$. Also transform $\x$ into $\r=R(\x,\Z)$ when following Algorithm \ref{alg:transform_maxway}.

\vspace{3pt}

Adjust the conditional model of $\x$ or its transformation $\r$: 
\[
\rho^n(\cdot \mid g(\Z), h(\Z))\mbox{ or }\rho_r^n(\cdot \mid g(\Z), h(\Z))=\Lsc_{\rho}(\bD^u;g,h,\Z).
\]

\vspace{3pt}

Conduct the CRT resampling $M$ times from $\rho^n(\cdot \mid g(\Z), h(\Z))$ or $\rho_r^n(\cdot \mid g(\Z), h(\Z))$ as presented in Algorithm \ref{alg:maxway} or Algorithm \ref{alg:transform_maxway}.

\end{algorithm}

\begin{algorithm}[H]
\caption{\label{alg:insample:cart} The in-sample training Maxway (Maxway$_{\rm in}$) CRT.}
{\bf Input:} the same as Algorithm \ref{alg:outsample:cart} except that we no longer has holdout training data $\bD^h$.

\vspace{3pt}

Implementation procedures are the same as those in Algorithm \ref{alg:outsample:cart} except that we obtain $g(\Z)$ through $g(\Z)=\Lsc_g((\by,\bZ);\Z)$.

\end{algorithm}

Theoretical analysis in Section \ref{sec:theory:robust} applies to the holdout training version, i.e. the Maxway$_{\rm out}$ CRT presented in Algorithm \ref{alg:outsample:cart}. However, under a typical SSL scenario, it requires splitting the whole labeled dataset into training and testing sets similar to \cite{tansey2018holdout}. This can essentially impact the power in finite sample studies as shown in \cite{liu2020fast}. In contrast, the in-sample version Maxway$_{\rm in}$ fully utilizes the labeled samples for constructing the CRT and thus can be typically more powerful than Maxway$_{\rm out}$ when using the same number of labels. Nevertheless, the robustness of Maxway$_{\rm in}$ may be impacted since the estimated $g$ is not independent of the testing data $\D$, which can be viewed as an over-fitting issue conceding the theoretical guarantee on the robustness of Maxway$_{\rm in}$ provided in Theorem \ref{theo:almost_double_robust}. 

In our numerical experiments, we further study the robustness of Maxway$_{\rm in}$ and compare it with Maxway$_{\rm out}$. Interestingly, we found that Maxway$_{\rm in}$ is not necessarily less robust than Maxway$_{\rm out}$. In specific, when using learning algorithms with strong shrinkage or regularization like lasso, Maxway$_{\rm in}$ actually shows better type-I error control than Maxway$_{\rm out}$; see more details in Section \ref{sec:sim}. We now present two examples for choices of the learning algorithms and test statistic functions used in Algorithms \ref{alg:outsample:cart}. For Algorithms \ref{alg:insample:cart}, one just needs to simply replace the holdout training data $\bD^h$ with $(\by,\bZ)$ in these two examples.

\begin{CARTexample}[Lasso]
Let $\widehat{\gamma}_{yz}$ be the fitted lasso coefficients for the generalized linear model (GLM) of $\by^h$ against $\Z^h$ and take $g(\Z)=\Lsc_g(\bD^h;\Z)=(\bZ\widehat{\gamma}_{yz},\bZ_{\bullet,\mathrm{top}(k)})$ with $\mathrm{top}(k)$ representing the indices of the $k$ largest entries in $|\widehat{\gamma}_{yz}|$. Let $\widehat{\gamma}_{xz}$ be the fitted lasso coefficients for the GLM of $\x^u\sim\Z^u$. Consider two different scenarios:

\begin{enumerate}
    \item[(i)] When $X\mid Z$ is assumed to be a gaussian linear model, transform $\x$ and $\x^u$ through 
    \[
    \{\r,\r^u\}=\{R(\x,\Z),R(\x^u,\Z^u)\}=\{\x-\bZ\widehat{\gamma}_{xz},\x^u-\bZ^u\widehat{\gamma}_{xz}\},
    \]
    and set $h(\Z)=\Lsc_h(\bD^u;\Z)$ as null. Then fit a linear regression for $\r^u\sim g(\Z^u)$ to estimate $\rho^n_r(\cdot \mid g(\Z), h(\Z))$ as the gaussian linear distribution of $\r\mid g(\Z)$ with the estimated mean denoted as $\EEhat{R\mid g(\Z)}$ and the variance estimated by $\|\r-\EEhat{R\mid g(\Z)}\|_2^2/n$.
    
    \item[(ii)] When $X$ is binary and $X\mid Z$ is assumed to be a logistic model, do not transform $\x$ but take $h(\Z)=\Lsc_h(\bD^u;\Z)=\bZ\widehat{\gamma}_{xz}$ and fit a logistic regression for $\x^u$ against $g(\Z^u)$ and $h(\Z^u)$, to estimate the Maxway distribution $\rho^n(\cdot \mid g(\Z), h(\Z))$ with the estimated mean denoted as $\EEhat{X\mid g(\Z),h(\Z)}$.
\end{enumerate}

\noindent Let $\EEhat{Y\mid\Z}$ denote the predictor for $\by$ determined from $\Z$ and $\widehat\bep_y=\by-\EEhat{Y\mid\Z}$. For gaussian $X$, let $\widehat\bep_x=\r-\EEhat{R\mid g(\Z)}$ while for binary $X$, let $\widehat\bep_x=\x-\EEhat{X\mid g(\Z),h(\Z)}$. Inspired by \cite{liu2020fast}, we introduce two choices on the test statistic function $T$ as follows.

\begin{enumerate}
    \item[(1)] {\rm Main effect (d$_0$ statistic):} $T(\by,\x,g(\Z),h(\Z))\mbox{ or }T(\by,\r,g(\Z),h(\Z)) =|\widehat\bep_y\trans\widehat\bep_x|$.
    
    \item[(2)] {\rm Interaction effect (d$_{\mathrm{I}}$ statistic):}
    \[
    T(\by,\x,g(\Z),h(\Z))\mbox{ or }T(\by,\r,g(\Z),h(\Z)) = \widehat{\beta}_{x,1}^2+k^{-1}\sum_{j=2}^{k+1}\widehat{\beta}_{x,j}^2,
    \]
    where $\widehat{\beta}_x\in\mathbb{R}^{k+1}$ are the least-squares coefficients of $\widehat\bep_y$ against $(\widehat\bep_x,\widehat\bep_x\odot\bZ_{\bullet,\mathrm{top}(k)})$.
\end{enumerate}

\label{example:1}
\end{CARTexample}

\begin{CARTexample}[Random forest]
Fit a random forest (RF) model for $\by^h\sim\Z^h$ and take $g(\Z)=\Lsc_g(\bD^h;\Z)=(\EEhat{Y\mid\Z},\bZ_{\bullet,\mathrm{top}(k)})$ where $\EEhat{Y\mid\Z}$ represents the prediction for $\by$ and $\mathrm{top}(k)$ is the indices of the $k$ largest Gini indices fitted by RF (a common measure of variable importance in RF, see \url{https://cran.r-project.org/web/packages/randomForest} for more details). Also, fit an RF model for $\x^u\sim\Z^u$ to obtain $\EEhat{X\mid\Z}$ as the prediction of $\x$. Again, consider two scenarios:
\begin{enumerate}
    \item[(i)] When $X\mid Z$ is assumed to be gaussian, transform $\x$ and $\x^u$ through 
    \[
    \{\r,\r^u\}=\{R(\x,\Z),R(\x^u,\Z^u)\}=\{\x-\EEhat{X\mid\Z},\x^u-\EEhat{X\mid\Z^u}\},
    \]
    and set $h(\Z)=\Lsc_h(\bD^u;\Z)$ as null. Then fit an RF for $\r^u\sim g(\Z^u)$ to learn the Maxway distribution $\rho^n_r(\cdot \mid g(\Z), h(\Z))$ with the variance again estimated by $\|\r-\EEhat{R\mid g(\Z)}\|_2^2/n$.
    
    \item[(ii)] When $X$ is binary, do not transform $\x$ but take $h(\Z)=\EEhat{X\mid\Z}$ and fit an RF for $\x^u$ against $\{g(\Z^u),h(\Z^u)\}$ to learn the Maxway distribution $\rho^n(\cdot \mid g(\Z), h(\Z))$.
\end{enumerate}
Finally, fit an RF for $\widehat\bep_y\sim(\widehat\bep_x,\bZ_{\bullet,\mathrm{top}(k)})$ where $\widehat\bep_y=\by-\EEhat{Y\mid\Z}$, $\widehat\bep_x=\r-\EEhat{R\mid g(\Z)}$ for gaussian $X$ and $\widehat\bep_x=\x-\EEhat{X\mid g(\Z),h(\Z)}$ for binary $X$, and take the test statistic $T$ as the fitted Gini index of $\widehat\bep_x$.

\label{example:2}
\end{CARTexample}

\subsection{Surrogate-assisted semi-supervised learning}\label{sec:sa:ssl}
{\darkred 
As an important topic in SSL, surrogate-assisted semi-supervised learning (SA-SSL) has gained extensive interest in many application fields such as electronic health record (EHR) based data-driven biomedical studies \citep{zhang2022prior,hou2021surrogate}. In SA-SSL, in addition to the triplet $(X,Y,Z)$, there is a surrogate or silver standard label $S$ that is much more feasible and accessible than $Y$ in data collection and can be viewed as a noisy measure of $Y$.
In the literature, there are two common types of surrogates, early-endpoint surrogates \citep{prentice1989surrogate,vanderweele2013surrogate}, and post-hoc surrogates \citep{hong2019semi,zhang2020maximum, zhang2022prior}. 

\begin{figure}
    \centering
    \includegraphics[width=0.9\textwidth]{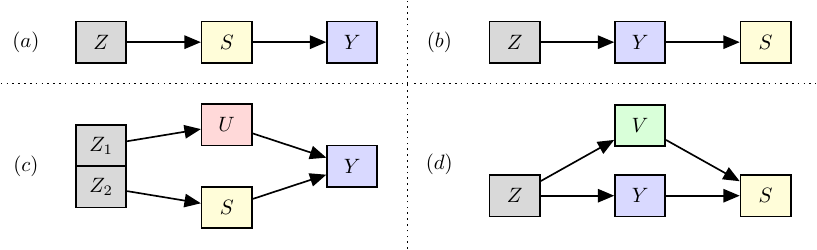}
    \caption{Directed acyclic graphs (DAGs) of four data generation examples about $(Z,Y,S)$ discussed in Section \ref{sec:sa:ssl}. (a) Data generation process for early-endpoint surrogates; (b) Data generation process for post-hoc surrogates; (c) An example where $Z\indp S\mid g(Z) \Rightarrow Z\indp Y\mid (g(Z), Z_1)$; (d) An example under which the surrogate $S$ is not valid.}
    \label{fig:diagram:surrogate}
\end{figure}

Early-endpoint surrogates, also referred to as surrogate endpoints, are measures that can be used to predict the effect of a treatment on a longer-term outcome \citep{prentice1989surrogate,vanderweele2013surrogate}. These surrogates can be biomarkers or clinical parameters that can be measured relatively quickly, usually within a few weeks or months of starting treatment. For example, in clinical trials, tumor response rate ($S$) is often used as a surrogate for overall survival ($Y$); and blood pressure ($S$) is commonly used as a surrogate for cardiovascular events such as heart attacks ($Y$). In both examples, $Z$ may represent other baseline variables collected before the clinical trial. 
Figure \ref{fig:diagram:surrogate}(a) illustrates the data generating mechanism for early-endpoint surrogates. As can be seen from the figure that the early-endpoint surrogates satisfy $Z \indp Y \mid S$.

Post-hoc surrogates, on the other hand, are measures usually taken after measurement of the true outcome \citep{hong2019semi,zhang2020maximum, zhang2022prior}. For example, in studies linking EHR with genomic data \citep{hong2019semi}, $Z$ represents biological markers observed at baseline, $Y\in\{0,1\}$ is the true status of some disease or condition of our interests associated with $Z$, and the surrogate $S$ is taken as some EHR surrogate for $Y$, e.g., count of its main diagnostic code. See Figure \ref{fig:diagram:surrogate}(b) for the data generation mechanism. We can tell from the figure that post-hoc surrogates satisfy $Z \indp S \mid Y$. 


We define that a surrogate variable $S$ is \emph{valid} with respect to a function class $\mathcal{G}$ if for any function $g \in \mathcal{G}$ of $Z$,
\begin{equation}
Z\indp S\mid g(Z) \Rightarrow Z\indp Y\mid g(Z).
\end{equation}
It is not hard to see that such a valid surrogate can be used to get useful knowledge of $Y\mid Z$ without observing $Y$. In the context of the Maxway CRT, we can then use the data of surrogates to learn the function $g$. We then provide a few examples where the surrogate variable is valid. 
Firstly, we establish in Proposition \ref{prop:surro_early} that early-endpoint surrogates (Figure \ref{fig:diagram:surrogate}(a)) are valid. Secondly, for post-hoc surrogates, we need a bit more structure for it to be valid. 
In the EHR studies introduced above, we often consider $Y$ as a binary variable. We show in Proposition \ref{prop:2} that post-hoc surrogates are valid if the outcome $Y$ is binary. 
We also show in Proposition \ref{prop:surro_linear} that if $Y \sim Z$ is a linear model and if we restrict $g$ to be a linear function of $Z$, then post-hoc surrogates are valid. The proofs of the propositions can be found in Appendix \ref{subsection:perfect_surrogate_proof}. 

\begin{prop}
Assume that the surrogate variable $S$ satisfies $Y \indp Z \mid S$. Then $S$ is a valid surrogate with respect to any function class $\mathcal{G}$. 
\label{prop:surro_early}
\end{prop}

\begin{prop}
Assume that the surrogate variable $S$ satisfies $S \indp Z \mid Y$. 
Further, assume that $Y\in\{0,1\}$, and that for some $a\in\mathbb{R}$, $\PP{S\leq a\mid Y=1}\neq \PP{S\leq a\mid Y=0}.$
Then $S$ is a valid surrogate with respect to any function class $\mathcal{G}$. 
\label{prop:2}
\end{prop}

\begin{prop}
Assume that the surrogate variable $S$ satisfies $S \indp Z \mid Y$. Assume further that $Y = Z \trans \thetas + \eta$, for some $\eta$ independent of $Z$. If $\EE{SZ} \neq 0$, then $S$ is a valid surrogate with respect to $\mathcal{G} = \cb{g: g(Z) = Z \trans a, a \in \mathbb{R}^p }$. 
\label{prop:surro_linear}
\end{prop}

More generally, when the relationship between $S$ and $Y$ becomes more complicated, $S$ may not be valid anymore. In Figure \ref{fig:diagram:surrogate}(d), we provide an example under which the surrogate $S$ is not valid. However, even in this scenario, if the data has some additional structure, then we are able to modify $g$ and get $Z \indp Y \mid g(Z)$. More specifically, in Figure \ref{fig:diagram:surrogate}(c), $Z$ can be decomposed into two independent components $Z_1$ and $Z_2$. In this case, it is straightforward to see that if $Z \indp S \mid g(Z)$, then $Z \indp Y \mid (g(Z), Z_1)$. Here, 
$S$ doesn't fully capture the relationship between $Y$ and $Z$; in particular, the part of $Z_1$ is left out. Nevertheless, if we append $g(Z)$ with the leftout information $Z_1$, then them together will be able to fully capture $Y \mid Z$.


Given the above discussions, we propose in Algorithm \ref{alg:outsample:cart:SA:SSL} the Maxway CRT approach for SA-SSL that naturally uses the large unlabeled surrogate samples to learn $g(\Z)$. On the one hand, as will be seen in Convergence rate example \ref{exam:SIM_theory}, when the surrogate $S$ is valid, our approach tends to be more robust than the SSL Maxway CRT, i.e., Algorithms \ref{alg:outsample:cart}. On the other hand, when $S$ is invalid or of poor quality, the Maxway CRT still preserves the possibility to draw valid inference due to its double robustness introduced by Theorems \ref{theo:exact_inference} and \ref{theo:almost_double_robust}, whereas the existing SA-SSL approaches \citep[e.g.,][]{hong2019semi}, which rely solely on the valid surrogate assumption, tend to fail. 
}

\begin{algorithm}[H]
\caption{\label{alg:outsample:cart:SA:SSL} The surrogate-assisted semi-supervised learning (SA-SSL) Maxway CRT.}
{\bf Input:} learning algorithms $\Lsc_g^S$, $\Lsc_h$, and $\Lsc_{\rho}$; a test statistic function $T$; unlabeled data with surrogate $\bD^{uS}=(\bs^u,\bx^u,\bZ^u)$; labeled data $\bD=(\by,\bx,\bZ)$; and number of randomizations $M$. 

\vspace{3pt}

Implementation procedures are the same as those in Algorithm \ref{alg:outsample:cart} except that we obtain $g(\Z)$ through $g(\Z)=\Lsc_g^S(\bD^{uS};\Z)$ leveraging the surrogate $\bs^u$.
\end{algorithm}

\begin{remark}
We briefly remark on practical choices of the learning algorithm $\Lsc_g^S$ for $S\sim Z$. It is not hard to show that when $Y\sim Z$ follows a GLM, $S$ satisfying $S \indp Z \mid Y$ follows a single index model (SIM) given $Z$. Thus, when we desire to fit a penalized GLM to learn $Y\sim Z$ in the SSL scenario as in Implementation example \ref{example:1}, we can fit a penalized SIM for $S\sim Z$ in SA-SSL to learn $g(\cdot)$. For nonparametric or machine learning approaches like RF used in Implementation example~\ref{example:2}, we suggest still fitting RF in SA-SSL to learn $S\sim Z$ and $g(\cdot)$.
\label{rem:sass:learn}
\end{remark}

{\darkred
Finally, we provide an example, where we study the rate of the type-I error bound using the SSL and SA-SSL Maxway CRT approaches.

\begin{rateexampleprime}{1'}[Gaussian linear model cont'd]
\label{exam:SIM_theory}
Suppose that we have $n$ labeled samples $\bD=(\by,\bx,\bZ)$ and $N$ unlabeled samples with surrogate: $\bD^{uS}=(\bs^u,\bx^u,\bZ^u)$ where $N\gg n$ and $\bs^u=(S_1^u,S_2^u,\ldots,S_N^u)$.
Assume that the surrogates are post-hoc surrogates satisfying $S \indp Z \mid Y$. 
Assume further that $Y_i = Z_{i \cdot} \trans \theta^{\star} + \eta_i$, and $X_i = Z_{i \cdot}\trans \beta^{\star} + \varepsilon_i$, where $\varepsilon_i \sim \mathcal{N} (0,1)$ and $\eta_i \sim \mathcal{N} (0,1)$ are noise terms independent of $Z_{i \cdot}$.

\textit{Implementation of the Maxway CRT.}
Let $\theta$ and $\beta$ be estimators of $\thetas$ and $\betas$ respectively. We take $g(Z_{i \cdot}) =  Z_{i \cdot} \trans \theta$ to be an estimate of the conditional mean function of $Y_i$, take the transformation $R(X_i, Z_{i \cdot}) = X_i - Z_{i \cdot} \trans \beta$ and take $h(Z)$ as the null set.

\textit{Convergence rate assumptions.}
Assume that $\beta^{\star}$ is $s_\beta$ sparse, and it can be estimated with lasso on data of sample size $m$ with rate $\norm{\beta - \beta^{\star}}^2 \lesssim s_\beta \log(p)/m$. 
Similarly, assume that $\theta^{\star}$ is $s_\theta$ sparse, and it can be estimated with lasso on data of sample size $m$ with rate $\norm{\theta - \theta^{\star}}^2 \lesssim s_\theta \log(p)/m$. We refer to \citet{bickel2009simultaneous} and \citet{van2009conditions} for a more detailed discussion on the rate of lasso.

\textit{Rate of type-I error inflation using SA-SSL Maxway CRT.}
Since $Y \sim Z$ follows a linear model and the surrogate $S$ satisfies $S \indp Z \mid Y$, the surrogate $S$ follows a single index model (SIM) given $Z$, i.e., $S = f(Z\trans\thetas, e)$ with $e \indp Z$. \citet{li1989regression} establishes that when $S$ follows a SIM, the direction of $\thetas$ can be recovered using the least square regression of $S$ against $Z$; see also \citet{zhang2022prior}. In particular, this implies that if we run lasso with $S$ as response and $Z$ as predictors on the unlabeled samples, we can obtain an estimator $\theta$ of $\thetas$ such that $\norm{\theta \gamma - \theta^{\star}}^2 \lesssim s_\theta \log(p)/N$ for some constant $\gamma$. Therefore, we can show that using the inner-product statistic, the type-I error inflation of the SA-SSL Maxway CRT can be bounded by $\PP{p_{\maxway} \leq \alpha} - \alpha \lesssim \sqrt{s_\theta \log(p)/N}\sqrt{s_\beta \log(p) n/N}$. For arbitrary statistic, the bound becomes $\sqrt{s_\theta \log(p)n/N}\sqrt{s_\beta \log(p) n/N}$.

\textit{Comparison with the SSL Maxway CRT and the model-X CRT.} We can further establish that in this example using the inner-product statistic, the type-I error inflation of the SSL Maxway CRT can be bounded by $\PP{p_{\maxway} \leq \alpha} - \alpha \lesssim \sqrt{s_\theta \log(p)/n}\sqrt{s_\beta \log(p) n/N}$. For arbitrary statistic, the bound of the type-I error inflation becomes $\sqrt{s_\theta \log(p)}\sqrt{s_\beta \log(p) n/N}$. For the model-X CRT, the bound is $\sqrt{s_\theta \log(p)n/N}\sqrt{s_\beta \log(p) n/N}$. We include a summary of the bounds in Table~\ref{tab:summary_of_bound_SIM}. 
We first note that compared with the SSL Maxway CRT, the SA-SSL Maxway CRT achieves a better convergence rate. These results are due to the larger sample size of the surrogate samples. We also note that compared with the model-X CRT, the SA-SSL Maxway CRT gives a smaller rate of type-I error inflation as long as $s_{\theta} \log(p)/N \to 0$ using inner-product statistic (or $s_{\theta} \log(p)n/N \to 0$ for arbitrary statistic), demonstrating the robustness of the Maxway CRT. 
\end{rateexampleprime}

\begin{table}[htbp]
    \caption{Comparison of bounds on type-I error inflation in Convergence rate example \ref{exam:SIM_theory}. }
    \label{tab:summary_of_bound_SIM}
    \centering
    \begin{tabular}{|c|c|r|}
    \hline
        Method& Test statistic & \makecell{Rate of bound on\\type-I error inflation}  \\
        \hline
        SA-SSL Maxway CRT& Inner-product & $\sqrt{n} \Delta_{\theta} \Delta_{\beta} \qquad  $ \\
        \hline
        SA-SSL Maxway CRT& Arbitrary & $n \Delta_{\theta} \Delta_{\beta} \qquad $ \\
        \hline
        SSL Maxway CRT & Inner-product & $\sqrt{N} \Delta_{\theta} \Delta_{\beta} \qquad $\\
        \hline
        SSL Maxway CRT& Arbitrary & $\sqrt{n N} \Delta_{\theta} \Delta_{\beta} \qquad $ \\
        \hline
        Model-X CRT& Arbitrary/Inner-product & $\Delta_{\beta} \qquad$ \\
        \hline
        \multicolumn{3}{|r|}{Here we take $\Delta_{\theta}=\sqrt{s_\theta \log(p) /N}$ and $\Delta_{\beta}=\sqrt{s_\beta \log(p) /N}$. }\\
        \hline
    \end{tabular}
\end{table}

We provide more details of this example in Appendix \ref{subsection:detail_example_gauss_linear_cont}. We also include an additional example in Appendix \ref{subsection:example}, where we give rates of the bound on type-I error inflation of the model-X and Maxway CRT. }

\subsection{Transfer learning}\label{sec:tl}
{\darkred

The core idea of SA-SSL is to learn the low-dimensional $g(\Z)$ from some external data set that (i) has a much larger sample size or richer information than $\bD$ being used for the CRT; (ii) can correctly reveal, or at least be fairly informative to the true model of $\by\sim\Z$ in the targeted $\bD$. Generally speaking, any external data or knowledge with properties (i) and (ii) can be potentially incorporated with the Maxway framework to enhance the robustness of inference. This motivates us to further consider a transfer learning (TL) scenario where the external knowledge comes from some source data set $\bD^e=(\by^e,\bZ^e)$ with a much larger sample size (e.g., the ethnic majority group) compared to the target data $\bD=(\by,\bx,\bZ)$ (e.g., some minority group). The model $\by^e\sim\bZ^e$ tends to share some similarity with $\by\sim\bZ$, e.g., $Y$ depending on the same small subset of covariates in $Z$. However, such similarity is not always ensured and methods adaptive to the model discrepancy between the source and target data are highly desirable and preferable \citep[e.g.][]{li2022transfer,li2022transfer1, tian2022transfer,gu2022robust}.

With a similar spirit to the SA-SSL scenario, we propose Algorithm \ref{alg:outsample:cart:TL} that learns $g(\Z)$ using the external source data $\bD^e=(\by^e,\bZ^e)$ and transfers it to assist the Maxway CRT on the target data $\bD=(\by,\bx,\bZ)$.

\begin{algorithm}[H]
\caption{\label{alg:outsample:cart:TL} The transfer learning (TL) Maxway CRT.}
{\bf Input:} learning algorithms $\Lsc_g^e$, $\Lsc_h$, and $\Lsc_{\rho}$; a test statistic function $T$; unlabeled data $\bD^u=(\bx^u,\bZ^u)$; labeled data $\bD=(\by,\bx,\bZ)$; external source data $\bD^e=(\by^e,\bZ^e)$; and number of randomizations $M$. 

\vspace{3pt}

Implementation procedures are the same as those in Algorithm \ref{alg:outsample:cart} except that we obtain $g(\Z)$ through $g(\Z)=\Lsc_g^e(\bD^e;\Z)$ leveraging the source data set $\bD^e$.
\end{algorithm}

Similar to what to see in the SA-SSL Maxway CRT, the TL Maxway CRT tends to be more robust than the SSL Maxway CRT when the source and target data have very similar models for $Y\sim Z$, due to the larger sample size of $\y^e$. Meanwhile, when the source data is more dissimilar from the target data, the TL Maxway CRT remains viable for making valid inference thanks to its double robustness, as established by Theorems \ref{theo:exact_inference} and \ref{theo:almost_double_robust}. 
 }
 
\section{Simulation studies}\label{sec:sim}

\subsection{Semi-supervised setting}\label{sec:sim:ss}

We first conducted simulation studies to evaluate the Maxway CRT and compare it with existing approaches in terms of robustness and power, under the SSL scenario introduced in Section \ref{sec:const:semi}. {\bf R} codes for the implementation can be found at \url{https://github.com/moleibobliu/Maxway_CRT}. For data generation, we consider the following three configurations with different types of models for $X\mid Z$ and $Y\mid Z$. 

\begin{enumerate}
\item[] (SS.I) {\bf Gaussian linear $X\mid Z$ and $Y\mid Z$.} Generate $Z\in\mathbb{R}^p$ from $\mathcal{N}(\bzero, \bSigma)$ where $p=500$ and $\bSigma=(\sigma_{ij})_{p\times p}$ with $\sigma_{ij}=0.5^{|i-j|}$. Then generate $X$ and $Y$ following:
    \[
    X=0.3\sum_{j=1}^5\nu_j Z_j+\eta\sum_{\ell\in\Isc_1}\nu_{\ell}Z_{\ell}+\epsilon_1;\quad Y=\gamma h(X,Z)+0.3\sum_{j=1}^5\nu_j Z_j+\eta\sum_{\ell\in\Isc_2}\nu_{\ell}Z_{\ell}+\epsilon_2,
    \]
where $\epsilon_1,\epsilon_2$ are two $\mathcal{N}(0,1)$ noises, each $\nu_{j}$ is randomly picked from $\{-1,1\}$, and $\Isc_1,\Isc_2$ are two disjoint sets of indices randomly drawn from $\{6,7,\ldots,p\}$ satisfying $|\Isc_1|=|\Isc_2|=25$. 

\item[] (SS.II) {\bf Logistic linear $X\mid Z$ and gaussian linear $Y\mid Z$.} Generate $\Isc_1$, $\Isc_2$, $\nu_j$, $Z$ and $Y$ in the same way as (SS.I) while generate $X$ from $\{0,1\}$ following:
\[
\PP{X=1\mid Z}={\rm expit}\Big(0.3\sum_{j=1}^5\nu_j Z_j+\eta\sum_{\ell\in\Isc_1}\nu_{\ell}Z_{\ell}\Big),\quad\mbox{where}\quad{\rm expit}(a)=\frac{e^{a}}{1+e^{a}}.
\]

\item[] (SS.III) {\bf Non-linear $X\mid Z$ and $Y\mid Z$.} Generate $Z\in\mathbb{R}^p$ from $\mathcal{N}(\bzero, \bSigma)$ where $p=40$ and $\bSigma=(0.2^{|i-j|})_{p\times p}$, and $X$ and $Y$ following:
\begin{align*}
X=&0.5(I_1+I_3)+0.4(I_2+I_4+I_1I_4+I_2I_3)+0.15\Big(\sum_{j=21}^{24}I_j+ I_{21}I_{22}+I_{23}I_{24}\Big)+\epsilon_1;\\
Y=&\gamma h(X,\Z) +0.5(I_1+I_4+I_1I_4+I_2I_3)+0.4(I_2+I_3)+0.15\Big(\sum_{j=31}^{34}I_j+ I_{31}I_{32}+I_{33}I_{34}\Big)+\epsilon_2,
\end{align*}
where $\epsilon_1,\epsilon_2$ are again two $\mathcal{N}(0,1)$ noises, $I_1=I(Z_1>0)$, $I_2=I(Z_2>0.5)$, $I(Z_3>-0.5)=I_3$, $I(|Z_4|>1)$ and $I_j=I(Z_j>0)$ for any $j\geq 5$.

\end{enumerate}
Term $h(X,\Z)$ depicts the effect of $X$ on $Y$ given $Z$ under the alternative, and $\gamma$ measures the magnitude and direction of $X$'s effect. For $h(X,\Z)$, we separately consider two choices in Configurations (SS.I) and (SS.II): (SS.I) $h(X,\Z)=X$ for pure linear effect and (SS.II) $h(X,\Z)=X+X\sum_{j=1}^5Z_j$ for a mixture of linear and interaction effect, and set $h(X,\Z)=(0.5X^2+\sin\{\pi(X-1)/4\})(I_1+I_2)$ in Configuration (SS.III). When evaluating the type-I error of testing $X\indp Y\mid\Z$, we set $\gamma=0$. For power evaluation, we let $\gamma$ vary within a proper range to plot the power curve of each approach against $\gamma$. 

In (SS.I) and (SS.II), $\Isc_1\cap\Isc_2=\emptyset$ so the parameter $\eta$ controls the part of $\Z$'s effects nearly not confounding the relationship between $X$ and $Y$. As it increases, the overlap between $\Z$'s effects on $X$ and $Y$ gets smaller, and thus $Z$'s confounding effect on $X$ and $Y$ becomes less significant. We set $\eta=0$ for ``strong overlapping", $\eta=0.1$ for ``moderate overlapping", and $\eta=0.2$ for ``weak overlapping". In all configurations, we generate $n=250$ samples with $(X,Y,Z)$ for randomization tests and $N$ additional samples of $(X,Z)$ to estimate the distribution of $X\mid Z$. We let $N$ vary in a wide range to evaluate the type-I error inflation under different qualities of X-modeling. We also evaluate the power under a large $N$ enabling accurate estimation of $X\mid Z$ and proper type-I error control.

We implement four approaches including (a) the {\bf Maxway$_{\rm in}$ CRT}: the in-sample training version of our approach, i.e., Algorithm \ref{alg:insample:cart}; (b) the {\bf Maxway$_{\rm out}$ CRT}: the holdout training version of our approach in Algorithm \ref{alg:outsample:cart}; (c) the {\bf model-X CRT} proposed by \cite{candes2018panning}; and (d) the {\bf model-X CPT} proposed by \cite{berrett2018conditional}. Since the Maxway$_{\rm out}$ CRT requires an estimate of the conditional model $Y\mid Z$ independent from the data used for randomization tests, we generate $n=250$ additional samples of $(Y,Z)$ as the training dataset in Algorithm \ref{alg:outsample:cart}. In this way, we actually use more labels for the Maxway$_{\rm out}$ CRT than that for other approaches. For a fair comparison between our proposal and the model-X CRT and CPT, one should refer to the Maxway$_{\rm in}$ version. While we include the Maxway$_{\rm out}$ CRT for comparison to study if the Maxway$_{\rm in}$ CRT would encounter over-fitting issues in the finite-sample studies.

In Configuration (SS.I), we model $X\mid Z$ as a gaussian linear model with variance $\sigma^2$, then implement linear lasso tuned by cross-validation on the unlabeled data to estimate its conditional mean and use the sample variance of its residual evaluated on the labeled testing data to estimate $\sigma^2$. This could ensure the resampled $\x\supm$ to have nearly the same scale as the observed $\x$ and thus the randomization test robust to the estimation error in $\sigma^2$. In all approaches, we use the linear lasso to estimate $Y\mid Z$ and construct both the d$_0$ and d$_{\mathrm{I}}$ statistics; see Implementation example \ref{example:1}. We also follow Implementation example \ref{example:1} to specify $g(Z)$ and fit a linear model for $X-\EEhat{X\mid Z}$ against $g(Z)$, to adjust the conditional distribution of $X$. For (SS.II), we adopt the same setup as in (SS.I) with all the linear models replaced with the logistic model. In Configuration (SS.III), we use RF to learn the non-linear $Y\mid Z$ and measure the non-linear dependence of $Y$ on $X$, as introduced in Implementation example \ref{example:2}. Also, we fit RF to learn $X\mid Z$ and adopt Implementation example \ref{example:2} to adjust it against $g(Z)$. For all configurations, we set the reduced dimensionality $k=\lceil 2\log p\rceil$ and the nominal level as $0.05$.

In Figure \ref{fig:gauss}, we plot the type-I error against $N$ varying from $250$ to $2000$, as well as the power adjusted to type-I error (defined as the original average power minus the type-I error) when $N=2000$ under Configuration (SS.I). In the main paper, we only include the results corresponding to pure linear effect, i.e. $h(X,\Z)=X$, and the d$_0$ test statistic. We observe similar patterns as Figure \ref{fig:gauss} in the remaining setups of (SS.I), i.e. $h(X,\Z)$ containing the interaction effect or the d$_{\mathrm{I}}$ test statistic is used. These results are presented in Figures \ref{fig:gauss:add:1}--\ref{fig:gauss:add:3} of Appendix \ref{sec:app:num}. Similarly, we present in Figure \ref{fig:binary} the type-I error and power under Configuration (SS.II) with $h(X,\Z)=X$ and the d$_0$ statistic used, with the other setups of (SS.II) presented in Figures \ref{fig:binary:add:1}--\ref{fig:binary:add:3}. Finally, we plot the type-I error and average power (evaluated with $N=3200$) under Configuration (SS.III) in Figure \ref{fig:ml}.

\begin{figure}
    \centering
    \includegraphics[width=0.37\textwidth]{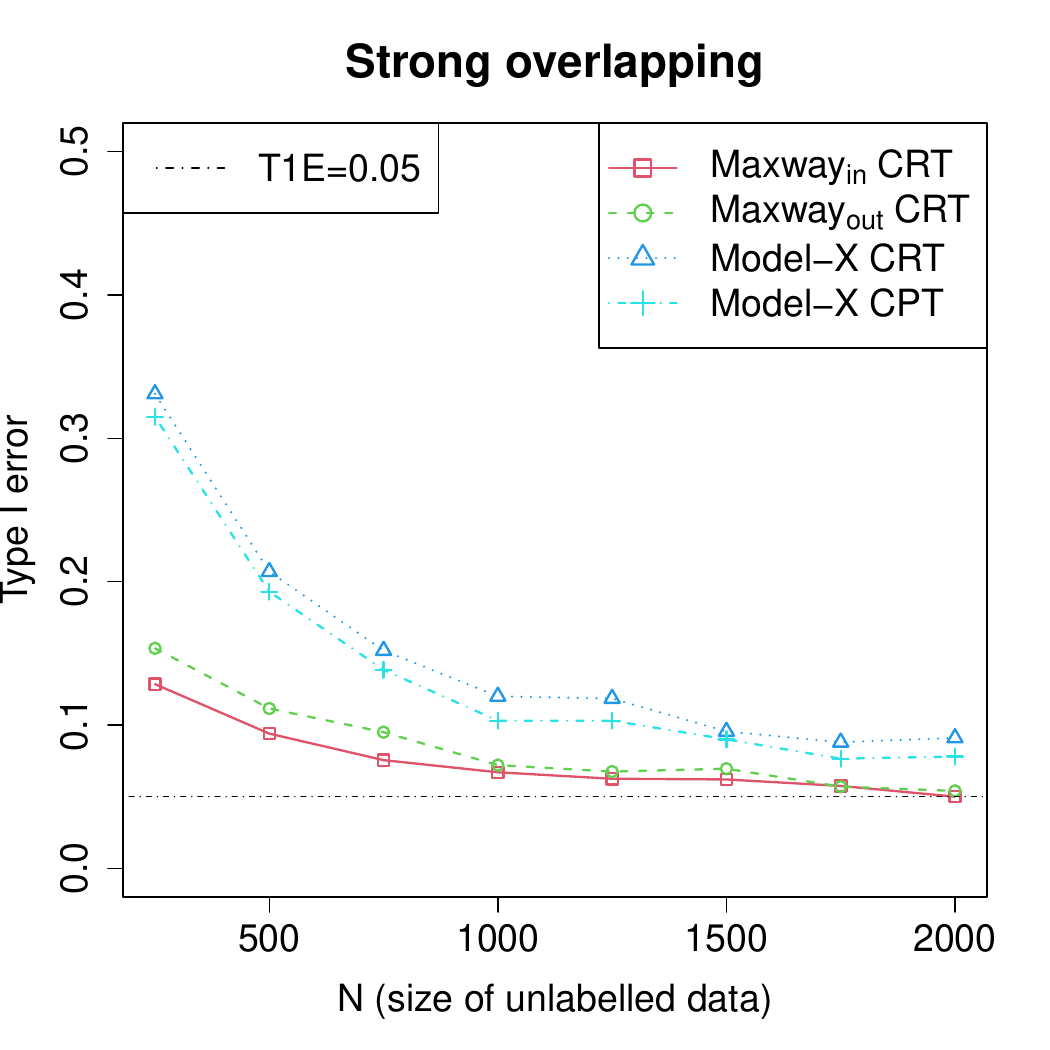}
    \includegraphics[width=0.37\textwidth]{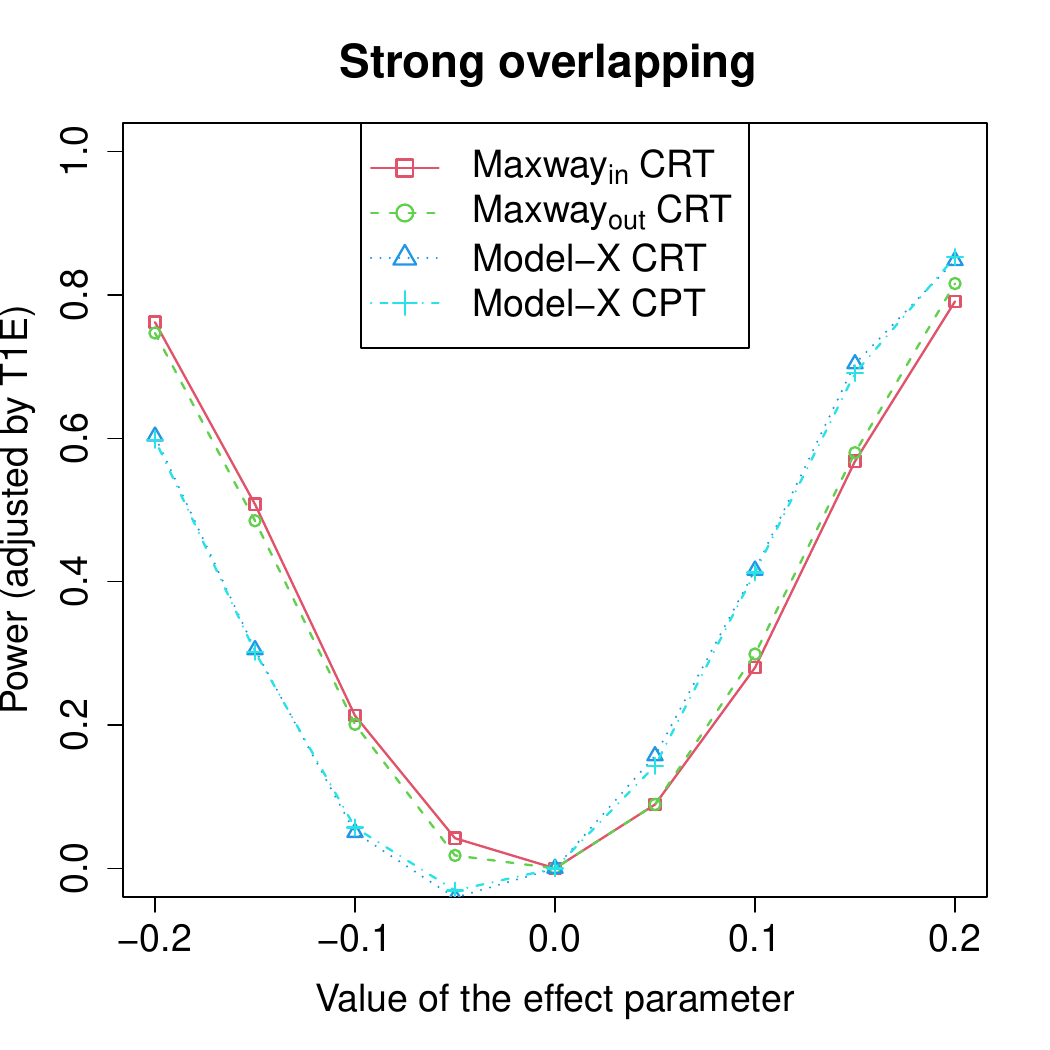}
    \includegraphics[width=0.37\textwidth]{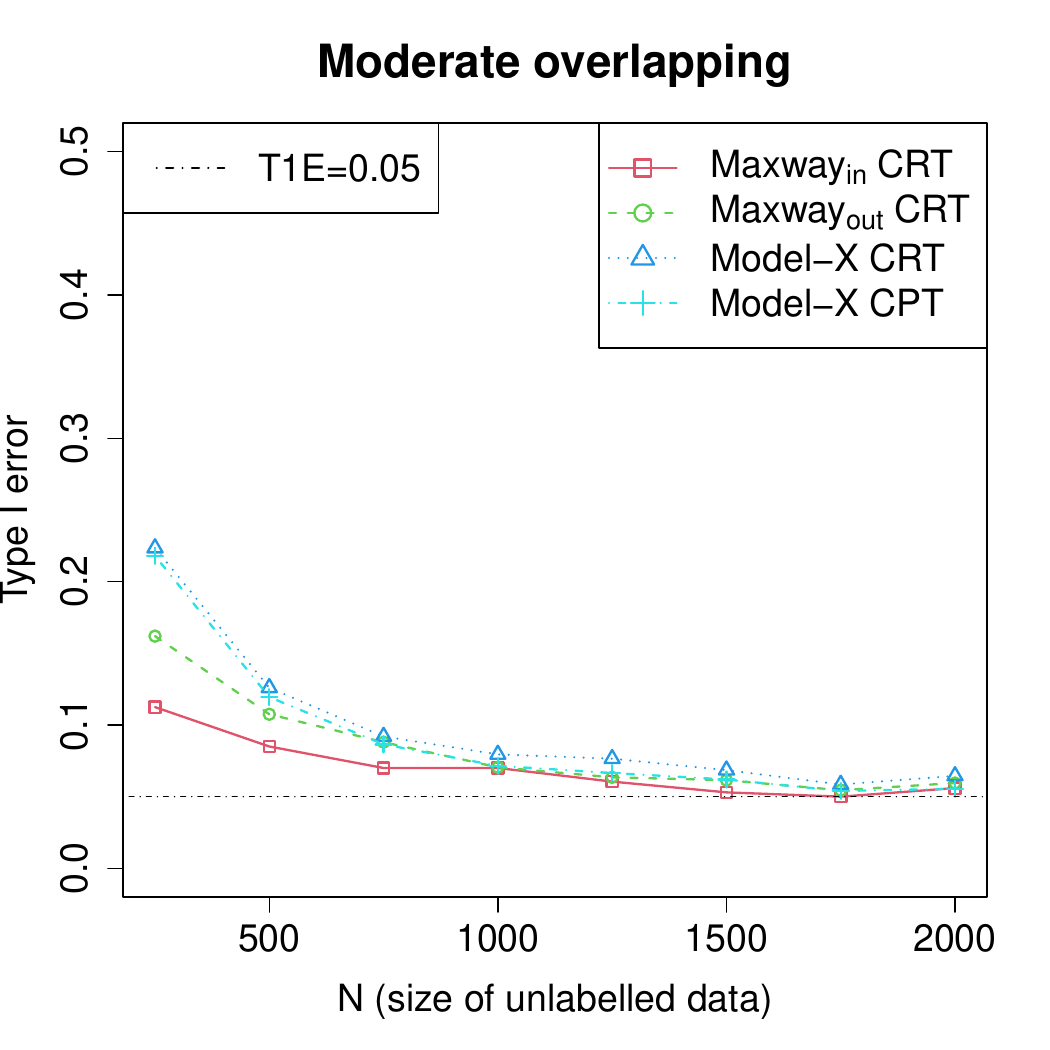}
    \includegraphics[width=0.37\textwidth]{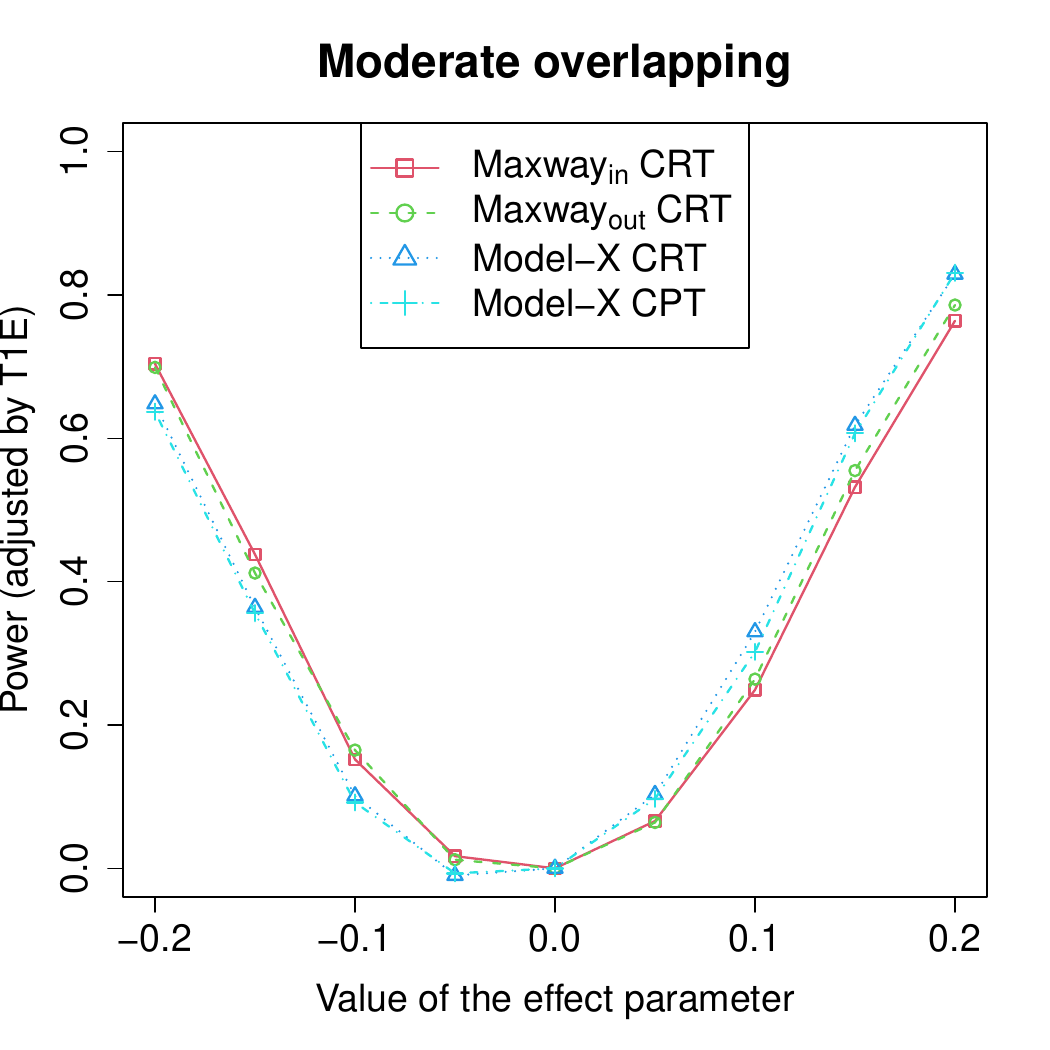}
    \includegraphics[width=0.37\textwidth]{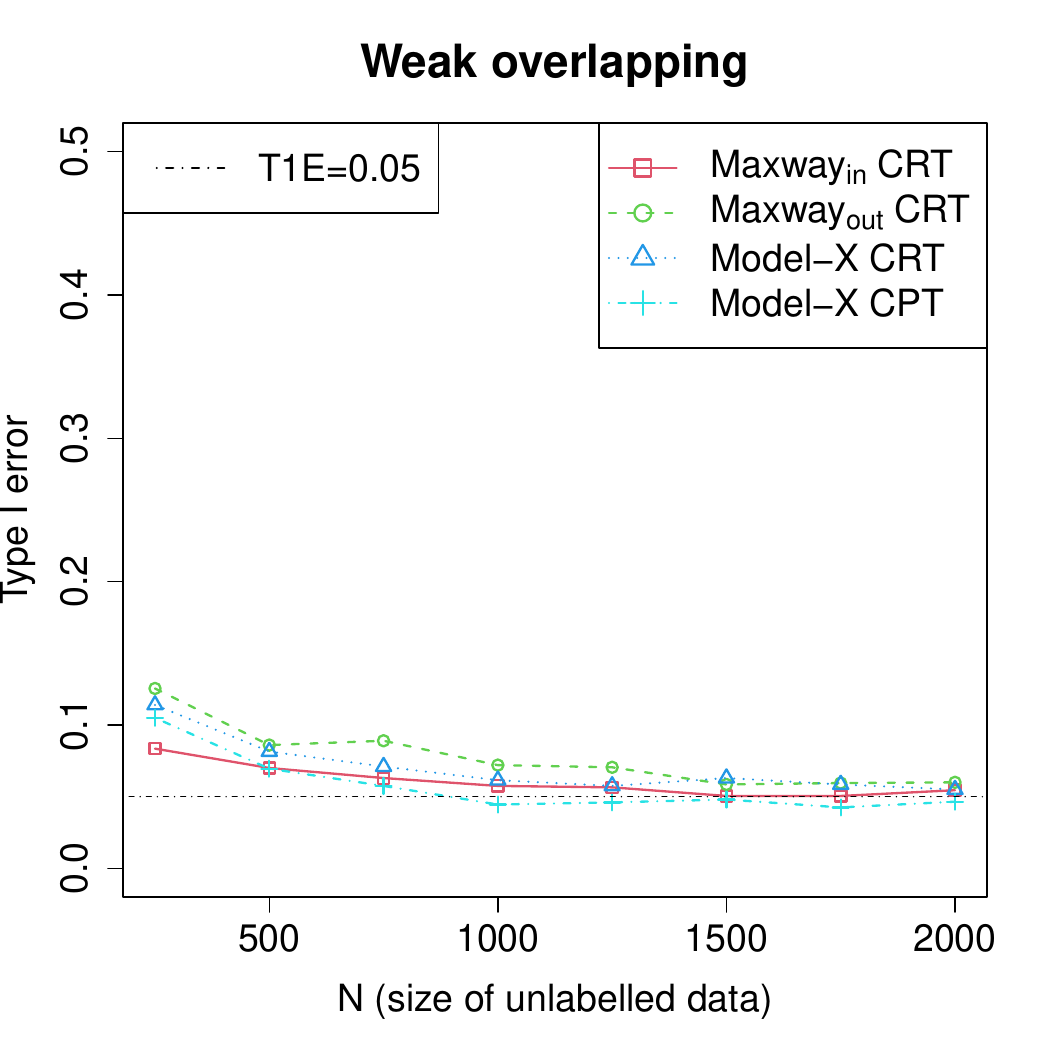}
    \includegraphics[width=0.37\textwidth]{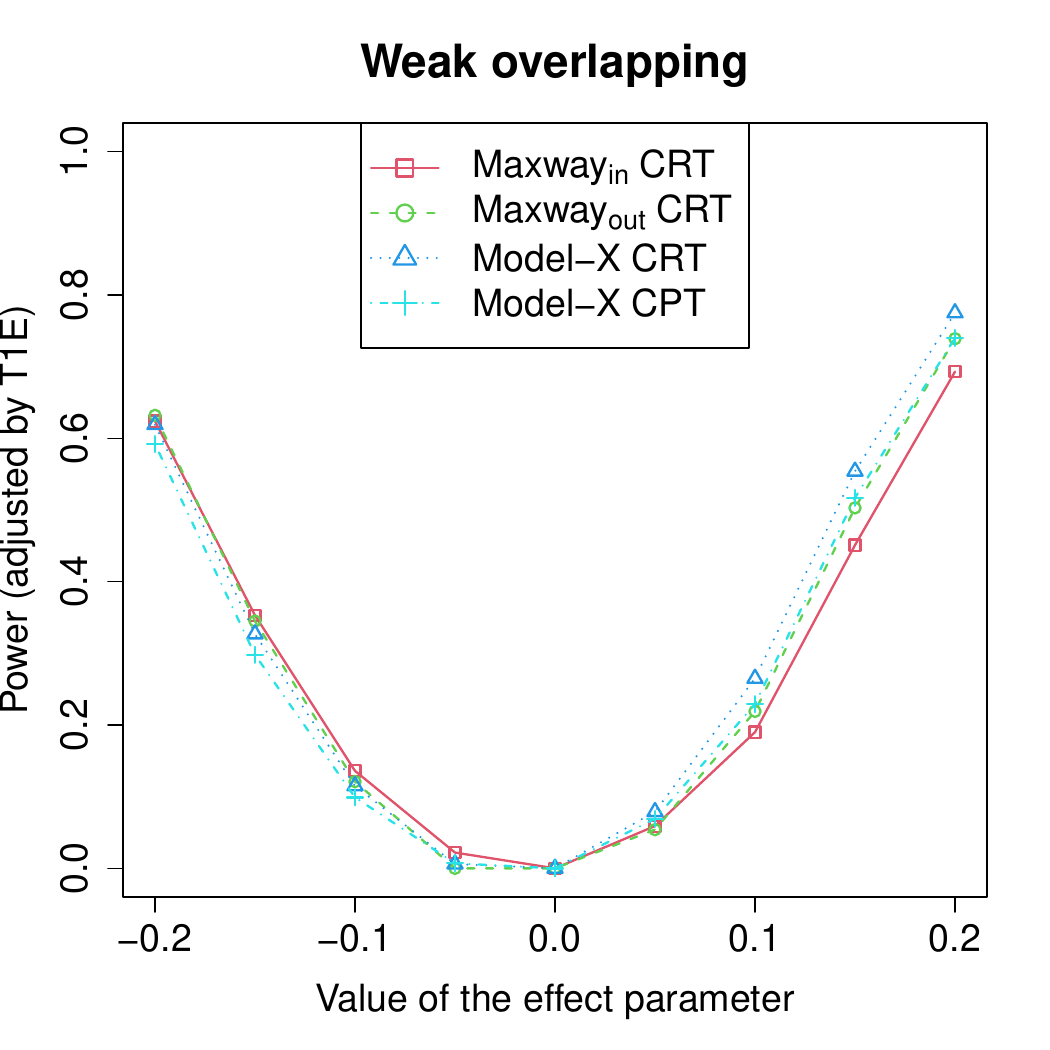}
    \caption{Type-I error and average power (adjusted by type-I error, i.e. the original average power minus the type-I error) under the three overlapping scenarios (i.e. $\eta=0,0.1,0.2$) of Configuration (SS.I) {\bf gaussian linear $X\mid Z$ and $Y\mid Z$} with $h(X,Z)=X$ and the d$_0$ statistic used for testing, as introduced in Section \ref{sec:sim:ss}. The replication number is $500$ and all standard errors are below $0.01$.} 
    \label{fig:gauss}
\end{figure}

\begin{figure}
    \centering
    \includegraphics[width=0.37\textwidth]{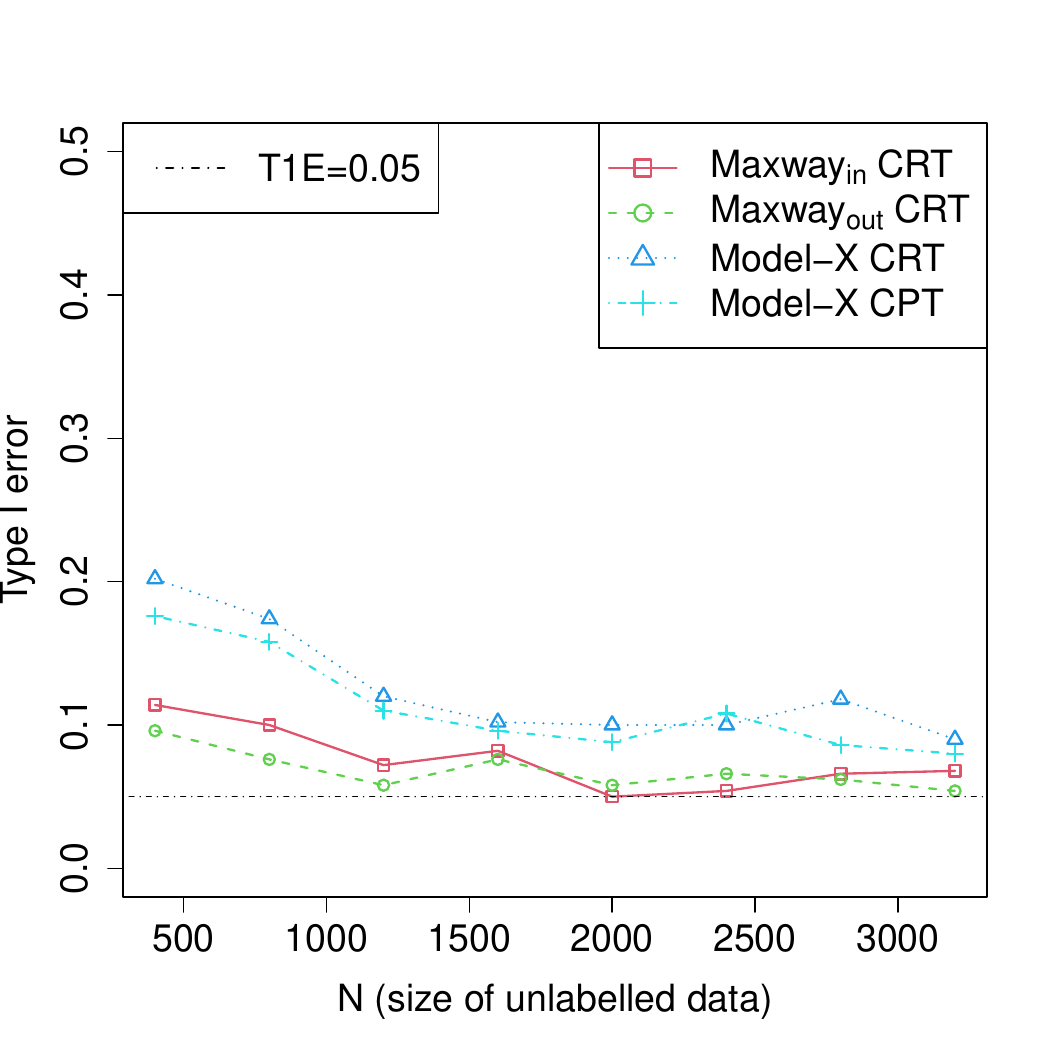}
    \includegraphics[width=0.37\textwidth]{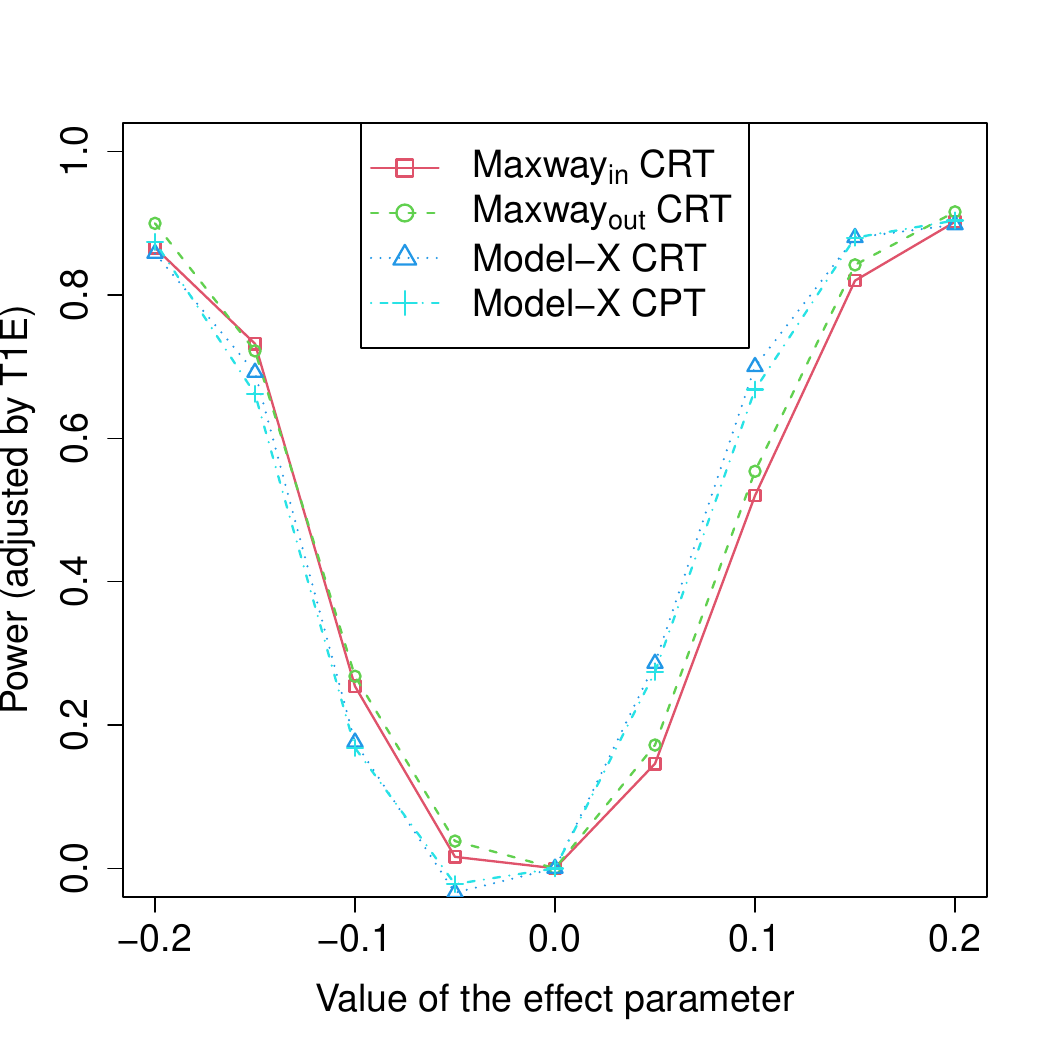}
    \caption{Type-I error and average power (adjusted by type-I error) under Configuration (SS.III) {\bf non-linear $X\mid Z$ and $Y\mid Z$} with the importance of $X$ measured using random forest, as introduced in Section \ref{sec:sim:ss}. The replication number is $500$ and all standard errors are below $0.01$.} 
    \label{fig:ml}
\end{figure}

Under the linear model configurations (SS.I) and (SS.II), the Maxway$_{\rm in}$ CRT shows much less type-I error inflation than both the model-X CRT and CPT when there is a strong overlapped effect of $Z$ on $X$ and $Y$ or small $N$. Under the weak overlapping scenarios or with larger $N$, all approaches tend to have better validity while our approach still achieves better type-I error control. To understand the results, note that for model-X inference, the dependence of $X$ on $Z$ is not adequately characterized and adjusted due to the shrinkage bias of lasso. This is even pronounced when the quality of model-X is poor under small $N$, and could cause severe confound to $X$ and $Y$ especially when their shared variation from $Z$ (i.e. $0.3\sum_{j=1}^5\nu_j Z_j$) is more dominating as in the strong overlapping scenario. Our Maxway approach mitigates this impact by adjusting the learned $X\mid Z$ against the low-dimensional important features predictive of $Y$.

{\darkred 
Interestingly, this adjustment also makes the powers of our approach different from the model-X CRT and CPT, especially under strong overlapping. Note that in our data generation, $Z$ has the same sign of the confounding effect. When such confound is not adjusted adequately, it tends to make $X$ and $Y$ positively correlated. Consequently, when $\gamma$, the signal of $X$ is positive, the inadequately adjusted confounding of $Z$ makes the effect of $X$ on $Y$ spuriously stronger and thus increases the power of the model-X approaches compared with the Maxway approach. In contrast, it makes the power of model-X lower than the Maxway when $\gamma<0$ and $Z$'s confounding effect is opposite to the effect of $X$. We study this phenomenon with more details in Appendix \ref{sec:app:sym} and propose an alternative power calculation procedure in simulations to adjust for such confounding bias and produce a more comparable power evaluation. As shown in Figure \ref{fig:bias:adj}, after proper adjustment, the model-X and Maxway approaches show basically the same power across all signals.} Thus, compared with the model-X approaches, Maxway has a similar power on average but its performance is more balanced between the positive and negative signals. 

In Configuration (SS.III), the importance of $X$ is characterized by a highly non-linear test statistic more complicated than those in (SS.I) and (SS.II), as well as existing semiparametric inference approaches like DML and double selection. Also, the models of $X\mid Z$ and $Y\mid Z$ are non-linear and harder to estimate than those in (SS.I) and (SS.II). So compared with the linear settings, all approaches require a larger training size for $X\mid Z$ and lower dimensionality of $Z$ to achieve proper type-I error control. While our Maxway approach still attains significantly better type-I error control and more balanced power than the model-X CRT/CPT. This is again because our adjustment with a low-dimensional RF model reduces the confound not well adjusted by the high-dimensional RF model for $X\mid Z$.

To understand how the Maxway CRT with the in-sample training of $Y\mid Z$ performs compared to its out-of-sample version using an additional set of $n=250$ labeled samples for a holdout $Y\mid Z$ training, we inspect and compare the performance of the Maxway$_{\rm in}$ and Maxway$_{\rm out}$ CRT in the three configurations. 
One may expect that Maxway$_{\rm out}$ should have better type-I error control than Maxway$_{\rm in}$, due to the potential over-fitting issue of Maxway$_{\rm in}$ discussed in Section \ref{sec:const:semi}. Interestingly, this is the case in Configuration (SS.III) with non-linear and complicated models but is contradictory to our results in (SS.I) with linear models. This is probably because, unlike random forest, lasso highly shrinks the model coefficients and concurs with ``under-fitting" rather than ``over-fitting". Thus, the in-sample fitting of $Y\mid Z$ with lasso  results in a smaller empirical partial correlation between $X$ and $Y$ under the null, which is related to the observation that the mean square of residuals of lasso tends to over-estimate the noise level \citep{sun2012scaled}.

\subsection{Surrogate-assisted semi-supervised setting}\label{sec:sim:sass}

{\darkred

We further extend our simulation studies to the SA-SSL scenario described in Section \ref{sec:sa:ssl}. Consider two data generation configurations with different models of $X\mid Z$, $Y\mid Z$, and $S\mid (Y,Z)$.

\begin{enumerate}

\item[] (SAS.I) {\bf Logistic linear $X\mid Z$ and $Y\mid Z$.} Generate $Z\in\mathbb{R}^p$ from $\mathcal{N}(\bzero, \bSigma)$ where $p=500$ and $\bSigma=(0.3^{|i-j|})_{p\times p}$. Then generate $X$ and $Y$ following:
\[
\PP{X=1\mid Z}={\rm expit}\Big(0.2\sum_{j=1}^5\nu_j Z_j\Big),\quad \PP{Y=1\mid Z}={\rm expit}\Big(\gamma X+0.4\sum_{j=1}^5\nu_j Z_j\Big),
\]
where $\nu_{j}$ is randomly picked from $\{-1,1\}$. Finally, generate $S$ given $(Y,Z)$ following
\[
S=\zeta_y Y + 0.2\zeta_z\sum_{\ell\in\Isc}Z+\epsilon,
\]
where $\epsilon\sim \mathcal{N}(0,1)$ and $\Isc$ is an index set randomly drawn from $\{6,7,\ldots,p\}$ satisfying $|\Isc|=10$.

\item[] (SAS.II) {\bf Non-linear $X\mid Z$ and $Y\mid Z$.} Generate $Z\in\mathbb{R}^p$ from $\mathcal{N}(\bzero, \bSigma)$ where $p=40$ and $\bSigma=(0.3^{|i-j|})_{p\times p}$, and $X$ and $Y$ following:
\begin{align*}
X=&\One{0.5I_1+0.4I_2+0.5I_3+0.4I_4+0.5(I_1I_4+I_2I_3)+\epsilon_1>0};\\
Y=&\One{\gamma \sin\{\pi(X-1)/4\} +0.3I_1+0.5I_2+0.5I_3+0.6I_4+0.5(I_1I_2+I_3I_4)+\epsilon_2>0},
\end{align*}
where $\epsilon_1,\epsilon_2\sim\mathcal{N}(0,0.5)$, $I_1=I(Z_1>0)$, $I_2=I(Z_2>0.5)$, $I(Z_3>-0.5)=I_3$ and $I(|Z_4|>1)$. Again, generate $S$ following:
\[
S=0.5\zeta_y Y + 0.5\zeta_z\sum_{\ell\in\Isc}Z+\epsilon,
\]
where $\epsilon\sim \mathcal{N}(0,1)$ and $\Isc$ is an index set randomly drawn from $\{6,7,\ldots,p\}$ satisfying $|\Isc|=10$. 

\end{enumerate}

Similar to Section \ref{sec:sim:ss}, generation mechanisms of $(Y,X,Z)$ in (SAS.I) and (SAS.II) correspond to linear and non-linear models respectively. We again set $\gamma=0$ to evaluate type-I error and let $\gamma$ vary from $-0.8$ to $0.8$ for power evaluation. As is outlined in Algorithm \ref{alg:outsample:cart:SA:SSL}, different from the SSL scenario, we generate surrogate $S$ for the $N$ unlabeled samples and train models for $S\sim Z$ instead of $Y\sim Z$ to learn the function $g(Z)$ used in the Maxway CRT. To generate $S$ in (SAS.I) and (SAS.II), we consider three settings separately: (i) {\bf strong and perfect surrogate}: $\zeta_y=3$ and $\zeta_z=0$; (ii) {\bf weak and perfect surrogate}: $\zeta_y=1$ and $\zeta_z=0$; (iii) {\bf strong and imperfect surrogate}: $\zeta_y=3$ and $\zeta_z=10^{-1/2}$. When $|\zeta_y|$ becomes larger, $S$ will be more predictive of $Y$ and thus $S\sim Z$ can provide more precise $g(Z)$. When $\zeta_z=0$, it holds that $Z\indp S\mid Y$ and by Proposition \ref{prop:2}, $S$ is a perfect surrogate. For $\zeta_z\neq 0$, it is not hard to show that $S$ is imperfect and $S\sim Z$ could be less informative of $Y\sim Z$ while there is still a hope of leveraging $S$ to improve robustness since $S\sim Z$ is a sparse model involving all predictors of $Y$. We again generate $n=250$ samples with $(X,Y,Z)$ for randomization tests and $N$ samples of $(S,X,Z)$ to estimate the distribution of $X\mid Z$ and learn $g(Z)$ from the model of $S\sim Z$, with $N$ varying in a proper range to evaluate the performance in controlling type-I error. For power evaluation, we stick to the strong and perfect surrogate and $N=2000$ in both configurations.

We include three approaches for comparison include (a) the {\bf SA-SSL Maxway CRT} introduced in Algorithm \ref{alg:outsample:cart:SA:SSL}; (b) the {\bf model-X CRT}; and (c) the {\bf model-X CPT}. Similar to Section \ref{sec:sim:ss}, we adopt Implementation example \ref{example:1} based on (logistic) lasso in Configuration (SAS.I) and Implementation example \ref{example:2} based on RF in (SAS.II), and use the same test statistic to implement the model-X CRT and CPT. The only difference lies in the step of learning $g(Z)$ as we no longer use $Y$ but follow Remark \ref{rem:sass:learn} to fit a sparse SIM in (SAS.I) and an RF model in (SAS.II) for $S\sim Z$ with all $N$ unlabeled samples. For both configurations, we set the parameter $k=\lceil 1.5\log p\rceil$ and the nominal level as $0.05$. 

The type-I error and power plots are presented in Figure \ref{fig:sur:binary} for Configuration (SAS.I) and in Figure \ref{fig:sur:rf} for (SAS.II). The SA-SSL Maxway CRT with all different qualities of surrogates (i.e. settings (i)--(iii)) has significantly smaller type-I errors than the model-X CRT and CPT. For example, the Maxway CRT with perfect surrogate successfully controls the type-I error around the nominal level $0.05$ when $N=2000$ under both configurations while the model-X inference approaches have their type-I error inflation larger than $50\%$ of the nominal level under (SAS.I) and $100\%$ under (SAS.II). Similar to the SSL scenario, the Maxway CRT shows lower power than the model-X CRT/CPT when $\gamma>0$ and higher than the latter when $\gamma<0$ while they achieve similar power in overall. This is again due to the relatively inadequate adjustment of $Z$'s confounding effect by the model-X approaches as discussed in Section \ref{sec:sim:ss} and studied in Appendix \ref{sec:app:sym}. In addition, the Maxway CRT constructed with a strong and perfect surrogate achieves the best type-I error control among Settings (i)--(iii) in both configurations. Interestingly, a strong but imperfect surrogate turns out to work better than a perfect but weak surrogate for relatively small $N$ but worse than the latter for large $N$.

Finally, we note that the TL scenario introduced in Section \ref{sec:tl} can be studied with quite similar designs being used in this section because, for both SA-SSL and TL, we are essentially concerned about the same question. That is how the discrepancy between the prior information learned from the surrogate or the external data and the underlying true model of $Y\sim \Z$ affects the performance of our method. Such discrepancy is reflected by the data generation parameter $\zeta_z$ in this section.

}

\begin{figure}[H]
    \centering
    \includegraphics[width=0.45\textwidth]{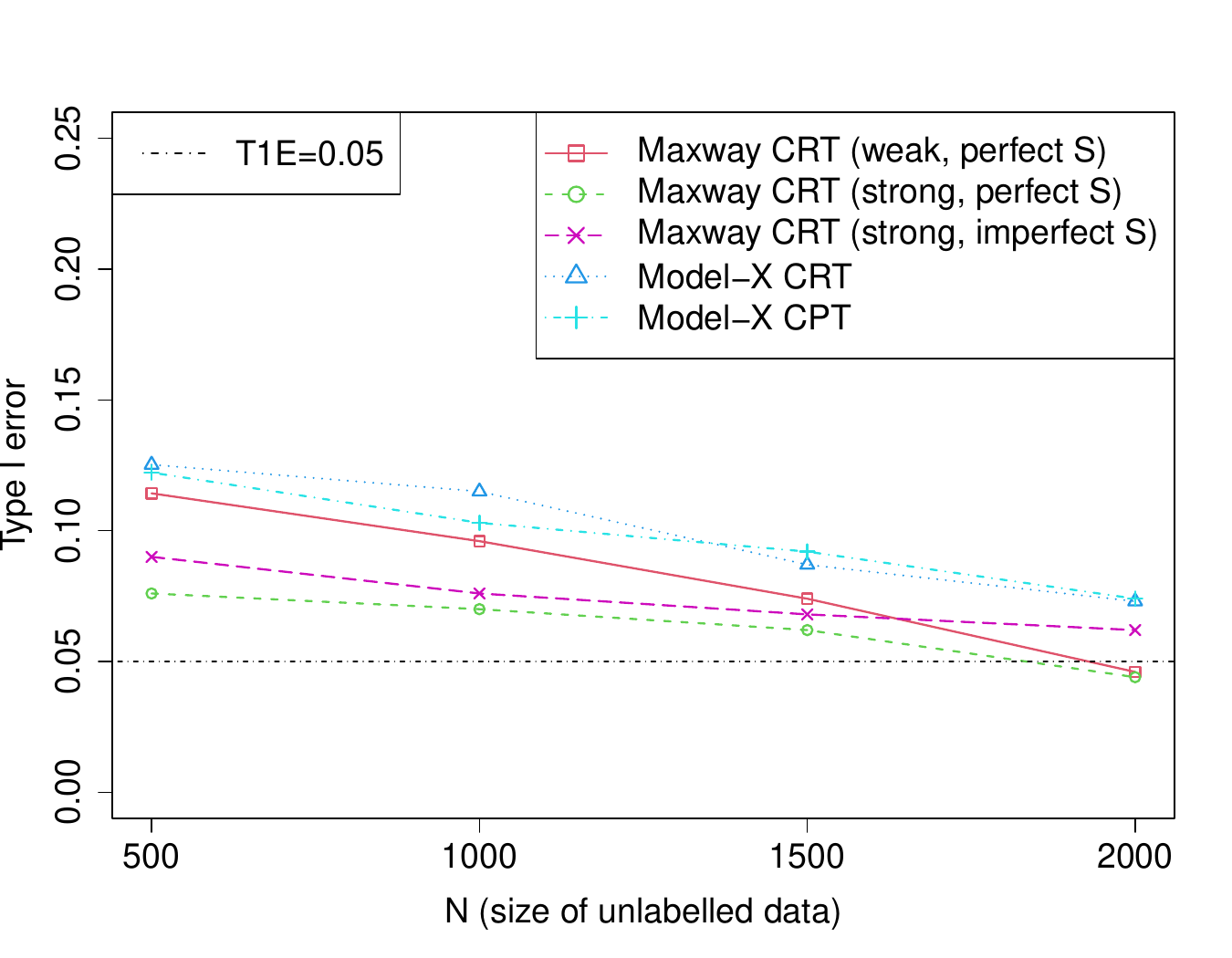}
    \includegraphics[width=0.45\textwidth]{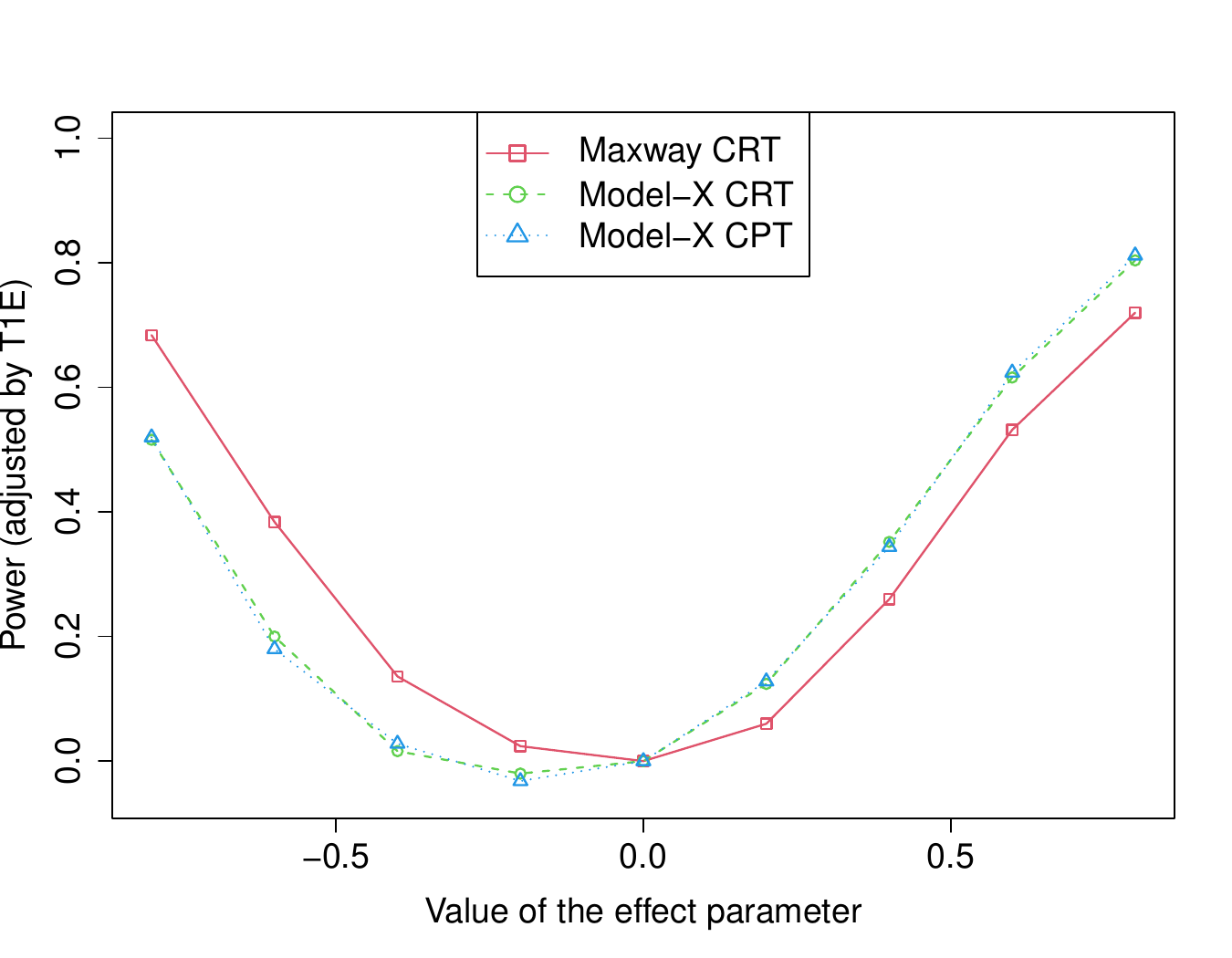}
    \caption{\darkred Type-I error and average power (adjusted by type-I error) under the Configuration (SAS.I) described in Section \ref{sec:sim:sass}. The three versions of the (SA-SSL) Maxway CRT in the Type-I error plot refer to the setups (i)--(iii) of $S\mid (Z,Y)$. The replication number is $1000$ and all standard errors are below $0.01$.} 
    \label{fig:sur:binary}
\end{figure}

\begin{figure}[H]
    \centering
    \includegraphics[width=0.45\textwidth]{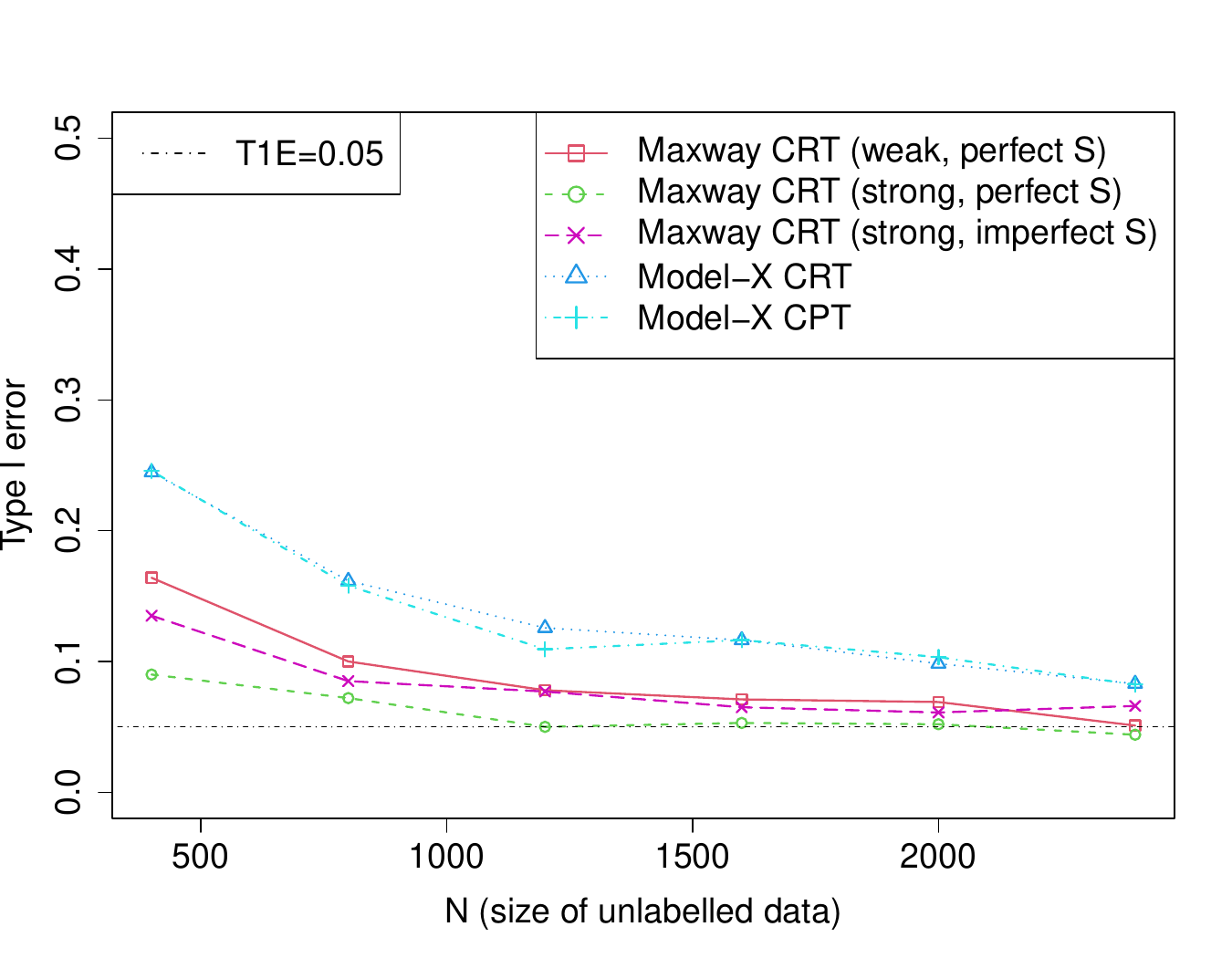}
    \includegraphics[width=0.45\textwidth]{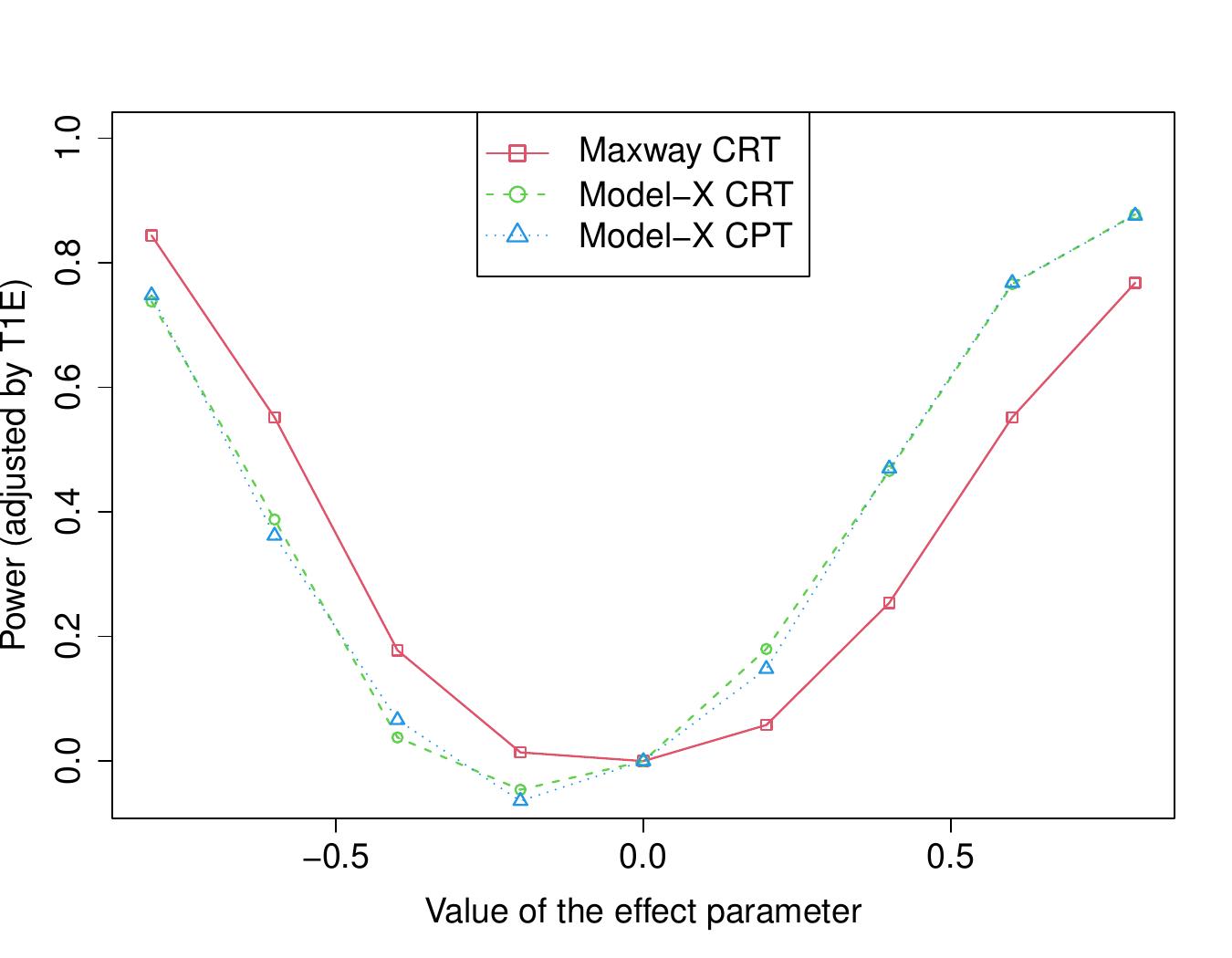}
    \caption{\darkred Type-I error and average power (adjusted by type-I error) under the Configuration (SAS.II) described in Section \ref{sec:sim:sass}. The three versions of the (SA-SSL) Maxway CRT refer to the setups (i)--(iii) of $S\mid (Z,Y)$. The replication number is $1000$ and all standard errors are below $0.01$.} 
    \label{fig:sur:rf}
\end{figure}

\section{Real Examples}

{\darkred

\subsection{An SA-SSL: studying obesity paradox with EHR data}

Besides the standard SSL setting, we also implement the Maxway CRT on an SA-SSL example concurred in an EHR-based biomedical study. Obesity is a common risk factor for type II diabetes (T2D) \citep{chan1994obesity,reaven1995pathophysiology} and both obesity and T2D are known to increase the risk of heart failure (HF) \citep{ali1999clinical,kenchaiah2002obesity}. However, it has been found in existing studies that among the patients already having T2D, the risk of HF becomes negatively associated with the presence of obesity, which is known as the ``obesity paradox" \citep{hainer2013obesity}. For example, in the cohort study of \citep{pagidipati2020association} with around 14,000 subjects having T2D and cardiovascular disease (CVD) at baseline, the HF risk in the overweight group was actually lower than the under/normal weight group (hazard ratio $0.83$, $95\%$ CI $0.71$--$0.98$). Interestingly, it is still an open problem whether the obesity paradox is a true negative association or just an epidemiological artifact caused by insufficient adjustment of confounding effects. As an example, one possible explanation for the obesity paradox is the lead-time bias, which suggests that obese individuals may develop CVD at an earlier age when they are generally in a healthier state and have fewer accompanying medical conditions than non-obese individuals who develop CVD \citep{elagizi2018overview,pagidipati2020association}. CRT could be a promising method to solve this problem since it can remove the confounding bias by conditioning on a large number of demographic and baseline adjustment features. We studied this problem on an EHR cohort extracted from Mass General Brigham (MGB) Healthcare System.

Our T2D cohort is defined as the subjects with at least one International Classification of Diseases (ICD) code of T2D occurring before the onset of the diagnostic code for HF. For each subject, we define the time of having the first T2D code in EHR as the baseline, and take the exposure variable $X$ as the indicator for the presence of overweight at this baseline. To remove potential confounding bias, we include the following sets of EHR features in the adjustment covariates $Z$: (i) demographic variables like age, gender, and ethnicity; (ii) indicators for the presence of all the other diagnostic codes (rolled up to PheCodes) at the baseline, which reflects the existence of any other diseases; (iii) the total health utilization measured by the days of visits up to the baseline. The outcome $Y$ is the gold standard label for HF obtained via chart review and the surrogate variable $S$ is naturally taken as the log count of the ICD code for HF. There are $N=11858$ subjects with $n=84$ labeled and the number of cases ($Y=1$) being $27$, and the number of the baseline adjustment features $p=516$. We notice that $S$ is an informative yet error-prone outcome for the true HF status $Y$, with the area under the receiver operating characteristic curve (AUC) being $0.85$ on the labeled samples.

\begin{table}
\caption{\label{tab:info:obe} Summary information of the data sets used in our real example of SA-SSL, studying the obesity paradox with the EHR data from MGB.}
 \centering
 \begin{tabular}{cccc}
    \hline
   & CRT data & X-modeling data & Surrogate data \\
    \hline
    \\[-2.5ex] 
   Structure & $(\by,\bx,\bZ)$ & $(\bx^u,\bZ^u)$ &  $(\bs^u,\bZ^u)$ \\
   Sample size  & 84 & 11,858 & 11,858  \\
  \# of Cases ($Y=1$ or $S>0$) & 27 & -- & 3,387  \\
  \# of Exposed ($X=1$) & 12 & 1,531 &  -- \\
   \\[-2.5ex] 
      \hline
    \end{tabular}
\end{table}

We include the SA-SSL version of the Maxway CRT (see Algorithm \ref{alg:outsample:cart:SA:SSL}), the model-X CRT, and the model-X CPT to test $X\indp Y\mid Z$. For implementation, we again fit logistic lasso to estimate both $X\mid Z$ and $Y\mid Z$ in all approaches, single index regression with the lasso penalty \citep{neykov2016l1} for $S\sim Z$ to learn the Maxway distribution as suggested in Remark \ref{rem:sass:learn}, and construct the d$_0$ statistic for importance measure. For all approaches, we resample for $M=1000$ times to estimate the $p$-values.

\begin{table}
\caption{\label{table:pval:sas} $p$-values of the d$_0$ statistic for the association between the risk of heart failure (HF) and the presence of overweight conditional on all the other disease conditions at the baseline of T2D.}
    \centering
    
 \begin{tabular}{cccc}
    \hline
        & Model-X CRT & Model-X CPT & Maxway CRT \\
    \hline
    \\[-2.5ex] 
   $p$-value & $0.052$ & $0.116$ &  $0.041$  \\
   \\[-2.5ex] 
      \hline
    \end{tabular}
\end{table}

The output $p$-values are presented in Table \ref{table:pval:sas}. The model-X and Maxway CRT output close $p$-values rejecting $X\indp Y\mid Z$ at the level $0.05$ while the model-X CPT produces a much bigger $p$-value being non-significant. This is probably due to the conservativeness of the CPT caused by conditioning on more observed information, which could become even more pronounced with a small labeled sample size $n$ in this example. Note that, unlike the model-X CPT, the Maxway CRT does not produce a more conservative $p$-value than the model-X CRT in this example. More importantly, we found the observed partial covariance between $Y$ and $X$ (adjusted to $Z$) is positive in all methods, indicating that the presence of obesity increases the risk of HF in the T2D cohort. Thus, after adequately adjusting for the other disease conditions at the baseline, obesity shows no benefit but probably an adverse effect in terms of survival from HF. This finding supports the popular argument that the ``obesity paradox" is actually an epidemiological artifact \citep{elagizi2018overview}.

}

\subsection{A TL example: the adverse effect of statins among Africans}

{\darkred 
Coronary artery disease (CAD) is a prevalent disease that affects the functioning of the heart and is the leading cause of death worldwide. Statins are commonly prescribed drugs that reduce low-density lipoprotein (LDL) levels, subsequently lowering CAD risks through {\em HMGCR} inhibition \citep{nissen2005statin}. However, the use of statins is associated with an increased risk of new-onset type II diabetes (T2D). Previous studies have examined the potential side effects of statins in developing T2D \citep{waters2013cardiovascular,macedo2014statins}; however, there is still no sufficient and robust evidence as to whether and on what kind of population statin use increases the risk of T2D. In this example, our goal is to test the effect of statin use on T2D risk among the African (AFR) cohort. In such cases, one useful strategy is leveraging the larger European (EUR) data set to assist the analysis of the target AFR data in the belief that genetic models ($Y\sim Z$) of the two ethnic groups are similar \citep[e.g.]{cai2022semi}.

While randomized control trials can be expensive and sometimes unethical, and observational studies based on medical records may encounter unmeasured confounding bias, we take an alternative route to study this problem based on UK Biobank (UKB) data that links the T2D phenotype with genomic profiles. In specific, we use the genetic variant {\em rs12916}-T as a surrogate variable for statin use. It serves as a treatment indicator: if a subject carries {\em rs12916}-T, then set the exposure $X=1$, and if they do not carry it, then $X=0$. This variant can be used as a trustworthy substitute treatment variable for statin use because it is located in the {\em HMGCR} gene, which encodes the drug target of statins, and has been shown to be an unbiased and reliable proxy for the pharmacological action of statins on their target, HMG-CoA reductase inhibition \citep{swerdlow2015hmg,wurtz2016metabolomic}. To be more specific, \cite{wurtz2016metabolomic} demonstrates a strong similarity between the metabolic changes resulting from statin use, such as lowered LDL cholesterol levels, and those associated with {\em rs12916}-T, with an R-square of $0.94$. Also note that such a strategy, i.e., using some functional genetic variants as proxies for certain pharmacological actions, has been frequently adopted in biomedical studies \citep[e.g.]{interleukin2012interleukin,liu2021integrative,guo2022assessing}.

For all the EUR and AFR subjects in UKB data, we extract their statin proxy variant {\em rs12916}-T, as well as $p=355$ adjustment features $Z$ including age, gender, and $353$ genetic variants associated with T2D or its related phenotypes including high LDL, high–density lipoprotein (HDL) and body mass index (BMI). Response $Y$ is chosen as the status of T2D diagnosed by doctors as a medical condition (either self-reported or with diagnostic codes). Similar to the TL scenario introduced in Section \ref{sec:tl}, we take the $n=3,345$ AFR subjects as the target data set $\bD=(\by,\bx,\bZ)$ used for the CRT. We also use the same AFR samples for X-modeling. Meanwhile, we incorporate a large source data $(\by^e,\bZ^e)$ consisting of $446,531$ EUR subjects with the same set of variables, to provide external knowledge about $g(\Z)$ as described in Algorithm \ref{alg:outsample:cart:TL}. Basic information about the data sets is summarized in Table \ref{tab:info:ukb}. Note that since UKB is a typical cohort not associated with any specific diseases, our data has a low $\mathbb{P}(Y=1)$. Thus, compared to the seemingly large total sample sizes, the case numbers ($433$ on AFR; $27,433$ on EUR) may better reflect the amounts of statistically effective information. In this sense, although using $\bD$, the same data as the CRT, for X-modeling, we still expect $X\mid Z$ to be more effectively estimated than $Y$'s model estimated using $\bD$ because the number of exposed ($X=1$) subjects is significantly larger than the case number.

}

We include the TL Maxway CRT (i.e., Algorithm \ref{alg:outsample:cart:TL}), the model-X CRT, and the model-X CPT to test $X\indp Y\mid Z$ on AFR. For implementation, we fit logistic lasso to estimate both $X\mid Z$ and $Y\mid Z$ in all approaches, and low-dimensional logistic regression to learn the Maxway distribution as described in Implementation example \ref{example:2}. For the importance measure, we use the d$_0$ test statistic. For all approaches, we resample $\x$ for $M=2000$ times to estimate the $p$-values.

{\darkred 

\begin{table}
\caption{\label{tab:info:ukb} Summary information of the data sets used in our real example of TL, studying the adverse effect of statins among the AFR cohort from UKB.}
\centering
 \begin{tabular}{cccc}
    \hline
   & CRT data & X-modeling data & External data \\
    \hline
    \\[-2.5ex] 
   Structure & $(\by,\bx,\bZ)$ & $(\bx^u,\bZ^u)=(\bx,\bZ)$ &  $(\by^e,\bZ^e)$ \\
   Sample size  & 3,345 & 3,345 & 446,531  \\
  \# of Cases ($Y=1$) & 433 & -- &  27,433  \\
  \# of Exposed ($X=1$) & 1,357 & 1,357 &  -- \\
   \\[-2.5ex] 
      \hline
    \end{tabular}
\end{table}

\begin{table}
\caption{\label{table:pval} $p$-values of the d$_0$ statistic for the association between the risk of T2D and the statin variant {\em rs12916}-T conditional on other T2D related gene variants among the AFR subjects.}
    \centering
 \begin{tabular}{ccccc}
    \hline
   & Model-X CRT & Model-X CPT & Maxway CRT \\
    \hline
    \\[-2.5ex] 
   $p$-value & $0.033$ & $0.038$ &  $0.043$  \\
   \\[-2.5ex] 
      \hline
    \end{tabular}
\end{table}

The output $p$-values of the three approaches are presented in Table \ref{table:pval}. While the model-X CRT produces the smallest $p$-value and our method produces the largest one, the $p$-values of the three methods are not that far from each other and all lead to the decision of rejecting the null hypothesis when the nominal level is $0.05$. We find that the observed partial covariance between $Y$ and $X$ (adjusted to $Z$) is positive, indicating that the presence of {\em rs12916}-T, the functional SNP of statins, significantly increases the risk of T2D among the AFR subjects. Thus, from a biological perspective and focusing on the AFR cohort, our study supports and complements findings in existing clinical studies according to which statins tend to increase the risk of new-onset T2D \citep{waters2013cardiovascular,carter2013risk,macedo2014statins,mansi2015statins}. 

Finally, we notice an interesting fact that in the low-dimensional logistic regression of $Y$ against $X$, age and gender on the data $\bD$, the $p$-value for the effect of $X$ turns out to be $0.030$, which is the closest to the output of the model-X CRT and farthest from the Maxway CRT. This result may indicate that compared with the model-X CRT, our method actually has a more adequate adjustment to the high-dimensional genetic features.

}

\section{Discussion}\label{sec:discuss}
{\darkred
\paragraph{Power of the Maxway CRT.} While our studies mainly focus on robustness, power is another important aspect. We shall remark on the power of our method based on its connection with the model-X distilled CRT (dCRT). Note that when knowledge of $X\mid Z$ is perfect, i.e., $X\indp Z\mid h(Z)$ and $\rho=\rho^\star$, our Algorithm \ref{alg:insample:cart}, the Maxway$_{\rm in}$ CRT is actually {\em equivalent} to the dCRT with the ``distillation" procedure $g(\Z)=\Lsc_g((\by,\bZ);\Z)$ in \cite{liu2020fast}. Meanwhile, under this perfect model-X scenario, Algorithm \ref{alg:outsample:cart}, the Maxway$_{\rm out}$ CRT is essentially the same as \cite{katsevich2020theoretical}'s version of the dCRT with an out-of-sample $g(\cdot)$. Thus, their power analysis of the d$_0$CRT, a natural specification of the dCRT against local ($n^{-1/2}$-rate) semiparametric alternatives 
can be directly applied to our method. Their results imply that the Maxway$_{\rm out}$ CRT with a d$_0$ construction achieve an essential power against local alternatives; see Theorem 4.1 and Section 4 of \cite{katsevich2020theoretical}. Further, when the machine learning estimator of $\EE{Y\mid Z}$ is consistent, our method can be shown to achieve local efficiency. We also notice that a similar analysis in \cite{wang2020power} could be applied on the Maxway$_{\rm in}$ CRT constructed with lasso. In addition, as is shown in extensive simulation studies in \cite{liu2020fast} and this work, beyond the partial linear model and d$_0$ statistic, the model-X and Maxway approaches also have essential power against various types of nonlinear and  hierarchical interaction alternative.

Though our simulation studies show some power discrepancy between our method and the model-X (d)CRT at the first glance, we find that this only reflects the shrinkage bias issue of the model-X CRT shifting the power curve rather than contradicting the power equivalence between the model-X dCRT and the Maxway CRT. We demonstrate this point through additional derivation and simulation results in Appendix \ref{sec:app:sym}. Interestingly, after a simple adjustment on such shifting bias, the model-X and Maxway approaches produce nearly the same power across all signals. This is coherent with our above discussion.

}

\paragraph{An even more robust CRT.} In this paper, we have proposed the Maxway CRT, which improves upon the type-I error inflation compared to the original model-X CRT. As the name ``model and adjust X with the assistance of Y" suggests, the role of $X$ and $Y$ is not symmetric in the Maxway CRT. In particular, the Maxway CRT requires full knowledge of the distribution of $X \mid Z$, whereas it only requires a sufficient statistic of the distribution of $Y \mid Z$, which is used to ``adjust X".  As discussed in Section \ref{sec:sa:ssl}, surrogate datasets can be used to obtain such sufficient statistics, but they may not be good enough to train the full  distribution of $Y \mid Z$. In other application scenarios where full knowledge of the distribution of $Y \mid Z$ is available, we can further enhance the robustness of the CRT. Consider the following procedure: Run the Maxway CRT and obtain a $\pval$ $p_{\operatorname{maxway},1}$. Swap $\x$ and $\y$, run the Maxway CRT again on the swapped dataset, and obtain a $\pval$ $p_{\operatorname{maxway},2}$. Finally, report the maximum of the two $\pval$s, i.e., take $p_{\textnormal{model-xy} } = \max\cb{p_{\operatorname{maxway},1},  p_{\operatorname{maxway},2}}$. It is not hard to verify that under the same conditions as in Theorem \ref{theo:almost_double_robust}, the type-I error of this new procedure can be bounded by 
\[
\PP{p_{\textnormal{model-xy} }  \leq \alpha} \leq \alpha + 2\EE{\Delta_x \Delta_y} + \min\cb{\EE{\Delta_{x|g,h}}, \EE{\Delta_{y|g,h}}},
\] 
where $\Delta_{x|g,h}$ is the total variation distance between the true distribution of $\x \mid g(\Z), h(\Z)$ and the estimated one, and the same is true for $\Delta_{y|g,h}$ and $\y$. Unlike the Maxway CRT, this procedure is ``truly" doubly robust, since either perfect information of the distribution of $X\mid Z$ or that of $Y\mid Z$ can ensure the exact type-I error control. Nevertheless, the new procedure is strictly more conservative than the Maxway CRT, which might decrease its power. We leave the further investigation of this procedure to future work. Questions of interest include: What are the suitable application scenarios? How would this procedure compare to the Maxway CRT in practice in terms of power and type-I error? 

\paragraph{Extension to Knockoffs.}
Another interesting direction for future work is in extending our proposed method to other model-X procedures, including the model-X knockoffs. When the X-modeling is not perfect, \citet{barber2020robust} quantify the inflation in FDR of the model-X knockoffs in terms of a distance between the true distribution and the sampling distribution.
It will be of interest to study whether some knowledge of the distribution of $Y \mid X$ can help decrease the inflation in FDR for knockoffs. In this paper, we gain extra robustness in type-I error control by conditioning on a sufficient statistic of $Y \mid Z$. The question is, can we construct knockoffs conditioning on similar statistics and thus obtain better bounds for FDR?

\paragraph{Maxway distribution.}
When introducing the Maxway CRT in Algorithm \ref{alg:maxway}, we formulate the Maxway distribution $\rho$ as an estimator of $\rhos$, where $\rhos$ is taken to be the conditional distribution of $\x$ given $g(\Z)$ and $h(\Z)$. This is not the only possible formulation of $\rhos$. In fact, let $\rhos$ be a distribution such that for any $\tilde{\x} \sim \rhos$ independent from $\x$, the following conditions hold:
\begin{enumerate}
    \item[(1)] For any $h$ and $g$, $\tilde{\x}$ is exchangeable with the observed $\x$ conditional on $g(\Z)$.
    \item[(2)] If $\x\indp \Z\mid h(\Z)$, then for any $g$, $\tilde{\x}$ is exchangeable with $\x$ conditional on $g(\Z),h(\Z)$.
\end{enumerate}
We can show that the Maxway CRT with this new definition of $\rhos$ still enjoys the almost double robustness property given by Theorems \ref{theo:exact_inference} and \ref{theo:almost_double_robust}. Obviously, the conditional distribution of $\x$ given $g(\Z)$ and $h(\Z)$ is just one of the distributions satisfying these conditions and it is chosen in our framework due to its intuitive interpretation and implementation in practice. This more general definition of the Maxway distribution given by Conditions (1) and (2) additionally reveals that $g(\Z)$ and $h(\Z)$ are not symmetric in our framework and may not be treated in the same way when constructing the Maxway distribution. For example, one could just set $g(\Z)$ as the ``predictors" and $h(\Z)$ as the ``offsets" when learning the Maxway distribution, which further reduces the statistical dimensionality of the regression procedure.


\bibliographystyle{apalike}
\bibliography{library}

\clearpage
\newpage
\setcounter{page}{1}

\appendix
\setcounter{lemma}{0}
\setcounter{equation}{0}
\setcounter{theorem}{0}
\setcounter{figure}{0}
\setcounter{table}{0}
\renewcommand{\thefigure}{A\arabic{figure}}
\renewcommand{\thetable}{A\arabic{table}}
\renewcommand{\theequation}{A\arabic{equation}}
\renewcommand{\thelemma}{A\arabic{lemma}}
\renewcommand{\thetheorem}{A\arabic{theorem}}

\setcounter{definition}{0}
\renewcommand{\thedefinition}{A\arabic{definition}}

\section*{Appendix}

\section{Proofs}
\subsection{Two useful lemmas}
\begin{lemma}
\label{lemma:conditional_TV}
Let $U \in \RR^k$,$V\in \RR^k$ and $W \in \RR^l$ be three random vectors. Then
\begin{equation}
    d_{\operatorname{TV}}\p{(U,W), (V,W)}
    \leq \EE{d_{\operatorname{TV}} \p{U \mid W, V \mid W}}. 
\end{equation}
\end{lemma}
\begin{proof}
Recall that by definition of the total variation distance, we have that
\begin{equation}
\begin{split}
 d_{\operatorname{TV}}\p{(U,W), (V,W)}
 &= \sup_{A \in B(\mathbb{R}^{k + l})}  \abs{\PP{(U,W) \in A} - \PP{(V,W) \in A} }, \textnormal{ and}\\
d_{\operatorname{TV}} \p{U \mid W, V \mid W}
& = \sup_{A \in B(\mathbb{R}^{k})}  \abs{\PP{U \in A \mid W} - \PP{V \in A \mid W} },
\end{split}
\end{equation}
where $B(\mathbb{R}^{k + l})$ and $B(\mathbb{R}^{k })$ are the Borel $\sigma$-algebra on $\mathbb{R}^{k + l}$ and $\mathbb{R}^{k}$ respectively.  
For any fixed $A \in B(\mathbb{R}^{k + l})$, let $A_w = \cb{u \in \mathbb{R}^k: (u,w) \in A}$. 
Then,
\begin{equation}
\begin{split}
\PP{(U,W) \in A} - \PP{(V,W) \in A}
&= \EE{\PP{U \in A_W \mid W} - \PP{V \in A_W \mid W}}\\
& \leq \EE{\sup_{A \in B(\mathbb{R}^{k})}  \abs{\PP{U \in A \mid W} - \PP{V \in A \mid W}}}\\
& = \EE{d_{\operatorname{TV}} \p{U \mid W, V \mid W}}. 
\end{split}
\end{equation}
Similarly, by swapping $U$ and $V$, we can get that
\begin{equation}
    \PP{(V,W) \in A} - \PP{(U,W) \in A} \leq \EE{d_{\operatorname{TV}} \p{U \mid W, V \mid W}}. 
\end{equation}
Combining the two, we get that
\begin{equation}
\abs{\PP{(V,W) \in A} - \PP{(U,W) \in A}} \leq \EE{d_{\operatorname{TV}} \p{U \mid W, V \mid W}},
\end{equation}
and thus
\begin{equation}
d_{\operatorname{TV}}\p{(U,W), (V,W)} \leq \EE{d_{\operatorname{TV}} \p{U \mid W, V \mid W}}.
\end{equation}

\end{proof}

\begin{lemma}
\label{lemma:M_infinite_finite}
Let $M \in \mathbb{N}_+$ be a fixed positive integer.  Let $A \in [0,1]$ be a random variable. Assume that there exists a constant $\delta \geq 0$, such that $\PP{A \leq \alpha} \leq \alpha + \delta$ for any $\alpha \in (0,1)$. Let $B$ be a random variable such that conditional on $A$, we have that $B \sim \operatorname{Binom}(M, A)$. Write $R = \frac{1 + B}{1+ M}$. Then 
\begin{equation}
\PP{R \leq \alpha} \leq \alpha + \delta. 
\end{equation}
\end{lemma}

\begin{proof}
We start by noting the $R \leq \alpha$ is equivalent to $B \leq \lfloor \alpha(M+1) \rfloor - 1$. Thus, 
\[\PP{R \leq \alpha \mid A} = \sum_{k = 0}^{\lfloor \alpha(M+1) \rfloor - 1} \binom{M}{k} A^{k} (1-A)^{M-k}.  \]
Therefore
\[
\begin{split}
\PP{R \leq \alpha }
&= \EE{\PP{R \leq \alpha \mid A}}
= \EE{ \sum_{k = 0}^{\lfloor \alpha(M+1) \rfloor - 1} \binom{M}{k} A^{k} (1-A)^{M-k}}\\
&= \int_{0}^1  \frac{d}{da}\p{ \sum_{k = 0}^{\lfloor \alpha(M+1) \rfloor - 1} \binom{M}{k}  a^{k} (1-a)^{M-k}} (1 - F_A(a)) da,
\end{split}
\]
where $F_A(\cdot)$ is the cumulative distribution function of $A$. Note that derivative is always non-positive, because for any $0< a_1 < a_2 < 1$, 
\[
\PP{\operatorname{Binom}(M, a_1) \leq \lfloor \alpha(M+1) \rfloor - 1} \geq 
\PP{\operatorname{Binom}(M, a_2) \leq \lfloor \alpha(M+1) \rfloor - 1}. 
 \]
Note also that we have $F_A(a) \leq a + \delta$. 
Thus 
\[
\begin{split}
\PP{R \leq \alpha }
&\leq \int_{0}^1  \frac{d}{da}\p{ \sum_{k = 0}^{\lfloor \alpha(M+1) \rfloor - 1} \binom{M}{k}  a^{k} (1-a)^{M-k}} (1 - a - \delta) da\\
&\leq \int_{0}^1  \frac{d}{da}\p{ \sum_{k = 0}^{\lfloor \alpha(M+1) \rfloor - 1} \binom{M}{k}  a^{k} (1-a)^{M-k}} (1 - a) da\\
& \qquad \qquad \qquad \qquad  + \int_{0}^1  \frac{d}{da}\p{ \sum_{k = 0}^{\lfloor \alpha(M+1) \rfloor - 1} \binom{M}{k}  a^{k} (1-a)^{M-k}} \delta da\\
& = \int_{0}^1   \sum_{k = 0}^{\lfloor \alpha(M+1) \rfloor - 1} \binom{M}{k}  a^{k} (1-a)^{M-k} da +  \delta \sqb{\sum_{k = 0}^{\lfloor \alpha(M+1) \rfloor - 1} \binom{M}{k}  a^{k} (1-a)^{M-k}}_{a \to 0}^{a \to 1}. 
\end{split}
\]

For the first term, note that 
\[
\int_{0}^1 a^k (1 - a)^{M-k} = \frac{k! (M-k)!}{(M+1)!}.
\]
Thus 
\[\int_{0}^1   \sum_{k = 0}^{\lfloor \alpha(M+1) \rfloor - 1} \binom{M}{k}  a^{k} (1-a)^{M-k} da
=  \sum_{k = 0}^{\lfloor \alpha(M+1) \rfloor - 1}  \frac{1}{M+1}
= \frac{\lfloor \alpha(M+1) \rfloor}{M+1} \leq \alpha. 
\]

For the second term, we can easily verify that $\lim_{a \to 1} \sum_{k = 0}^{\lfloor \alpha(M+1) \rfloor - 1} \binom{M}{k}  a^{k} (1-a)^{M-k} = 1$ and $\lim_{a \to 0} \sum_{k = 0}^{\lfloor \alpha(M+1) \rfloor - 1} \binom{M}{k}  a^{k} (1-a)^{M-k} = 0$. Thus,
\[ \delta \sqb{\sum_{k = 0}^{\lfloor \alpha(M+1) \rfloor - 1} \binom{M}{k}  a^{k} (1-a)^{M-k}}_{a \to 0}^{a \to 1} = \delta. \]

Combining the two terms, we get 
\[ \PP{R \leq \alpha } \leq \alpha + \delta. \]
\end{proof}

\subsection{Proof of Theorem \ref{theo:exact_inference}}
We start with (i). If $\x,  \x^{1}, \dots, \x^{(M)}$ are exchangeable conditioning on $\Z$, then $\x,  \x^{1}, \dots, \x^{(M)}$ are also exchangeable conditioning on $\Z$ and $\y$. Therefore, the statistics $T(\y, \x, g(\Z), h(\Z))$, $\dots$, $T(\y,\x^{(M)},g(\Z), h(\Z))$ will be exchangeable. As a consequence, $\sum_{m=1}^M \mathbf{1}\Big\{T(\y,\x\supm,g(\Z), h(\Z))\geq T(\y,\x, g(\Z), h(\Z) )\Big\}$ is stochastically no smaller than $\operatorname{Unif}\cb{0, \dots, M}$, and thus the Maxway CRT $p$-value defined in \eqref{eqn:maxway_CRT} is stochastically no smaller than $\operatorname{Unif}[0,1]$.

\sloppy{For (ii), we start by noting that $\Z \indp \y \mid g(\Z)$ together with $\x \indp \y \mid \Z$ implies that $\x \indp \y \mid (g(\Z), h(\Z))$. Thus $\x,  \x^{1}, \dots, \x^{(M)}$ are exchangeable conditioning on $g(\Z)$, $h(\Z)$ and $\y$. Therefore, similar to what we have seen in (i), we have that the statistic $T(\y, \x, g(\Z), h(\Z))$, $\dots$, $T(\y,\x^{(M)},g(\Z), h(\Z))$ will be exchangeable. Thus $\sum_{m=1}^M \mathbf{1}\Big\{T(\y,\x\supm,g(\Z), h(\Z))\geq T(\y,\x, g(\Z), h(\Z) )\Big\}$ is stochastically no smaller than $\operatorname{Unif}\cb{0, \dots, M}$, and hence the Maxway CRT $p$-value defined in \eqref{eqn:maxway_CRT} is stochastically no smaller than $\operatorname{Unif}[0,1]$. }

{\darkred
\subsection{Proof of Proposition \ref{theo:main0}}

We follow the proof idea of Theorem 4 in \citet{berrett2018conditional}. Let $\tilde{\x}$ be an additional copy drawn also from $\rho^{n}( \cdot \mid g(\Z), h(\Z))$ independently of $Y$ and of $\x, \x^{(1)}, \dots, \x^{(M)}$. Then, since conditional on $g(\Z), h(\Z)$, the copies $\tilde{\x}, \x^{(1)}, \dots, \x^{(M)}$ are independent and are independent of $\x$ and $\y$, we have 
\begin{equation}
\begin{split}
&d_{\operatorname{TV}}\Big(\big( (\x, \y, \x^{(1)}, \dots, \x^{(M)}) \mid g(\Z), h(\Z) \big), \big( (\tilde{\x}, \y, \x^{(1)}, \dots, \x^{(M)}) \mid g(\Z), h(\Z) \big)  \Big)\\
& \qquad = d_{\operatorname{TV}}\Big(\big( (\x, \y) \mid g(\Z), h(\Z) \big), \big( (\tilde{\x}, \y) \mid g(\Z), h(\Z) \big)  \Big)\\
& \qquad = d_{\operatorname{TV}} \Big(p( \cdot \mid g(\Z), h(\Z)), q( \cdot \mid g(\Z), h(\Z)) \Big).
\end{split}
\end{equation}

Now let $A_{\alpha}$ be defined as
\begin{equation}
A_{\alpha}:=\left\{\left(\mathbf{x}, \mathbf{y}, \mathbf{x}^{(1)}, \ldots, \mathbf{x}^{(M)}\right): \frac{1+\sum_{m=1}^{M} \mathbf{1}\left\{T(\mathbf{y}, \mathbf{x}^{(m)}, g(\Z), h(\Z)) \geq T(\mathbf{y}, \mathbf{x}, g(\Z), h(\Z))\right\}}{1+M} \leq \alpha\right\}, 
\end{equation}
i.e., the set where we would obtain a $p$-value $p_{\maxway} \leq \alpha$. Then
\begin{equation}
\begin{split}
&\PP{p_{\maxway} \leq \alpha \mid g(\Z), h(\Z)}
 = \PP{(\x, \y, \x^{(1)}, \dots, \x^{(M)}) \in A_{\alpha} \mid g(\Z), h(\Z)}\\
&\leq \PP{(\tilde{\x}, \y, \x^{(1)}, \dots, \x^{(M)}) \in A_{\alpha} \mid g(\Z), h(\Z)}\\
&\qquad  + d_{\operatorname{TV}}\Big(\big( (\x, \y, \x^{(1)}, \dots, \x^{(M)}) \mid g(\Z), h(\Z) \big), \big( (\tilde{\x}, \y, \x^{(1)}, \dots, \x^{(M)}) \mid g(\Z), h(\Z) \big)  \Big)\\
&= \PP{(\tilde{\x}, \y, \x^{(1)}, \dots, \x^{(M)}) \in A_{\alpha} \mid g(\Z), h(\Z)}\\
&\qquad \qquad \qquad \qquad + d_{\operatorname{TV}} \Big(p( \cdot \mid g(\Z), h(\Z)), q( \cdot \mid g(\Z), h(\Z)) \Big). 
\end{split}
\end{equation}
Finally, since $\tilde{\x}, \x^{(1)}, \ldots, \x^{(M)}$ are clearly i.i.d. after conditioning on $\y, g(\Z), h(\Z)$, and are therefore exchangeable, by definition of $A_{\alpha}$ we must have
\begin{equation}
\PP{(\tilde{\x}, \y, \x^{(1)}, \dots, \x^{(M)}) \in A_{\alpha} \mid g(\Z), h(\Z)} \leq \alpha,
\end{equation}
proving the desired bound.

\subsection{Proof of Theorem \ref{theo:almost_double_robust}}

We will show that the bound in \eqref{eqn:mayway_bound0} is no larger than $2 \Delta_x \Delta_y + \Delta_{x|g,h}$. We start with upper bounding the bound in \eqref{eqn:mayway_bound0}. Let $q^{\star}( \cdot \mid g(\Z), h(\Z))
= \rho^{\star n}( \cdot \mid g(\Z), h(\Z)) f_{\y \mid g(\Z), h(\Z)} ( \cdot \mid g(\Z), h(\Z)) = f_{\x \mid g(\Z), h(\Z)}( \cdot \mid g(\Z), h(\Z)) f_{\y \mid g(\Z), h(\Z)} ( \cdot \mid g(\Z), h(\Z))$. In words, $q^{\star}$ is the distribution of sampling $\x$ and $\y$ independently from their distribution conditional on $g(\Z)$ and $h(\Z)$. Then we have that
\begin{equation}
\label{equation:tv_bound1}
\begin{split}
&d_{\operatorname{TV}} \Big(p( \cdot \mid g(\Z), h(\Z)), q( \cdot \mid g(\Z), h(\Z)) \Big)\\
& \qquad \qquad \leq d_{\operatorname{TV}} \Big(p( \cdot \mid g(\Z), h(\Z)), q^{\star}( \cdot \mid g(\Z), h(\Z)) \Big) \\
& \qquad \qquad \qquad \qquad \qquad \qquad+ d_{\operatorname{TV}} \Big(q^{\star}( \cdot \mid g(\Z), h(\Z)), q( \cdot \mid g(\Z), h(\Z)) \Big)\\
& \qquad \qquad = d_{\operatorname{TV}} \Big(p( \cdot \mid g(\Z), h(\Z)), q^{\star}( \cdot \mid g(\Z), h(\Z)) \Big) \\
& \qquad \qquad \qquad \qquad \qquad \qquad + d_{\operatorname{TV}} \Big(\rho^{\star  n}( \cdot \mid g(\Z), h(\Z)), \rho^{n}( \cdot \mid g(\Z), h(\Z)) \Big)\\
& \qquad \qquad = d_{\operatorname{TV}}\Big(p( \cdot \mid g(\Z), h(\Z)), q^{\star}( \cdot \mid g(\Z), h(\Z)) \Big) + \Delta_{x|g,h}. 
\end{split}
\end{equation}

We will then focus on studying $d_{\operatorname{TV}} \Big(p( \cdot \mid g(\Z), h(\Z)), q^{\star}( \cdot \mid g(\Z), h(\Z)) \Big)$. To this end, we switch notation and define $\x, \y, \check{\x}$ and $\check{\y}$ to be four random vectors sampled independently conditioning on $\Z$ from the following distributions respectively: 
\begin{equation}
\begin{split}
    &\x \sim f_{\x \mid \Z}(\cdot \mid \Z), \quad 
    \check{\x} \sim f_{\x \mid g(\Z), h(\Z)}(\cdot \mid  g(\Z), h(\Z)), \\
    &\y \sim f_{\y \mid \Z}(\cdot \mid \Z), \quad 
    \check{\y} \sim f_{\y \mid g(\Z), h(\Z)}(\cdot \mid  g(\Z), h(\Z)). 
\end{split}
\end{equation}
Then we can write the total variation distance as
\begin{equation}
\label{equation:tv_bound2}
    \begin{split}
 & d_{\operatorname{TV}} \Big(p( \cdot \mid g(\Z), h(\Z)), q^{\star}( \cdot \mid g(\Z), h(\Z)) \Big)\\
 & \qquad \qquad = \sup_{A \in B(\mathbb{R}^{2n})} \abs{\PP{(\x,\y) \in A \mid g(\Z), h(\Z)} - \PP{(\check{\x},\check{\y}) \in A \mid g(\Z), h(\Z)}},
    \end{split}
\end{equation}
where $B(\mathbb{R}^{2n})$ is the Borel $\sigma$-algebra on $\mathbb{R}^{2n}$. We will then study properties of $\PP{(\check{\x},\check{\y}) \in A \mid g(\Z), h(\Z)}$. To this end, 
for any fixed $\mathbf{x}_0 \in \mathbb{R}^n$, let $q_A(\mathbf{x}_0, g(\Z), h(\Z)) = \PP{(\mathbf{x}_0, \y) \in A \mid g(\Z), h(\Z)}$, and let 
$\check{q}_A(\mathbf{x}_0, g(\Z), h(\Z)) = \PP{(\mathbf{x}_0, \check{\y}) \in A \mid g(\Z), h(\Z)}$. 
Since $\check{\y} \sim f_{\y \mid g(\Z), h(\Z)}(\cdot \mid  g(\Z), h(\Z))$, we have that $q_A(\mathbf{x}_0, g(\Z), h(\Z)) = \check{q}_A(\mathbf{x}_0, g(\Z), h(\Z))$. This then implies that
\begin{equation}
\begin{split}
\PP{(\check{\x},\check{\y}) \in A \mid g(\Z), h(\Z)} 
&= \EE{\PP{(\check{\x},\check{\y}) \in A \mid g(\Z), h(\Z), \check{\x}} \mid g(\Z), h(\Z)}\\
&  = \EE{\check{q}_A(\check{\x}, g(\Z), h(\Z)) \mid g(\Z), h(\Z) } \qquad (\star)\\
& = \EE{q_A(\check{\x}, g(\Z), h(\Z)) \mid g(\Z), h(\Z)}\\
& = \EE{\PP{(\check{\x},\y) \in A \mid g(\Z), h(\Z), \check{\x}} \mid g(\Z), h(\Z)} \qquad (\ast)\\
&= \PP{(\check{\x},\y) \in A \mid g(\Z), h(\Z)}. 
\end{split}
\end{equation}
Here line $(\star)$ follows from $\check{\x} \indp \check{\y} \mid g(\Z), h(\Z)$ and line $(\ast)$ follows from $\check{\x} \indp \y \mid g(\Z), h(\Z)$. We can also conduct the same analysis with $\x$ and $\y$ flipped and get
\begin{equation}
\PP{(\check{\x},\check{\y}) \in A \mid g(\Z), h(\Z)} = 
\PP{(\x,\check{\y}) \in A \mid g(\Z), h(\Z)} = 
\PP{(\check{\x},\y) \in A \mid g(\Z), h(\Z)}. 
\end{equation}

Therefore,
\begin{equation}
\label{eqn:decom0}
\begin{split}
   & \PP{(\x,\y) \in A \mid g(\Z), h(\Z)} - \PP{(\check{\x},\check{\y}) \in A \mid g(\Z), h(\Z)}\\
   & \qquad \qquad = \PP{(\x,\y) \in A \mid g(\Z), h(\Z)} - \PP{(\x,\check{\y}) \in A \mid g(\Z), h(\Z)}\\
   & \qquad \qquad \qquad \qquad + \PP{(\check{\x},\y) \in A \mid g(\Z), h(\Z)} - \PP{(\check{\x},\check{\y}) \in A \mid g(\Z), h(\Z)}. 
   \end{split}
\end{equation}
For any fixed $\mathbf{x}_0 \in \mathbb{R}^n$, let $\pi_A(\mathbf{x}_0, \Z) = \PP{(\mathbf{x}_0, \y) \in A \mid \Z}$ and let $\check{\pi}_A(\mathbf{x}_0, \Z) = \PP{(\mathbf{x}_0, \check{\y}) \in A \mid \Z}$. 
Then we can rewrite \eqref{eqn:decom0} as
\begin{equation}
\begin{split}
& \PP{(\x,\y) \in A \mid g(\Z), h(\Z)} - \PP{(\check{\x},\check{\y}) \in A \mid g(\Z), h(\Z)}\\
& \qquad = \PP{(\x,\y) \in A \mid g(\Z), h(\Z)} - \PP{(\x,\check{\y}) \in A \mid g(\Z), h(\Z)}\\
& \qquad \qquad + \PP{(\check{\x},\y) \in A \mid g(\Z), h(\Z)} - \PP{(\check{\x},\check{\y}) \in A \mid g(\Z), h(\Z)}\\
& \qquad = \EE{\PP{(\x,\y) \in A\mid \x, \Z}\mid g(\Z), h(\Z)}  - \EE{\PP{(\x,\check{\y}) \in A\mid \x, \Z}\mid g(\Z), h(\Z)}\\
&  \qquad \qquad + \EE{\PP{(\check{\x},\y) \in A\mid \check{\x},  \Z}\mid g(\Z), h(\Z)}  - \EE{\PP{(\check{\x},\check{\y}) \in A\mid \check{\x}, \Z}\mid g(\Z), h(\Z)}\\
& \qquad = \EE{\pi_A(\x, \Z)\mid g(\Z), h(\Z)} - \EE{\check{\pi}_A(\x, \Z) \mid g(\Z), h(\Z)}\\
&  \qquad \qquad+ \EE{\pi_A(\check{\x}, \Z) \mid g(\Z), h(\Z)} - \EE{\check{\pi}_A(\check{\x}, \Z)\mid g(\Z), h(\Z)},
   \end{split}
\end{equation}
where the last line follows from the fact that $\x, \y, \check{\x}, \check{\y}$ are independent conditioning on $\Z$. Simplifying the above equations, we get that
\begin{equation}
\label{eqn:diff_diff_theta}
\begin{split}
&\PP{(\x,\y) \in A \mid g(\Z), h(\Z)} - \PP{(\check{\x},\check{\y}) \in A \mid g(\Z), h(\Z)}\\
& \qquad = \EE{\pi_A(\x, \Z) - \check{\pi}_A(\x, \Z) \mid g(\Z), h(\Z)}
- \EE{\pi_A(\check{\x}, \Z) - \check{\pi}_A(\check{\x}, \Z) \mid g(\Z), h(\Z)}. 
\end{split}
\end{equation}
We note that the difference between $\pi_A$ and $\check{\pi}_A$ can be bounded by the total variation distance:
\begin{equation}
\begin{split}
    \abs{\pi_A(\mathbf{x}_0, \Z) - \check{\pi}_A(\mathbf{x}_0, \Z)}
    &\leq d_{\operatorname{TV}}(\y \mid \Z, \check{\y} \mid \Z)\\
    &= d_{\operatorname{TV}}\p{f_{\y \mid \Z} ( \cdot \mid \Z), f_{\y \mid g(\Z), h(\Z)} ( \cdot \mid g(\Z), h(\Z))}
    = \Delta_y. 
\end{split}
\end{equation}
Therefore, \eqref{eqn:diff_diff_theta} is essentially the difference of expectations of a bounded function under different distributions. 
Hence, by properties of the total variations distance, we have that
\begin{equation}
\label{equation:tv_bound3}
\begin{split}
    &\abs{\PP{(\x,\y) \in A \mid g(\Z), h(\Z)} - \PP{(\check{\x},\check{\y}) \in A \mid g(\Z), h(\Z)}}\\
    & \qquad \qquad\qquad \qquad\qquad \qquad \leq 2 \Delta_y d_{\operatorname{TV}}(\x \mid \Z, \check{\x} \mid \Z)
    = 2 \Delta_y \Delta_x. 
\end{split}
\end{equation}

Finally, combining \eqref{equation:tv_bound1}, \eqref{equation:tv_bound2} and \eqref{equation:tv_bound3}, we get that 
\begin{equation}
    d_{\operatorname{TV}} \Big(p( \cdot \mid g(\Z), h(\Z)), q( \cdot \mid g(\Z), h(\Z)) \Big)
    \leq 2 \Delta_x\Delta_y + \Delta_{x|g,h}.
\end{equation}
Together with Proposition \ref{theo:main0}, the above implies that 
\begin{equation}
\PP{p_{\operatorname{maxway}}(\bD) \leq \alpha \mid g(\Z), h(\Z)} 
 \leq \alpha + 2\Delta_x\Delta_y + \Delta_{x|g,h}.
\end{equation}

\subsection{Proof of Proposition \ref{theo:lower_bound}}
\label{appendix:lower_bound_proof}
Following the proof of Theorem 5 in \citet{berrett2018conditional}, we can show that 
there exists a statistic $T$ such that 
\begin{equation}
\label{eqn:lower_eqn1}
\begin{split}
&\sup _{\alpha \in[0,1]}\p{ \PP{p_{\operatorname{maxway}}(\bD) \leq \alpha \mid \y, g(\Z), h(\Z)} -\alpha} \\
& \qquad  \geq d_{\operatorname{TV}} \Big(f_{\x \mid \y, g(\Z), h(\Z)}( \cdot \mid \y, g(\Z), h(\Z)), \rho^n( \cdot \mid g(\Z), h(\Z)) \Big)  -0.5(1+o(1)) \sqrt{\frac{\log (M)}{M}}. 
\end{split}
\end{equation}
The proof steps are exactly the same, except that we replace $\Z$ by $g(\Z), h(\Z)$. For simplicity, we omit the details here. 

In particular, in the proof, we construct a test statistic that achieves the above lower bound. This test statistic is constructed in the following way. By properties of the total variation distance, we have that there exists a set $A(\y,g(\Z), h(\Z))$ such that 
\begin{equation}
\begin{split}
&   \PP{\x \in A(\y,g(\Z), h(\Z)) \mid \y, g(\Z), h(\Z)} 
 = \PP{\tilde{\x} \in A(\y,g(\Z), h(\Z)) \mid \y,g(\Z), h(\Z)} +\\
& \qquad\qquad\qquad\qquad\qquad  d_{\operatorname{TV}} \Big(f_{\x \mid \y, g(\Z), h(\Z)}( \cdot \mid \y, g(\Z), h(\Z)), \rho^n( \cdot \mid g(\Z), h(\Z)) \Big).
\end{split}
\end{equation}
where $\tilde{\x} \sim \rho^n( \cdot \mid g(\Z), h(\Z))$ independently of $\y$. The test statistic $T$ is constructed as 
\begin{equation}
    T(\x, \y, g(\Z), h(\Z)) = \mathbbm{1} \cb{\x \in A(\y,g(\Z), h(\Z))}. 
\end{equation}

Lemma \ref{lemma:conditional_TV} further implies that
\begin{equation}
\label{eqn:lower_eqn2}
\begin{split}
& \EE{d_{\operatorname{TV}} \Big(f_{\x \mid \y, g(\Z), h(\Z)}( \cdot \mid \y, g(\Z), h(\Z)), \rho^n( \cdot \mid g(\Z), h(\Z)) \Big) \mid g(\Z), h(\Z)} \\
& \qquad \geq 
d_{\operatorname{TV}} \Big(f_{\x, \y \mid  g(\Z), h(\Z)}( \cdot \mid g(\Z), h(\Z)), \rho^n( \cdot \mid g(\Z), h(\Z)) f_{\y \mid g(\Z), h(\Z)} ( \cdot \mid g(\Z), h(\Z)) \Big) \\
&\qquad = d_{\operatorname{TV}} \Big(p( \cdot \mid g(\Z), h(\Z)), q( \cdot \mid g(\Z), h(\Z)) \Big). 
\end{split}
\end{equation}
The desired results then follows from combining \eqref{eqn:lower_eqn1} and \eqref{eqn:lower_eqn2}. 
}

\subsection{Proof of Theorem \ref{theo:inner-product_stats}}
 Before proving the results, we introduce some notations. We use the bold letter $\boeta$ and $\bep$ to denote the vector $\boeta = (\eta_1, \dots, \eta_n)$ and $\bep = (\varepsilon_1, \dots, \varepsilon_n)$. 

We start with the case where $M = \infty$. In this case, the Maxway CRT $\pval$ can be written in the following form:
\begin{equation}
\breve{p}_{\operatorname{maxway}}(\bD) = \PP[\breve \x \sim \rho^n( \cdot \mid g(\Z) , h(\Z))]{T(\y,\breve \x,g(\Z), h(\Z))\geq T(\y,\x, g(\Z), h(\Z) ) \mid \x,\y,\Z }.
\end{equation}
This is in turn equivalent to the following form:
\begin{equation}
\begin{split}
\breve{p}_{\operatorname{maxway}}(\bD) = 
\PP{\abs{\breve \x \trans \y}  \geq \abs{\x \trans \y}  \mid \x,\y,\Z }
\end{split}
\end{equation}
where $\breve \x \sim \mathcal{N}(\mu_x(\Z), 1)$. Conditioning on $\x,\y, \Z$, the term $\y \trans \breve \x $ follows a gaussian distribution with mean $\mu_x(\Z) \trans \y $ and variance $ \norm{\y}^2 $. Therefore, 
\begin{equation}
\label{eqn:inner_product_key1}
\begin{split}
\breve{p}_{\operatorname{maxway}}(\bD)  =  \PP{ \abs{\breve \x \trans \y} \geq \abs{\x \trans \y}  \mid \x,\y,\Z }  = \Phi\p{\frac{-\mu_x(\Z) \trans \y - \abs{\x \trans \y}}{\norm{\y}}} + \Phi\p{\frac{\mu_x(\Z) \trans \y - \abs{\x \trans \y}}{\norm{\y}}},
\end{split}
\end{equation}
where $\Phi$ is the cumulative distribution function of a standard gaussian distribution. 

We will then focus on the term $(\x - \mu_x(\Z)) \trans \y/\norm{\y}$, and show that this term has a distribution very close to a standard gaussian. The term $(\x - \mu_x(\Z)) \trans \y$ can be decomposed into a few terms.
 \begin{equation}
\begin{split}
\p{\x -\mu_x(\Z)}\trans \y 
 &=\p{\x - \smux(\Z) + \smux(\Z) - \tmux(\Z) + \tmux(\Z) - \mu_x(\Z)}\trans \y\\
& = \bvare \trans \y + \p{\tmux(\Z) - \mu_x(\Z)} \trans \y +  \p{\smux(\Z) - \tmux(\Z)} \trans \y\\
& = \bvare \trans \y + \p{\tmux(\Z) - \mu_x(\Z)} \trans \y +  \p{\smux(\Z) - \tmux(\Z)} \trans \boeta +\\
& \qquad \qquad  \p{\smux(\Z) - \tmux(\Z)} \trans \tmuy(\Z) + \p{\smux(\Z) - \tmux(\Z)} \trans \p{\smuy(\Z) - \tmuy(\Z)},
\end{split}
\end{equation}
where recall that $\tmuy(Z) = \EE{\smuy(Z) \mid g(Z), h(Z)}$ and $\tmux(Z) = \EE{\smux(Z) \mid g(Z), h(Z)}$. 

The first term $\bvare \trans \y /\norm{\y}$ is distributed as a standard gaussian.  Moreover, this term is distributed as a standard gaussian conditional on $\y$ and $\Z$.  We will then move on to establish that the rest of the terms are small. The absolute value of the second term $\p{\tmux(\Z) - \mu_x(\Z)} \trans \y / \norm{\y}$ can be bounded using the Cauchy-Schwarz inequality: $\abs{\p{\tmux(\Z) - \mu_x(\Z)} \trans \y }/ \norm{\y} \leq \norm{\tmux(\Z) - \mu_x(\Z)}$. Therefore, $\EE{\abs{\p{\tmux(\Z) - \mu_x(\Z)} \trans \y }/ \norm{\y}} \leq \sqrt{\EE{ \norm{\tmux(\Z) - \mu_x(\Z)}^2}} = \sqrt{n} \sqrt{\EE{\p{\tmux(Z) - \mu_x(Z)}^2}}$. 

 \sloppy{For the third term $\p{\smux(\Z) -\tmux(\Z)}\trans \boeta$, note that $\EE{\eta \mid Z} = 0$; thus
$\EE{\p{\p{\smux(\Z) -\tmux(\Z)}\trans \boeta}^2} = n \EE{(\smux(Z) - \tmux(Z))^2 \eta^2} = n \EE{(\smux(Z) - \tmux(Z))^2} \EE{\eta^2}$, since the cross terms vanish.} Thus, 
\begin{equation}
\begin{split}
\EE{\abs{\p{\smux(\Z) -\tmux(\Z)}\trans \boeta}/\norm{y}} &\leq \sqrt{\EE{\p{\p{\smux(\Z) -\tmux(\Z)}\trans \boeta}^2} \EE{1/\norm{\y}^2}}\\
& = \sqrt{n} \sqrt{\EE{(\smux(Z) - \tmux(Z))^2}} \sqrt{\EE{\eta^2}} \sqrt{ \EE{1/\norm{\y}^2}}\\
& \leq C \sqrt{\EE{(\smux(Z) - \tmux(Z))^2}},
\end{split}
\end{equation}
for some constant $C$. The last inequality follows from Lemma \ref{lemma:y_bar_inverse}. 

\sloppy{The fourth term $\p{\smux(\Z) - \tmux(\Z)} \trans \tmuy(\Z)$ can be bounded similarly as the third term. Specifically, since $\tmux(Z) = \EE{X \mid g(Z), h(Z)} = \EE{\smux(Z) \mid g(Z), h(Z)}$, we have $\EE{\smux(Z) - \tmux(Z) \mid g(Z), h(Z)} = 0$. We also note that $\tmuy$ is measurable with respect to $h$ and $g$; thus $\EE{\p{\smux(Z) - \tmux(Z)} \tmuy(Z)} = 0$. Therefore, $\EE{\p{\p{\smux(\Z) - \tmux(\Z)} \trans \tmuy(\Z)}^2} = n \EE{\p{\smux(Z) - \tmux(Z)}^2 \tmuy(Z)^2}$.} Hence, 
\begin{equation}
\begin{split}
\EE{\abs{\p{\smux(\Z) - \tmux(\Z)} \trans \tmuy(\Z)}/\norm{\y}} 
&\leq  \sqrt{n} \sqrt{\EE{\p{\smux(Z) - \tmux(Z)}^2 \tmuy(Z)^2}} \sqrt{ \EE{1/\norm{\y}^2}}\\
& \leq C\sqrt{\EE{\p{\smux(Z) - \tmux(Z)}^2 \tmuy(Z)^2}}, 
\end{split}
\end{equation}
for some constant $C$. Again, the last inequality follows from Lemma \ref{lemma:y_bar_inverse}. 

Finally, again by Lemma \ref{lemma:y_bar_inverse}, the fifth term can be bounded by
\begin{equation}
\begin{split}
&\EE{\abs{\p{\smux(\Z) - \tmux(\Z)} \trans \p{\smuy(\Z) - \tmuy(\Z)}}/\norm{\y}} \\
&\qquad = \EE{\EE{\abs{\p{\smux(\Z) - \tmux(\Z)} \trans \p{\smuy(\Z) - \tmuy(\Z)}}/\norm{\y} \mid \Z}} \\
&\qquad \leq \frac{C}{\sqrt{n}}\EE{\abs{\p{\smux(\Z) - \tmux(\Z)} \trans \p{\smuy(\Z) - \tmuy(\Z)}}} \\
& \qquad \leq \frac{C}{\sqrt{n}}
 \sqrt{\EE{\norm{\smux(\Z) - \tmux(\Z)}^2}} \sqrt{\EE{\big\|\smuy(\Z) - \tmuy(\Z) \big\|^2}} \\
& \qquad \leq C\sqrt{n} \sqrt{\EE{\norm{\smux(Z) - \tmux(Z)}^2}} \sqrt{\EE{\big\|\smuy(Z) - \tmuy(Z) \big\|^2}},
\end{split}
\end{equation}
for some constant $C$. 

\sloppy{The above analysis implies that the term $(\x - \mu_x(\Z)) \trans \y/\norm{\y}$ can be written as $(\x - \mu_x(\Z)) \trans \y/\norm{\y} = W + U$, where $W\mid \y, \Z \sim \mathcal{N}(0,1) $ and 
\begin{equation}
\label{eqn:inner_product_key2}
\EE{\abs{U}} \leq C\p{\sqrt{n} \Delta_{\rho, \operatorname{mean}} + \Delta_{x, \operatorname{mean}} +  \tilde{\Delta}_{x, \operatorname{mean}} + \sqrt{n}\Delta_{x, \operatorname{mean}} \Delta_{y, \operatorname{mean}}},
\end{equation}
for some constant $C$, where 
\begin{equation}
    \begin{split}
         &\Delta_{x, \operatorname{mean}} = \sqrt{\EE{(\smux(Z) - \tmux(Z))^2}} , \qquad 
        \tilde{\Delta}_{x, \operatorname{mean}} = \sqrt{\EE{((\smux(Z) - \tmux(Z)\tmuy(Z))^2}}\\
        &\Delta_{y, \operatorname{mean}} = \sqrt{\EE{(\smuy(Z) - \tmuy(Z))^2}}, \qquad 
        \Delta_{\rho, \operatorname{mean}} = \sqrt{\EE{(\tmux(Z) - \mu_x(Z))^2}}. 
    \end{split}
\end{equation}
Thus, $\x \trans \y / \norm{\y}$ can be expressed as $\x \trans \y / \norm{\y} = \mu_x(\Z) \trans \y/\norm{\y} + W + U$, where $W\mid \y, \Z \sim \mathcal{N}(0,1) $ and $\EE{\abs{U}}$ is small.} Plugging the above into \eqref{eqn:inner_product_key1}, we get 
\begin{equation}
\begin{split}
&\breve{p}_{\operatorname{maxway}}(\bD)  \\
&\qquad = \Phi\p{\frac{-\mu_x(\Z) \trans \y}{\norm{\y}} - \frac{\abs{\x \trans \y}}{\norm{\y}}} + \Phi\p{\frac{\mu_x(\Z) \trans \y}{\norm{\y}} - \frac{\abs{\x \trans \y}}{\norm{\y}}}\\
&\qquad = \Phi\p{\frac{-\mu_x(\Z) \trans \y}{\norm{\y}} -\abs{\frac{\mu_x(\Z) \trans \y}{\norm{\y}} + W + U}} + \Phi\p{\frac{\mu_x(\Z) \trans \y}{\norm{\y}} - \abs{\frac{\mu_x(\Z) \trans \y}{\norm{\y}} + W + U}}\\
&\qquad \geq \Phi\p{\frac{-\mu_x(\Z) \trans \y}{\norm{\y}} -\abs{\frac{\mu_x(\Z) \trans \y}{\norm{\y}} + W}} + \Phi\p{\frac{\mu_x(\Z) \trans \y}{\norm{\y}} - \abs{\frac{\mu_x(\Z) \trans \y}{\norm{\y}} + W}} - \frac{2}{\sqrt{2\pi}} \EE{|U|}\\
& \qquad = \PP{\abs{V + \frac{\mu_x(\Z) \trans \y}{\norm{\y}}} \geq \abs{W + \frac{\mu_x(\Z) \trans \y}{\norm{\y}}} \mid \Z, \y, W} - \frac{2}{\sqrt{2\pi}} \EE{|U|},
\end{split}
\end{equation}
where $V \sim \mathcal{N}(0,1)$, and is independent of other random variables. Note that conditional on $\y$ and $\Z$, the random variables $V$ and $W$ are independently and identically distributed as $\mathcal{N}(0,1)$. Hence, conditional on $\y$ and  $\Z$, $ \PP{\abs{V + \frac{\mu_x(\Z) \trans \y}{\norm{\y}}} \geq \abs{W + \frac{\mu_x(\Z) \trans \y}{\norm{\y}}} \mid \Z, \y, W} \sim \operatorname{Unif}[0,1]$. This further implies that for any $\alpha \in (0,1)$, 
\[\PP{ \PP{\abs{V + \frac{\mu_x(\Z) \trans \y}{\norm{\y}}} \geq \abs{W + \frac{\mu_x(\Z) \trans \y}{\norm{\y}}} \mid \Z, \y, W}  \leq \alpha} = \alpha. \]
Thus $\PP{\breve{p}_{\operatorname{maxway}}(\bD) \leq \alpha} \leq \alpha + \frac{2}{\sqrt{2\pi}} \EE{|U|}$. Combining the above results with \eqref{eqn:inner_product_key2}, we get 
\begin{equation}
\begin{split}
 \PP{\breve{p}_{\operatorname{maxway}}(\bD) \leq \alpha} \leq \alpha + C\p{\sqrt{n} \Delta_{\rho, \operatorname{mean}} + \Delta_{x, \operatorname{mean}} +  \tilde{\Delta}_{x, \operatorname{mean}} + \sqrt{n}\Delta_{x, \operatorname{mean}} \Delta_{y, \operatorname{mean}}}.
 \end{split}
\end{equation}

We will then study the finite $M$ case. We can immediately verify that, the $\pval$ $p_{\operatorname{maxway}} (\bD)$ defined in \eqref{eqn:maxway_CRT} satisfies
$p_{\operatorname{maxway}} (\bD) = \frac{B + 1}{M + 1}$,
where $B \sim \operatorname{Binom}(M, \check{p}_{\operatorname{maxway}})$ conditioning on $\x, \y$ and $\Z$. Therefore, Lemma \ref{lemma:M_infinite_finite} implies that
\begin{equation}
\label{eqn:inner_product_proof_concl}
\begin{split}
\PP{p_{\operatorname{maxway}}(\bD) \leq \alpha} 
&\leq \alpha + C\p{\sqrt{n} \Delta_{\rho, \operatorname{mean}} + \Delta_{x, \operatorname{mean}} +  \tilde{\Delta}_{x, \operatorname{mean}} + \sqrt{n}\Delta_{x, \operatorname{mean}} \Delta_{y, \operatorname{mean}}}.
\end{split}
\end{equation}

{\darkred
\begin{lemma}
\label{lemma:y_bar_inverse}
Under the conditions of Theorem \ref{theo:inner-product_stats}, there exists a constant $C_3$, such that for any function $a(Z)$ of $Z$, 
\begin{equation}
    \EE{\frac{1}{\norm{\y + a(\Z)}^2}\mid \Z} \leq \frac{C_3}{n}. 
\end{equation}
\end{lemma}
\begin{proof}
By Assumption \ref{assu:mean_model}, we have that $Y_i = \smuy(Z_{i \cdot }) + \eta$, where $\eta$ is independent of $Z$. Therefore, conditioning on $Z_{i \cdot }$, $Y_i + a(Z_{i \cdot})$ has a density upper bounded by $C_1$. Picking any $\delta > 0$, we have that
\begin{equation}
\begin{split}
\EE{1/\|\y+ a(\Z)\|^2 \mid \Z}
&= \EE{\frac{1}{\sum_i (Y_i + a(Z_{i \cdot}))^2} \mid \Z}
\leq \frac{1}{n^2}\sum_i\EE{\frac{1}{(Y_i + a(Z_{i \cdot}))^2} \mid Z_{i \cdot }}\\
&= \frac{1}{n^2}\sum_i\EE{\frac{1}{(Y_i + a(Z_{i \cdot}))^2}\mathbbm{1}\cb{|Y_i + a(Z_{i \cdot})| \leq \delta} \mid Z_{i \cdot }}\\
&\qquad \qquad + \frac{1}{n^2}\sum_i\EE{\frac{1}{(Y_i + a(Z_{i \cdot}))^2}\mathbbm{1}\cb{ |Y_i + a(Z_{i \cdot})| > \delta} \mid Z_{i \cdot }}\\
& \leq \frac{1}{n} \p{2C_1\delta + \frac{1}{\delta^2}}. 
\end{split}
\end{equation}
We can then simply take $C_3 = 2C_1\delta + 1/\delta^2$. 
\end{proof}
}

\subsection{Proof of \eqref{eqn:mx_bound_inner}}
\label{subsection:proof_mx_bound_inner}
Similar to the proof of Theorem \ref{theo:inner-product_stats}, we start with the case where $M = \infty$. We again define an ``infinite-$M$" $\pval$
\begin{equation}
\breve{p}_{\mx}(\bD) = \PP[\breve \x \sim \mathcal{N}(\mu_{x, \mx}(\Z), I)]{T(\y,\breve \x,g(\Z), h(\Z))\geq T(\y,\x, g(\Z), h(\Z) ) \mid \x,\y,\Z }.
\end{equation}
This is in turn equivalent to the following form:
\begin{equation}
\begin{split}
\breve{p}_{\mx}(\bD) = 
\PP{\abs{\breve \x \trans \y}  \geq \abs{\x \trans \y}  \mid \x,\y,\Z }
\end{split}
\end{equation}
where $\breve \x \sim \mathcal{N}(\mu_{x, \mx}(\Z), I)$. Conditioning on $\x,\y, \z$, the term $\y \trans \breve \x $ follows a gaussian distribution with mean $\mu_{x, \mx}(\Z) \trans \y $ and variance $ \norm{\y}^2 $. Therefore, 
\begin{equation}
\label{eqn:inner_product_key1_CRT}
\begin{split}
\breve{p}_{\mx}(\bD)  &=  \PP{ \abs{\breve \x \trans \y} \geq \abs{\x \trans \y}  \mid \x,\y,\Z }\\
&= \Phi\p{\frac{-\mu_{x, \mx}(\Z) \trans \y - \abs{\x \trans \y}}{\norm{\y}}} + \Phi\p{\frac{\mu_{x, \mx}(\Z) \trans \y - \abs{\x \trans \y}}{\norm{\y}}},
\end{split}
\end{equation}
where $\Phi$ is the cumulative distribution function of a standard gaussian distribution. 

\sloppy{We will then focus on the term $(\x - \mu_{x, \mx}(\Z)) \trans \y/\norm{\y}$, and show that this term has a distribution very close to a standard gaussian. The term $(\x - \mu_{x, \mx}(\Z)) \trans \y$ can be decomposed into the following:
 \begin{equation}
\p{\x -\mu_{x, \mx}(\Z)}\trans \y 
 =\p{\x - \smux(\Z) + \smux(\Z)  - \mu_{x, \mx}(\Z)}\trans \y
= \bvare \trans \y + \p{\smux(\Z)  - \mu_{x, \mx}(\Z)}\trans \y. 
\end{equation}
Thus,
 \begin{equation}
\p{\x -\mu_{x, \mx}(\Z)}\trans \y /\norm{\y}
= \bvare \trans \y/\norm{\y} + \p{\smux(\Z)  - \mu_{x, \mx}(\Z)}\trans \y/\norm{\y}.
\end{equation}
The first term $\bvare \trans \y /\norm{\y}$ is distributed as a standard gaussian.  Moreover, this term is distributed as a standard gaussian conditional on $\y$ and $\Z$.  The absolute value of the second term $\p{\smux(\Z) - \mu_{x, \mx}(\Z)} \trans \y / \norm{\y}$ can be bounded using the Cauchy-Schwarz inequality: $\abs{\p{\smux(\Z) - \mu_{x, \mx}(\Z)} \trans \y }/ \norm{\y} \leq \norm{\smux(\Z) - \mu_{x, \mx}(\Z)}$. Therefore, $\EE{\abs{\p{\smux(\Z) - \mu_{x, \mx}(\Z)} \trans \y }/ \norm{\y}} \leq \sqrt{\EE{ \norm{\smux(\Z) - \mu_{x, \mx}(\Z)}^2}} = \sqrt{n} \sqrt{\p{\smux(Z) - \mu_{x, \mx}(Z)}^2}$. }

Then, following the same analysis in \eqref{eqn:inner_product_key2} - \eqref{eqn:inner_product_proof_concl}, we get the desired result. 

\subsection{Proof of Theorem \ref{theo:d0_stats}}

The proof idea is very similar to that of the inner-product statistics. We start with the infinite $M$ case, write down explicitly the expression of the $\pval$ and do a decomposition of a key term. Finally, Lemma \ref{lemma:M_infinite_finite} can be applied to connect the infinite $M$ case and the finite $M$ case. 

To this end, as in the proof of Theorem \ref{theo:inner-product_stats}, we define the infinite-$M$ $\pval$ 
\begin{equation}
\breve{p}_{\operatorname{maxway}}(\bD) = \PP[\breve \x \sim \rho^n( \cdot \mid g(\Z) , h(\Z))]{T(\y,\breve \x,g(\Z), h(\Z))\geq T(\y,\x, g(\Z), h(\Z) ) \mid \x,\y,\Z }.  
\end{equation}
This is in turn equivalent to the following form:
\begin{equation}
\begin{split}
\breve{p}_{\operatorname{maxway}}(\bD) = 
\PP{\abs{(\breve \x - \mu_x(\Z)) \trans (\y - \mu_y(\Z))}  \geq \abs{(\x - \mu_x(\Z)) \trans (\y - \mu_y(\Z))}  \mid \x,\y,\Z }
\end{split}
\end{equation}
where $\breve \x \sim \mathcal{N}(\mu_x(\Z), 1)$. Conditioning on $\x,\y, \z$, the term $(\breve \x - \mu_x(\Z)) \trans (\y - \mu_y(\Z))$ follows a gaussian distribution with mean $0$ and variance $ \norm{\y - \mu_y(\Z)}^2 $. Therefore, 
\begin{equation}
\label{eqn:d0_key1}
\begin{split}
\breve{p}_{\operatorname{maxway}}(\bD) & =  \PP{\abs{(\breve \x - \mu_x(\Z)) \trans (\y - \mu_y(\Z))}  \geq \abs{(\x - \mu_x(\Z)) \trans (\y - \mu_y(\Z))}  \mid \x,\y,\Z }\\
& =  2\Phi\p{- \frac{\abs{(\x - \mu_x(\Z)) \trans (\y - \mu_y(\Z))}}{\norm{\y - \mu_y(\Z)}}},
\end{split}
\end{equation}
where $\Phi$ is the cumulative distribution function of a standard gaussian distribution. 

We will then focus on the term $(\x - \mu_x(\Z)) \trans (\y - \mu_y(\Z))/\norm{\y - \mu_y(\Z)}$, and show that this term has a distribution very close to a standard gaussian. The term $\p{\x -\mu_x(\Z)}\trans (\y - \mu_y(\Z))$ can be decomposed into a few terms.
 \begin{equation}
\begin{split}
&\p{\x -\mu_x(\Z)}\trans (\y - \mu_y(\Z)) \\
 &\qquad =\p{\x - \smux(\Z) + \smux(\Z) - \tmux(\Z) + \tmux(\Z) - \mu_x(\Z)}\trans (\y - \mu_y(\Z))\\
&\qquad = \bvare \trans \y + \p{\tmux(\Z) - \mu_x(\Z)} \trans (\y - \mu_y(\Z)) +  \p{\smux(\Z) - \tmux(\Z)} \trans (\y - \mu_y(\Z))\\
&\qquad = \bvare \trans (\y - \mu_y(\Z)) + \p{\tmux(\Z) - \mu_x(\Z)} \trans (\y - \mu_y(\Z)) +  \p{\smux(\Z) - \tmux(\Z)} \trans \boeta +\\
&\qquad \qquad \qquad  \p{\smux(\Z) - \tmux(\Z)} \trans (\tmuy(\Z) - \mu_y(\Z)) + \p{\smux(\Z) - \tmux(\Z)} \trans \p{\smuy(\Z) - \tmuy(\Z)},
\end{split}
\end{equation}
where $\tmuy(Z) = \EE{\smuy(Z) \mid g(Z), h(Z)}$ and $\tmux(Z) = \EE{\smux(Z) \mid g(Z), h(Z)}$. 

\sloppy{The first term $\bvare \trans (\y - \mu_y(\Z)) /\norm{(y - \mu_y(\Z)}$ is distributed as a standard gaussian.  Moreover, this term is distributed as a standard gaussian conditional on $\y$ and $\Z$.  We will then move on to establish that the rest of the terms are small. The absolute value of the second term $\p{\tmux(\Z) - \mu_x(\Z)} \trans (\y - \mu_y(\Z)) / \norm{\y - \mu_y(\Z)}$ can be bounded using the Cauchy-Schwarz inequality: $\abs{\p{\tmux(\Z) - \mu_x(\Z)} \trans (\y - \mu_y(\Z)) }/ \norm{\y - \mu_y(\Z)} \leq \norm{\tmux(\Z) - \mu_x(\Z)}$. Therefore, $\EE{\abs{\p{\tmux(\Z) - \mu_x(\Z)} \trans (\y - \mu_y(\Z)) }/ \norm{\y - \mu_y(\Z)}} \leq \sqrt{\EE{ \norm{\tmux(\Z) - \mu_x(\Z)}^2}} = \sqrt{n} \sqrt{\p{\tmux(Z) - \mu_x(Z)}^2}$.} 

 For the third term $\p{\smux(\Z) -\tmux(\Z)}\trans \boeta$, note that $\EE{\eta \mid Z} = 0$; thus
$\EE{\p{\p{\smux(\Z) -\tmux(\Z)}\trans \boeta}^2} = n \EE{(\smux(Z) - \tmux(Z))^2 \eta^2} = n \EE{(\smux(Z) - \tmux(Z))^2} \EE{\eta^2}$, since the cross terms vanish. Thus, 
\begin{equation}
\begin{split}
&\EE{\abs{\p{\smux(\Z) -\tmux(\Z)}\trans \boeta}/\norm{(y - \mu_y(\Z)}}\\
& \qquad \qquad\leq \sqrt{\EE{\p{\p{\smux(\Z) -\tmux(\Z)}\trans \boeta}^2} \EE{1/\norm{(y - \mu_y(\Z)}^2}}\\
& \qquad \qquad = \sqrt{n} \sqrt{\EE{(\smux(Z) - \tmux(Z))^2}} \sqrt{\EE{\eta^2}} \sqrt{ \EE{1/\norm{(y - \mu_y(\Z)}^2}}\\
& \qquad \qquad \leq C\sqrt{\EE{(\smux(Z) - \tmux(Z))^2}},
\end{split}
\end{equation}
for some constant $C$. The last inequality follows from applying Lemma \ref{lemma:y_bar_inverse}. 

The fourth term $\p{\smux(\Z) - \tmux(\Z)} \trans \p{\tmuy(\Z) - \mu_y(\Z)}$ can be bounded similarly as the third term. Specifically, since $\tmux(Z) = \EE{X \mid g(Z), h(Z)} = \EE{\smux(Z) \mid g(Z), h(Z)}$, we have $\EE{\smux(Z) - \tmux(Z) \mid g(Z), h(Z)} = 0$. We also note that $\tmuy$ and $\mu_y$ are measurable with respect $h$ and $g$; thus $\EE{\p{\smux(Z) - \tmux(Z)} \p{\tmuy(Z) - \mu_y(Z)}} = 0$. Therefore, $\EE{\p{\p{\smux(\Z) - \tmux(\Z)} \trans \p{\tmuy(\Z) - \mu_y(\Z)}}^2} = n \EE{\p{\smux(Z) - \tmux(Z)}^2 \p{\tmuy(Z) - \mu_y(Z)}^2}$. Hence, again by Lemma \ref{lemma:y_bar_inverse},
\begin{equation}
\begin{split}
&\EE{\abs{\p{\smux(\Z) - \tmux(\Z)} \trans \p{\tmuy(\Z) - \mu_y(\Z)}}/\norm{\y - \mu_y(\Z)}}  \\
& \qquad \qquad \leq \sqrt{n} \sqrt{\EE{\p{\smux(Z) - \tmux(Z)}^2 \p{\tmuy(Z) - \mu_y(Z)}^2}} \sqrt{ \EE{1/\norm{\y - \mu_y(\Z)}^2}}\\
& \qquad \qquad \leq C \sqrt{\EE{\p{\smux(Z) - \tmux(Z)}^2 \p{\tmuy(Z) - \mu_y(Z)}^2}},
\end{split}
\end{equation}
for some constant $C$.

Finally, by Lemma \ref{lemma:y_bar_inverse}, the fifth term can be bounded by
\begin{equation}
\begin{split}
&\EE{\abs{\p{\smux(\Z) - \tmux(\Z)} \trans \p{\smuy(\Z) - \tmuy(\Z)}}/\norm{\y - \mu_y(\Z)}} \\
&\qquad \leq \EE{\EE{\abs{\p{\smux(\Z) - \tmux(\Z)} \trans \p{\smuy(\Z) - \tmuy(\Z)}}/\norm{\y - \mu_y(\Z)} \mid \Z}} \\
&\qquad \leq \frac{C}{\sqrt{n}}\EE{\abs{\p{\smux(\Z) - \tmux(\Z)} \trans \p{\smuy(\Z) - \tmuy(\Z)}}} \\
& \qquad \leq \frac{C}{\sqrt{n}}
 \sqrt{\EE{\norm{\smux(\Z) - \tmux(\Z)}^2}} \sqrt{\EE{\big\|\smuy(\Z) - \tmuy(\Z)\big\|^2}}\\
& \qquad \leq C\sqrt{n} 
 \sqrt{\EE{\norm{\smux(Z) - \tmux(Z)}^2}} 
 \sqrt{\EE{\big\|\smuy(Z) - \tmuy(Z)\big\|^2}}, 
\end{split}
\end{equation}
for some constant $C$. 

The above analysis implies that the term $(\x - \mu_x(\Z)) \trans (\y - \mu_y(\Z))/\norm{\y- \mu_y(\Z)}$ can be written as $(\x - \mu_x(\Z)) \trans (\y - \mu_y(\Z))/\norm{\y - \mu_y(\Z)} = W + U$, where $W\mid \y, \Z \sim \mathcal{N}(0,1) $ and 
\begin{equation}
\label{eqn:d0_key2}
\begin{split}
\EE{\abs{U}} &\leq  C\p{\sqrt{n} \Delta_{\rho, \operatorname{mean}} +  \check{\Delta}_{x, \operatorname{mean}} + \sqrt{n} \Delta_{x, \operatorname{mean}} \Delta_{y, \operatorname{mean}}},
\end{split}
\end{equation}
for some constant $C  > 0$, 
where  
\begin{equation}
    \begin{split}
         &\Delta_{x, \operatorname{mean}} = \sqrt{\EE{(\smux(Z) - \tmux(Z))^2}} , \quad 
        \check{\Delta}_{x, \operatorname{mean}} = \sqrt{\EE{(\smux(Z) - \tmux(Z))^2(\tmuy(Z) - \mu_y(Z))^2}}\\
        &\Delta_{y, \operatorname{mean}} = \sqrt{\EE{(\smuy(Z) - \tmuy(Z))^2}}, \quad 
        \Delta_{\rho, \operatorname{mean}} = \sqrt{\EE{(\tmux(Z) - \mu_x(Z))^2}}
    \end{split}
\end{equation}
Plugging the above into \eqref{eqn:d0_key1}, we get 
\begin{equation}
\begin{split}
\breve{p}_{\operatorname{maxway}}(\bD)  
& = 2 \Phi\p{- \frac{\abs{(\x - \mu_x(\Z)) \trans (\y - \mu_y(\Z))}}{\norm{\y - \mu_y(\Z)}}}\\
& = 2\Phi\p{-\abs{W+U}}
 \geq 2\Phi \p{-\abs{W}} - \frac{2}{\sqrt{2\pi}}\EE{\abs{U}}. 
\end{split}
\end{equation}
Since $2\Phi \p{-\abs{W}} \sim \operatorname{Unif}[0,1]$, the above implies that 
$\PP{\breve{p}_{\operatorname{maxway}}(\bD) \leq \alpha} \leq \alpha + \frac{2}{\sqrt{2\pi}} \EE{|U|}$. Combining the above results with \eqref{eqn:d0_key2}, we get 
\begin{equation}
\begin{split}
 \PP{\breve{p}_{\operatorname{maxway}}(\bD) \leq \alpha} \leq \alpha + 
C\p{\sqrt{n} \Delta_{\rho, \operatorname{mean}} + \Delta_{x, \operatorname{mean}} +  \check{\Delta}_{x, \operatorname{mean}} + \sqrt{n} \Delta_{x, \operatorname{mean}} \Delta_{y, \operatorname{mean}}}.
 \end{split}
\end{equation}

We will then study the finite $M$ case. We can immediately verify that, the $\pval$ $p_{\operatorname{maxway}} (\bD)$ defined in \eqref{eqn:maxway_CRT} satisfies
$p_{\operatorname{maxway}} (\bD) = \frac{B + 1}{M + 1}$,
where $B \sim \operatorname{Binom}(M, \check{p}_{\operatorname{maxway}})$ conditioning on $\x, \y$ and $\Z$. Therefore, Lemma \ref{lemma:M_infinite_finite} implies that
\begin{equation}
\begin{split}
 \PP{p_{\operatorname{maxway}}(\bD) \leq \alpha} \leq \alpha + C\p{\sqrt{n} \Delta_{\rho, \operatorname{mean}} + \Delta_{x, \operatorname{mean}} +  \check{\Delta}_{x, \operatorname{mean}} + \sqrt{n} \Delta_{x, \operatorname{mean}} \Delta_{y, \operatorname{mean}}}.
\end{split}
\end{equation}

\subsection{Proof of propositions: valid surrogates}
\label{subsection:perfect_surrogate_proof}

{\darkred
\subsubsection{Proof of Proposition \ref{prop:surro_early}}

If $Z\indp S\mid g(Z)$, then for any function $\psi$ and $\phi$,
\begin{equation}
    \begin{split}
& \EE{\psi(Z) \phi(Y) \mid g(Z)}
= \EE{\EE{\psi(Z) \phi(Y) \mid S,Z}\mid g(Z)}
= \EE{\psi(Z) \EE{ \phi(Y) \mid S,Z}\mid g(Z)}\\
& \qquad \stackrel{(a)}{=} \EE{\psi(Z) \EE{ \phi(Y) \mid S}\mid g(Z)}
\stackrel{(b)}{=} \EE{\psi(Z) \mid g(Z)} \EE{\EE{ \phi(Y) \mid S} \mid g(Z)}\\
& \qquad \stackrel{(c)}{=} \EE{\psi(Z) \mid g(Z)} \EE{ \phi(Y)\mid g(Z)}. 
    \end{split}
\end{equation}
Therefore, $Z\indp Y\mid g(Z)$. 
Here $(a)$ is because $Y \indp Z \mid Z$, $(b)$ is because $Z \indp S \mid g(Z)$, and $(c)$ is because 
\begin{equation}
\EE{ \phi(Y)\mid g(Z)}
= \EE{\EE{\phi(Y) \mid S,Z}\mid g(Z)}
= \EE{\EE{\phi(Y) \mid S}\mid g(Z)}. 
\end{equation}
}
\subsubsection{Proof of Proposition \ref{prop:2}}

Since $S \indp Z \mid Y$, we have
\begin{equation*}
\begin{split}
\PP{S\leq a\mid Z}=&\PP{S\leq a\mid Z,Y=1}\PP{Y=1\mid Z}+\PP{S\leq a\mid Z,Y=0}\PP{Y=0\mid Z}\\
=&\PP{S\leq a\mid Y=1}\PP{Y=1\mid Z}+\PP{S\leq a\mid Y=0}\left\{1-\PP{Y=1\mid Z}\right\},
\end{split}    
\end{equation*}
which implies that
\begin{equation}
\PP{Y=1\mid Z}=\frac{\PP{S\leq a\mid Z}-\PP{S\leq a\mid Y=0}}{\PP{S\leq a\mid Y=1}-\PP{S\leq a\mid Y=0}}.
\label{equ:prop:2:1}
\end{equation}
If $Z\indp S\mid g(Z)$, then $g(Z)$ is sufficient for $\PP{S\leq a\mid Z}$ as a function of $Z$. Hence, by (\ref{equ:prop:2:1}), $\PP{Y=1\mid Z}$ can be represented as a function of $g(Z)$, which indicates that $Z\indp Y\mid g(Z)$.

{\darkred
\subsubsection{Proof of Proposition \ref{prop:surro_linear}}
Here, we will show that for any linear function $g$ of $Z$, $Z\indp S\mid g(Z)$ implies that $Z\indp Y\mid g(Z)$. 

We require $g$ to be a linear function of $Z$. Assume that $g(Z) = Z \trans a$ for some $a \in \mathbb{R}^p$. It then suffices to show that $a = \gamma \thetas$ for some $\gamma \neq 0$. This is because $S \indp Y \mid Z\trans \thetas$. 
Define $t(Z) = \EE{S \mid Z} \in \mathbb{R}$. On the one hand, we note that $\EE{t(Z)^2 \mid g(Z)} = \EE{S t(Z) \mid g(Z)}=
\EE{S\mid g(Z)} \EE{t(Z) \mid g(Z)} = \p{\EE{t(Z) \mid g(Z)}}^2$. Hence, $\operatorname{Var}\sqb{t(Z) \mid g(Z)} = 0$. This then implies that $t(Z) = l(g(Z)) = l(Z \trans a)$ almost surely for some function $l$. On the other hand, since $Y \sim Z$ follows a linear model and the surrogate $S$ satisfies $S \indp Z \mid Y$, the surrogate $S$ follows a single index model (SIM) given $Z$, i.e., $S = f(Z\trans\thetas, e)$ with $e \indp Z$. Therefore, $t(Z) = \EE{S \mid Z} = \tilde{l}(Z \trans \thetas)$ for some function $\tilde{l}$.

We have established that $t(Z) = l(Z \trans a) = \tilde{l}(Z \trans \thetas)$ almost surely. This then implies that $\operatorname{Var}\sqb{t(Z) \mid Z\trans \thetas} = \operatorname{Var}\sqb{l(Z \trans a) \mid Z\trans \thetas}=
\operatorname{Var}\sqb{\tilde{l}(Z \trans \thetas) \mid Z\trans \thetas} = 0$, which implies that either $l$ is a constant function (a.s.) or $a = \gamma \thetas$ for some $\gamma \neq 0$. However, if $l$ is a constant function, then $t(Z)$ is constant and $\EE{S Z} = \EE{t(Z) Z} = c\EE{Z} = 0$, which contradicts the assumption that $\EE{S Z} \neq 0$. Therefore, $a = \gamma \thetas$ for some $\gamma \neq 0$.  
}

\section{Robustness of the transformed Maxway CRT}
\label{section:robust_transformed}
Similar to our analysis in Sections \ref{section:exact_inference}, \ref{sec:theory:robust:1} and \ref{section:specific_stats}, we can obtain robustness results on the transformed Maxway CRT. Specifically, we establish Theorems \ref{theo:exact_inference_tran}-\ref{theo:inner-product_stats_trans} as corollaries of Theorems \ref{theo:exact_inference}-\ref{theo:inner-product_stats} respectively. 

\begin{theorem}
\label{theo:exact_inference_tran}
Suppose that either of the following conditions holds: (SS.I) the vectors $\r, \r^{(1)}, \dots,  \r\supm$ are exchangeable conditioning on $\Z$; (SS.II) the vectors $\r, \r^{(1)}, \dots,  \r\supm$ are exchangeable conditioning on $\{h(\Z),g(\Z)\}$, and $\Z \indp \y \mid g(\Z)$. Then the transformed Maxway CRT $p$-value defined in \eqref{eqn:transform_maxway_CRT} is valid, i.e., $\PP{p_{\operatorname{t-maxway}}(\bD) \leq \alpha} \leq \alpha$ for any $\alpha \in [0,1]$ under the null hypothesis \eqref{eqn:indep_test}. 
\end{theorem}

\begin{theorem}[Type-I error bound: arbitrary test statistic]
Under the null hypothesis \eqref{eqn:indep_test}, for any $\alpha \in (0,1)$,  
\begin{equation}
\label{eqn:mayway_bound_trans}
\begin{split}
\PP{p_{\operatorname{t-maxway}}(\bD) \leq \alpha} 
& \leq \alpha + 2\EE{ \Delta_r \Delta_y} + \EE{\Delta_{r|g,h}},
\end{split}
\end{equation} 
where 
\begin{equation}
    \begin{split}
         \Delta_r &= d_{\operatorname{TV}} \Big(\rho_r^{\star n}( \cdot \mid g(\Z), h(\Z)), f_{\r \mid \Z}( \cdot \mid \Z) \Big),\\
         \Delta_y &= d_{\operatorname{TV}}\p{f_{\y \mid \Z} ( \cdot \mid \Z), f_{\y \mid g(\Z), h(\Z)} ( \cdot \mid g(\Z), h(\Z))},\\
		\Delta_{r|g,h} &= d_{\rm TV} \p{\rho_r^{\star n}( \cdot \mid g(\Z) , h(\Z)),\rho_r^n( \cdot \mid g(\Z) , h(\Z))}.
    \end{split}
\end{equation}
Here, $\rho^{\star n}_r$ is the distribution of $\r \mid g(\Z), h(\Z)$, and $\rho^{n}_r$, as an estimate $\rho^{\star n}_r$, is the distribution from which $\r\supm$ is sampled in Algorithm \ref{alg:transform_maxway}. 
\label{thm:main:1_trans}
\end{theorem}

For the inner product statistic, let $\mu^{\star}_y(Z) = \EE{Y \mid Z}$ and $\mu^{\star}_r(Z) = \EE{R(X,Z) \mid Z}$. 
Let $\tmuy(Z) = \EE{Y \mid g(Z), h(Z)} = \EE{\smuy(Z) \mid g(Z), h(Z)}$ and $\tilde{\mu}_r(Z) = \EE{R(X,Z) \mid g(Z), h(Z)} = \EE{\mu^{\star}_r(Z) \mid g(Z), h(Z)}$. 
We work under the following assumption. 
\begin{assumption}
\label{assu:mean_model_transform}
Under the null hypothesis, $Y = \smuy(Z) + \eta$ and $R(X,Z) = \mu^{\star}_r(Z) + \varepsilon$, where $\eta$ and $\epsilon$ are mean-zero random variables independent of $Z$ and independent of each other, and $\varepsilon \sim \mathcal{N}(0,1)$. 
\end{assumption} 
We focus on the following statistic
\begin{equation}
\label{eqn:corr_stats_transform}
T(\y,\r,g(\Z),h(\Z)) = \abs{\r \trans \y}.
\end{equation}
\begin{theorem}[Type-I error bound: inner-product statistic]
\label{theo:inner-product_stats_trans}
Under Assumption \ref{assu:mean_model_transform}, assume further that the transformed Maxway CRT (Algorithm \ref{alg:transform_maxway}) samples $\r\supm$ from a normal distribution, i.e., $\rho_r(\cdot \mid g(Z), h(Z))$ corresponds to $\mathcal{N}(\mu_r(Z), 1)$.
Assume that $\eta$ is a continuous random variable whose density is upper bounded by a constant $C_1$. 
Assume further that there exists a positive constant $C_2$ such that $\EE{\eta^2} \leq C_2$. 
Then there exists a positive constant $C$ such that for any $\alpha \in (0,1)$, the type-I error of the Maxway CRT using the inner-product statistic defined in \eqref{eqn:corr_stats} can be bounded by 
\begin{equation}
\label{eqn:mayway_bound_inner_trans}
\PP{p_{\operatorname{maxway}}(\bD) \leq \alpha}  \leq \alpha  + 
C  \p{\sqrt{n}\Delta_{\rho, \operatorname{mean}}  + \sqrt{n}\Delta_{r, \operatorname{mean}} \Delta_{y, \operatorname{mean}} + (\Delta_{r, \operatorname{mean}} + \tilde{\Delta}_{r, \operatorname{mean}})},
\end{equation} 
where  
\begin{equation}
    \begin{split}
         &\Delta_{r, \operatorname{mean}} = \sqrt{\EE{(\mu_{r}^{\star}(Z) - \tilde{\mu}_r(Z))^2}} , \qquad 
        \tilde{\Delta}_{r, \operatorname{mean}} = \sqrt{\EE{((\mu_{r}^{\star}(Z) - \tilde{\mu}_r(Z)\tmuy(Z))^2}}\\
        &\Delta_{y, \operatorname{mean}} = \sqrt{\EE{(\smuy(Z) - \tmuy(Z))^2}}, \qquad 
        \Delta_{\rho, \operatorname{mean}} = \sqrt{\EE{(\tilde{\mu}_r(Z) - \mu_r(Z))^2}}. 
    \end{split}
\end{equation}
\end{theorem}

\section{Additional examples and details}
\label{appendix:examples}

\subsection{Details of Convergence rate example \ref{exam:gaussian_linear0}}
\label{subsection:detail_example_gauss_linear}

Since we consider the transformed Maxway CRT (Algorithm \ref{alg:transform_maxway}) here, we will make use of Theorems \ref{thm:main:1_trans} and \ref{theo:inner-product_stats_trans} instead of Theorems \ref{theo:almost_double_robust} and \ref{theo:inner-product_stats} for notation clarity. 

\subsubsection{Arbitrary statistic}
We will apply Theorem \ref{thm:main:1_trans}. 
We will use the Pinsker's inequality to bound the total variation distance in terms of the Kullback-Leibler divergence. Specifically, if $A = (A_1, \dots, A_n)$ and $B = (B_1, \dots, B_n)$ where $A_i$'s are i.i.d. random variables and $B_i$'s are i.i.d. random variables, then the total variation between the distribution of $A$ and the distribution of $B$ can be bounded by 
\begin{equation}
\label{eqn:pinsker}
d_{\operatorname{TV}}(f_A(\cdot), f_B(\cdot)) \leq  
\sqrt{\frac{1}{2}  d_{\operatorname{KL}}(f_A(\cdot),f_B(\cdot))} \leq
\sqrt{\frac{n}{2}} \sqrt{d_{\operatorname{KL}}(f_{A_1}(\cdot),f_{B_1}(\cdot))}. 
 \end{equation}
We also note that for two univariate gaussian distributions, the Kullback-Leibler divergence between them is given by
\begin{equation}
\label{eqn:kl_gaussian}
d_{\operatorname{KL}} \p{\mathcal{N}( \mu_1, \sigma_1^2), \mathcal{N}(\mu_2, \sigma_2^2)} = \frac{(\mu_1 - \mu_2)^2}{2 \sigma_2^2} + \frac{\sigma_1^2 - \sigma_2^2}{2 \sigma_2^2} + \log(\sigma_2/\sigma_1). 
 \end{equation}
If $\mu_1, \mu_2 \stackrel{p}{\to} \mu$, $\sigma_1, \sigma_2 \stackrel{p}{\to} \sigma$, and $(\mu_1 - \mu_2)^2 \asymp \sigma_1^2 - \sigma_2^2 \ll 1$, then the above Kullback-Leibler divergence satisfies
\begin{equation}
\label{eqn:kl_gaussian2}
d_{\operatorname{KL}} \p{\mathcal{N}( \mu_1, \sigma_1^2), \mathcal{N}(\mu_2, \sigma_2^2)} \asymp (\mu_1 - \mu_2)^2 + (\sigma_1^2 - \sigma_2^2).
\end{equation}

We will analyze each of the terms $\Delta_r$, $\Delta_y$, and $\Delta_{r|g,h}$. For $\Delta_r$, recall that the $\Delta_r = d_{\operatorname{TV}} \Big(\rho_r^{\star n}( \cdot \mid g(\Z), h(\Z)), f_{\r \mid \Z}( \cdot \mid \Z) \Big) \lesssim n d_{\operatorname{KL}} \Big(\rho_r^{\star}( \cdot \mid g(Z_{i \cdot}), h(Z_{i \cdot})), f_{R(X_i,Z_{i \cdot}) \mid Z_{i \cdot}}( \cdot \mid Z_{i \cdot}) \Big)$. We note that the two distributions are both gaussian, and thus we can apply \eqref{eqn:kl_gaussian2} and bound the KL-divergence by the difference between means and variances. The mean of the first distribution is $\EE{R(X_i,Z_{i \cdot}) \mid g(Z_{i \cdot}), h(Z_{i \cdot})}= \EE{X_i - Z_{i \cdot} \trans \beta\mid g(Z_{i \cdot}), h(Z_{i \cdot})} = \EE{Z_{i \cdot} \trans (\betas-\beta) \mid g(Z_{i \cdot}), h(Z_{i \cdot})}$, whereas the mean of the second is $Z_{i \cdot} \trans (\betas-\beta)$. Therefore, 
\[
\begin{split}
 \EE{\p{\mu_1 - \mu_2}^2} 
&=  \EE{\operatorname{Var}\sqb{Z_{i \cdot} \trans \betas - Z_{i \cdot} \trans \beta\mid g(Z_{i \cdot}), h(Z_{i \cdot})}}
\leq \EE{\p{Z_{i \cdot} \trans \p{\beta - \betas}}^2}\\
& = \p{\beta - \betas}\trans \Sigma_z \p{\beta - \betas}
 \lesssim \Delta_{x, \lin}^2,
 \end{split}
\]
if the largest eigenvalue of $\Sigma_z$ is bounded. For the variances, the variance of the first distribution is $\operatorname{Var}\sqb{R(X_i,Z_{i \cdot}) \mid g(Z_{i \cdot}), h(Z_{i \cdot})}  =  \operatorname{Var}\sqb{\eta_i + Z_{i \cdot}\trans (\betas - \beta) \mid g(Z_{i \cdot}), h(Z_{i \cdot})}$. The variance of the second distribution is simply $1$. Therefore,
\[
\begin{split}
\EE{\sigma_1^2 - \sigma_2^2}
&= \EE{\operatorname{Var}\sqb{\eta_i + Z_{i \cdot}\trans (\betas - \beta)  \mid g(Z_{i \cdot}), h(Z_{i \cdot})} - 1} 
= \EE{\operatorname{Var}\sqb{Z_{i \cdot}\trans (\betas - \beta)  \mid  g(Z_{i \cdot}), h(Z_{i \cdot})}} \\
&\leq \EE{ \p{Z_{i \cdot}\trans (\beta - \betas)}^2}
 = \p{\beta - \betas}\trans \Sigma_z \p{\beta - \betas}
 \lesssim \Delta_{x, \lin}^2.
\end{split}
\]
Thus \eqref{eqn:kl_gaussian2} implies that $\EE{\Delta_r^2} \lesssim n\Delta_{x, \lin}^2$. 

The exact same analysis can be applied to $\Delta_y$ and we can get $\EE{\Delta_y^2} \lesssim n\Delta_{y, \lin}^2$. 
Similarly, for $\Delta_{r|g,h}$, since both the mean and variance of the gaussian distribution are estimated with linear regression with rate $\Delta_{\rho, \lin}$, we have that
$\EE{\Delta_{r|g,h}^2} \lesssim n\Delta_{\rho, \lin}^2$. 

Finally, by Theorem \ref{thm:main:1_trans}, we have that
\[
\begin{split}
&\textnormal{Type-I error inflation of the Maxway CRT} \\
& \qquad \leq 2\EE{\Delta_r \Delta_y} + \EE{\Delta_{x\mid g,h}}
\leq 2\sqrt{\EE{\Delta_r^2}\EE{\Delta_y^2}} + \sqrt{\EE{\Delta_{r\mid g,h}}^2}\\
& \qquad \lesssim n \Delta_{x, \lin} \Delta_{y, \lin}+ \sqrt{n} \Delta_{\rho, \lin}.
\end{split}
\]

\subsubsection{Inner-product statistic}
We will apply Theorem \ref{theo:inner-product_stats_trans}. In order to use Theorem \ref{theo:inner-product_stats_trans}, we implement the (transformed) Maxway CRT in a slightly different way. In particular, we take $\rho_r(Z_{i\cdot})$ to be $\mathcal{N}(\gamma_0 + \gamma_1 (Z_{i\cdot}\trans \theta),1)$, where $\gamma_0$ and $\gamma_1$ are estimated with linear regression of $R(X_i, Z_{i\cdot})$ on $Z_{i\cdot}\trans \theta$ (on external data) with error rate $\Delta_{\rho, \lin}$. 

To apply Theorem \ref{theo:inner-product_stats_trans}, we start with making sense of the notation and verifying the assumptions. Using the notation in Theorem \ref{theo:inner-product_stats_trans}, we have that $\smuy(Z) = Z \trans \thetas$ and $\mu_r^{\star}(Z) = Z \trans (\betas - \beta)$.
We also have that $\tmuy(Z) = \EE{Z \trans \thetas \mid Z \trans \theta}$ and $\tilde{\mu}_r(Z) = \EE{Z \trans (\betas - \beta) \mid Z \trans \theta}$. 
Therefore, we have that $Y = \smuy(Z) + \eta$ and $R(X,Z) = \mu_r^{\star}(Z) + \epsilon$ with both $\eta$ and $\epsilon$ following $\mathcal{N}(0,1)$, and thus the conditions in Assumption~\ref{assu:mean_model_transform} are satisfied. 

Now we are ready to study the bound in Theorem \ref{theo:inner-product_stats_trans}. We will analyze the terms $\Delta_{r,\operatorname{mean}}$,
$\tilde{\Delta}_{r,\operatorname{mean}}$,
$\Delta_{y,\operatorname{mean}}$ and $\Delta_{\rho,\operatorname{mean}}$ one by one. For $\Delta_{r,\operatorname{mean}}$, we have that
\[
\begin{split}
\Delta_{r,\operatorname{mean}}^2
&= \EE{\p{\tilde{\mu}_r(Z)- \mu_r^{\star}(Z)}^2} = \EE{\p{Z \trans (\beta - \betas) - \EE{Z \trans (\beta - \betas) \mid Z \trans \theta}}^2}\\
&\leq \EE{\p{Z \trans \p{\beta - \betas}}^2}
 = \p{\beta - \betas}\trans \Sigma_z \p{\beta - \betas}
 \lesssim \Delta_{x, \lin}^2.
 \end{split}
 \]
 Here, the last inequality follows from the fact that the largest eigenvalue of $\Sigma_z$ is bounded. 
 For $\tilde{\Delta}_{r,\operatorname{mean}}$, the analysis is similar, but slightly more complicated. 
 We firstly note that $\EE{\tmuy(Z)^4} \leq \EE{\p{Z \trans \thetas}^4} = 3 \p{\thetas{}\trans \Sigma_z \thetas}^2 \leq C$ for some constant $C >0$ if the largest eigenvalue of $\Sigma_z$ is bounded above by a constant.  Here, we make use of the fact that the fourth moment of a $\mathcal{N}(0, \sigma^2)$ random variable is $3 \sigma^4$. Thus, 
 \[
\begin{split}
\tilde{\Delta}_{r,\operatorname{mean}}^2
&=  \EE{\p{\tilde{\mu}_r(Z)- \mu_r^{\star}(Z)}^2\tilde{\mu}_y(Z)^2}
= \sqrt{\EE{\tilde{\mu}_y(Z)^4}} \sqrt{\EE{\p{\tilde{\mu}_r(Z)- \mu_r^{\star}(Z)}^4}}\\
&= \sqrt{C}\sqrt{\EE{\p{Z \trans (\beta - \betas) - \EE{Z \trans (\beta - \betas) \mid Z \trans \theta}}^4}}
\leq \sqrt{C} \sqrt{\EE{\p{Z \trans \p{\beta - \betas}}^4}}\\
& = \sqrt{3C}\p{\beta - \betas}\trans \Sigma_z \p{\beta - \betas}
 \lesssim \Delta_{x, \lin}^2. 
 \end{split}
 \]
Then for $\Delta_{y,\operatorname{mean}}$, 
 \[
\begin{split}
\Delta_{y,\operatorname{mean}}^2  
& = \EE{\p{\tilde{\mu}_y(Z)- \mu_y^{\star}(Z)}^2}
 = \EE{\p{Z\trans \thetas - \EE{Z\trans \thetas\mid Z\trans \theta}}^2}\\
&\leq \EE{\p{Z \trans \p{\theta - \thetas}}^2}
 =  \p{\theta - \thetas}\trans \Sigma_z \p{\theta - \thetas}
 \lesssim \Delta_{y, \lin}^2.
 \end{split}
 \]
 For $\Delta_{\rho,\operatorname{mean}}^2$, note that since $Z$ is distributed as a multivariate gaussian, $\tilde{\mu}_r(Z) = \EE{R(X,Z)\mid Z \trans \theta} = \EE{Z \trans (\betas - \beta)\mid Z \trans \theta}$ is linear in $Z \trans \theta$. In other words, there exist $\gamma_0^{\star}$ and $\gamma_1^{\star}$ such that $\tilde{\mu}_r(Z) = \gamma_0^{\star} + \gamma_1^{\star} (Z \trans \theta)$. Therefore, 
 \[
\EE{\Delta_{\rho,\operatorname{mean}}^2}
= \EE{\p{\tilde{\mu}_r(Z) - \mu_r(Z) }^2}
= \EE{\p{\gamma_0^{\star} + \gamma_1^{\star} (Z \trans \theta) -(\gamma_0 + \gamma_1 (Z \trans \theta)) }^2}
\lesssim \Delta_{\rho, \lin}^2.
 \]

 Finally, by Theorem \ref{theo:inner-product_stats_trans}, we have that
\[
\begin{split}
&\textnormal{Type-I error inflation of the Maxway CRT (inner-product statistic)} \\
&\qquad \lesssim \sqrt{n}\Delta_{\rho,\operatorname{mean}}
+\Delta_{r,\operatorname{mean}}
+ \tilde{\Delta}_{r,\operatorname{mean}}
+ \sqrt{n} \Delta_{r,\operatorname{mean}}  \Delta_{y,\operatorname{mean}}
\\
&\qquad \lesssim \sqrt{n} \Delta_{x, \lin} \Delta_{y, \lin}+ \sqrt{n} \Delta_{\rho, \lin}.
\end{split}
\]

\subsubsection{Model-X CRT}
For the model-X CRT, assume that we sample $X_i\supm$ independently from $\mathcal{N}(Z_{i \cdot} \trans \beta, 1)$. With arbitrary test statistic, we have that
\[
\PP{p_{\operatorname{mx}}(\bD) \leq \alpha} 
 \leq \alpha + \EE{d_{\operatorname{TV}}\p{f_{\x \mid \Z}( \cdot \mid \Z) , \widetilde{\rho}^n_{\x \mid \Z}( \cdot \mid \Z) }}.
\]
Recall that conditioning on $\Z$, $X_i \sim \mathcal{N}(Z_{i \cdot} \trans \betas, 1)$. 
Therefore, by \eqref{eqn:pinsker} and \eqref{eqn:kl_gaussian2}, we have that
\[
\begin{split}
    &\textnormal{Type-I error inflation of the model-X CRT (arbitrary statistic)} \\
    &\qquad \qquad =\EE{d_{\operatorname{TV}}\p{f_{\x \mid \Z}( \cdot \mid \Z) , \widetilde{\rho}^n_{\x \mid \Z}( \cdot \mid \Z) }} 
    \lesssim n\sqrt{\EE{\p{Z_{i \cdot} \trans \p{\beta - \betas}}^2}} \\
    &\qquad \qquad = \sqrt{n}\sqrt{\p{\beta - \betas}\trans \Sigma_z \p{\beta - \betas}}
 \lesssim \sqrt{n}\Delta_{x, \lin},
 \end{split}
\]
if the largest eigenvalue of $\Sigma_z$ is bounded.

If we draw attention to the inner-product statistic and apply \eqref{eqn:mx_bound_inner}, we can get a similar bound on the type-I error inflation of the model-X CRT. 
\[
\begin{split}
    &\textnormal{Type-I error inflation of the model-X CRT (inner-product statistic)} \\
    &\qquad \qquad \lesssim n\sqrt{\EE{\p{Z_{i \cdot} \trans \p{\beta - \betas}}^2}}  = \sqrt{n}\sqrt{\p{\beta - \betas}\trans \Sigma_z \p{\beta - \betas}}
 \lesssim \sqrt{n}\Delta_{x, \lin}.
 \end{split}
\]

{\darkred
\subsection{Details of Convergence rate example \ref{exam:SIM_theory}}
\label{subsection:detail_example_gauss_linear_cont}
In this example, we assume that $Z_{i \cdot} \sim \mathcal{N}(0, \Sigma_z)$, $Y_i = Z_{i \cdot} \trans \theta^{\star} + \eta_i$, and $X_i = Z_{i \cdot}\trans \beta^{\star} + \varepsilon_i$, where $\eta_i \sim \mathcal{N} (0,1)$ and $\varepsilon_i \sim \mathcal{N} (0,1)$ are noise terms independent with $Z_{i \cdot}$ and $\Sigma_z \succ 0$. Assume further that the surrogate satisfies $S_i \indp Z_i \mid Y_i$. Furthermore, assume that $\EE{S Z} \neq 0$, i.e., $Z$ has some predictive power of $S$.

We will make use of results in Convergence rate example \ref{exam:gaussian_linear0}. It suffices to quantify $\Delta_{\rho, \lin}$, $\Delta_{x, \lin}$ and $\Delta_{y, \lin}$ for semi-supervised learning and surrogate-assisted semi-supervised learning. 

Since $\Delta_{\rho, \lin}$ denote the error of a one-dimension linear regression, we have that $\Delta_{\rho, \lin} \asymp 1\sqrt{n}$. For $\Delta_{x, \lin}$, this is the error of a lasso regression on the unlabeled data of size $N$. Therefore, under our convergence rate assumptions, we have that $\Delta_{x, \lin} \lesssim \sqrt{s_\beta \log(p)/N}$. 

Finally, we study $\Delta_{y,\lin}$. In the semi-supervised learning scenario, $\Delta_{y, \lin}$ is the error of a lasso regression on the labeled holdout training data, and thus under our convergence rate assumptions, we have that $\Delta_{y, \lin} \lesssim \sqrt{s_\theta \log(p)/n}$. In the surrogate-assisted semi-supervised scenario, we need to do slightly more work. Since $Y \sim Z$ follows a linear model and the surrogate $S$ satisfies $S \indp Z \mid Y$, $S$ follows a single index model (SIM) given $Z$, i.e., $S = f(Z\trans\thetas, e)$ with $e \indp Z$. \citet{li1989regression} establishes that when $S$ follows a SIM, the direction of $\thetas$ can be recovered using the least square regression of $S$ against $Z$ (see also \citet{zhang2022prior}). In particular, this implies that if we run lasso with $S$ as response and $Z$ as predictors on the unlabeled samples, we can obtain an estimator $\theta$ of $\thetas$ such that $\norm{\theta \gamma - \theta^{\star}}^2 \lesssim s_\theta \log(p)/N$ for some constant $\gamma \neq 0$. 
Without additional information, we do not know the value of $\gamma$. However, this constant $\gamma$ does not have an impact on the Maxway CRT procedure or the resulting type-I error. 
More precisely, when we run the Maxway CRT, if we consider two candidates, $g_1$ and $g_2$, of the $g$ function such that $g_1(Z) = \gamma g_2(Z)$, then the conditional distribution of $X \mid g_1(Z) = g_1(z)$ is the same as $X \mid g_2(Z) = g_2(z)$ for any $z \in \mathbb{R}^p$. Therefore, the Maxway CRT procedure is the same if we take $g(Z)$ to be either $Z \trans \theta$ or $\gamma Z \trans\theta$. It thus suffices to bound $\Delta_{y, \lin}$ in the case where we take $g(Z)$ to be $\gamma Z \trans\theta$. 
When we take $g(Z)$ to be $\gamma Z \trans\theta$, $\Delta_{y, \lin}$ corresponds to $\norm{\theta \gamma - \theta^{\star}}$ and hence $\Delta_{y, \lin} \lesssim \sqrt{s_\theta \log(p)/N}$. 
}

\subsection{An additional convergence rate example}
\label{subsection:example}

\begin{rateexampleprime}{2}
Gaussian linear model; $g(\Z)$ is a subset of variables in $\Z$.
\label{exam:gauss_variable_selection}
\leavevmode
\begin{description}[font=\itshape,leftmargin=0cm,labelindent=0cm]
    \item[Modelling assumptions.]
Suppose that we have $n$ labeled samples $\bD=(\by,\bx,\bZ)$ and $N$ unlabeled samples (with or without surrogate): $\bD^{u}=(\bs^u,\bx^u,\bZ^u)$ or $\bD^{uS}=(\bx^u,\bZ^u)$.
Assume further that $Z_{i \cdot} \sim \mathcal{N}(0, \Sigma_z)$, $Y_i = Z_{i \cdot} \trans \theta^{\star} + \eta_i$, and $X_i = Z_{i \cdot}\trans \beta^{\star} + \varepsilon_i$, where $\eta_i \sim \mathcal{N} (0,1)$ and $\varepsilon_i \sim \mathcal{N} (0,1)$ are noise terms independent with $Z_{i \cdot}$ and $\Sigma_z \succ 0$. Let $\mcS^{\star} = \cb{j: \theta_j \neq 0}$ be the support set of $\theta$.

\item[Implementation of the Maxway CRT.]
We implement the transformed Maxway CRT (Algorithm \ref{alg:transform_maxway}) here. We take $g(Z_{i \cdot}) = Z_{i, \mcS}$ to be an estimate of $Z_{i, \mcS^{\star}}$. Take the transformation $R(X_i, Z_{i \cdot}) = X_i - Z_{i \cdot} \trans \beta$ where $Z_{i \cdot} \trans \beta$ is an estimate of the conditional mean function of $X_i$, and $h(Z)$ to be trivially null. 

 \item[Convergence rate assumptions.]
Assume that $\beta^{\star}$ is $s_\beta$ sparse, and it can be estimated with lasso on data of sample size $m$ with rate $\norm{\beta - \beta^{\star}}^2 \lesssim s_\beta \log(p)/m$. 
Further, assume that the set $\mcS^{\star}$ can be recovered exactly with lasso with probability $1 - \delta$, where $\delta = o_p(1)$. Note that the specific value of $\delta$ may depend on whether we are in a semi-supervised scenario or a surrogate-assisted semi-supervised scenario. We refer to \citet{bickel2009simultaneous} and \citet{van2009conditions} for a more detailed discussion on the rate of lasso and to \citet{zhao2006model} for the variable selection consistency of lasso. 

\item[Rate of type-I error inflation.]
\sloppy{In this example, we can establish that with arbitrary test statistic, $\EE{\Delta_x} \lesssim \sqrt{n} \EE{\norm{\beta - \beta^{\star}}} \lesssim \sqrt{s_\beta \log(p) n /N}$ and that $\EE{\Delta_{x|g,h}} \lesssim \sqrt{  n s_\theta/N}$. }
Therefore, the type-I error inflation of the Maxway CRT can be bounded by
\begin{equation}
\begin{split}
\EE{\Delta_{x|g,h}} +  2\EE{ \Delta_x \Delta_y}  \lesssim \sqrt{ \frac{  s_\theta n}{N}} + \delta \sqrt{\frac{s_\beta \log(p) n}{N}}.
\end{split}
\end{equation}
As a comparison, for the model-X CRT, we have
\begin{equation}
\textnormal{Type-I error inflation of the model-X CRT} \lesssim \sqrt{\frac{s_\beta \log(p) n}{N}}.
\end{equation}
We immediately notice that the bound of the Maxway CRT is better in rate than that of the model-X CRT if $s_\theta \lesssim s_\beta \log(p)$, i.e., if the model of $Y$ is not more complex than that of $X$ by a factor of $\log(p)$.
\end{description}

\end{rateexampleprime}

\subsubsection{Details of Convergence rate example \ref{exam:gauss_variable_selection}}
Since we consider the transformed Maxway CRT (Algorithm \ref{alg:transform_maxway}) here, we will make use of Theorem \ref{thm:main:1_trans} instead of Theorem \ref{theo:almost_double_robust} for notation clarity. 

Similar to the analysis in Convergence rate examples \ref{exam:gaussian_linear0} and \ref{exam:SIM_theory}, we will again use the Pinsker's inequality to bound the total variation distance in terms of the Kullback-Leibler divergence. 

We will analyze each of the terms $\Delta_r$, $\Delta_y$, and $\Delta_{r|g,h}$. For $\Delta_r$, recall that the $\Delta_r = d_{\operatorname{TV}} \Big(\rho_r^{\star n}( \cdot \mid g(\Z), h(\Z)), f_{\r \mid \Z}( \cdot \mid \Z) \Big) \lesssim n d_{\operatorname{KL}} \Big(\rho_r^{\star}( \cdot \mid g(Z_{i \cdot}), h(Z_{i \cdot})), f_{R(X_i,Z_{i \cdot}) \mid Z_{i \cdot}}( \cdot \mid Z_{i \cdot}) \Big)$. We note first that the two distributions are both gaussian; thus we can apply \eqref{eqn:kl_gaussian2} and bound the KL-divergence by the difference between means and variances. The mean of the first distribution is $\EE{R(X_i,Z_{i \cdot}) \mid g(Z_{i \cdot}), h(Z_{i \cdot})}= \EE{X_i - Z_{i \cdot} \trans \beta\mid g(Z_{i \cdot}), h(Z_{i \cdot})} = \EE{Z_{i \cdot} \trans (\betas-\beta) \mid g(Z_{i \cdot}), h(Z_{i \cdot})}$, whereas the mean of the second is $Z_{i \cdot} \trans (\betas-\beta)$. Therefore, 
\[
 \EE{\p{\mu_1 - \mu_2}^2} 
=  \EE{\operatorname{Var}\sqb{Z_{i \cdot} \trans \betas - Z_{i \cdot} \trans \beta\mid g(Z_{i \cdot}), h(Z_{i \cdot})}}
\leq \EE{\p{Z_{i \cdot} \trans \p{\beta - \betas}}^2}
 \lesssim s_\beta \log(p)/N,
\]
if the largest eigenvalue of $\Sigma_z$ is bounded. For the variances, the variance of the first distribution is $\operatorname{Var}\sqb{R(X_i,Z_{i \cdot}) \mid g(Z_{i \cdot}), h(Z_{i \cdot})}  =  \operatorname{Var}\sqb{\eta_i + Z_{i \cdot}\trans (\betas - \beta) \mid g(Z_{i \cdot}), h(Z_{i \cdot})}$. The variance of the second distribution is simply 1. Therefore,
\[
\begin{split}
\EE{\sigma_1^2 - \sigma_2^2}
&= \EE{\operatorname{Var}\sqb{\eta_i + Z_{i \cdot}\trans (\betas - \beta)  \mid g(Z_{i \cdot}), h(Z_{i \cdot})} - 1} 
= \EE{\operatorname{Var}\sqb{Z_{i \cdot}\trans (\betas - \beta)  \mid  g(Z_{i \cdot}), h(Z_{i \cdot})}} \\
&\leq \EE{ \p{Z_{i \cdot}\trans (\beta - \betas)}^2}
\lesssim s_\beta \log(p)/N. 
\end{split}
\]
Thus \eqref{eqn:kl_gaussian2} implies that $\EE{\Delta_r^2} \lesssim n s_\beta \log(p)/N$. 

The term $\Delta_y$ is zero when exact support recovery can be achieved. Otherwise, $\Delta_y$ is always upper bounded by 1. Therefore, we can bound $2\EE{ \Delta_r \Delta_y}$ by 
\[
2\EE{ \Delta_r \Delta_y} 
\lesssim \sqrt{\PP{\Delta_y > 0 } \EE{\Delta_r^2}}
 \leq  \delta \sqrt{\frac{s_\beta \log(p) n}{N}}. 
\]

The term $\Delta_{r|g,h}$ corresponds to a linear regression on $|\mcS|$ covariates, the rate of convergence of which is given by $\sqrt{ | \mcS| / N}$. But $|\mcS| = s_\theta$ with high probability. Thus $\EE{\Delta_{r|g,h}} \asymp \sqrt{n s_\theta/N}$. 

Finally, for the model-X CRT, assume that we sample $X_i\supm$ independently from $\mathcal{N}(Z_{i \cdot} \trans \beta, \widehat{\sigma}_x^2)$, where $\widehat{\sigma}_x^2$ is the sample variance of $\x - \Z\beta$ on (external) data of sample size $N$. Adopting notation from \eqref{eqn:d_x_prime} and making use of \eqref{eqn:original_CRT_bound_var}, we have that
\[
\begin{split}
\PP{p_{\operatorname{mx}}(\bD) \leq \alpha} 
 &\leq \alpha + \EE{d_{\operatorname{TV}}\p{f_{\x \mid \Z}( \cdot \mid \Z) , f_{\x \mid h(\Z)}( \cdot \mid h(\Z)) }} +\\
& \qquad \qquad \qquad \qquad  \EE{d_{\operatorname{TV}}\p{f_{\x \mid h(\Z)}( \cdot \mid h(\Z)) , \widetilde{\rho}^n_{\x \mid \Z}( \cdot \mid \Z) }},\\
& = \alpha + \EE{\Delta_x'} +  \EE{\Delta_{x|g,h}'}.
\end{split}
\]
If we redo the arguments for $d_r$ but replacing all ``conditional on $g$ and $h$" by ``conditional on $h$", we can show that $\EE{\Delta_{x}'} \lesssim \sqrt{n s_\beta \log(p)/N}$. Similarly, for $\Delta_{x|g,h}'$, we can redo the arguments for $d_{\rho}$ but replacing all ``conditional on $g$ and $h$" by ``conditional on $h$", and show that $\EE{\Delta_{x|g,h}'} \lesssim \sqrt{n /N}$.

\section{Power discrepancy between the model-X and Maxway CRT}\label{sec:app:sym}

{\darkred

In this section, we studied the issue of power discrepancy between the model-X and the Maxway approaches observed in most of our simulation studies. For demonstration, we take Configuration (SS.I) with gaussian linear $X\mid Z$ and $Y\mid Z$ as an example. Recall that if we set the ``overlapping" parameter $\eta=0$, the function $h(X,Z)=X$ for partial linear effect, and all the random signs $\nu_j=1$ for simplicity, we will have 
\begin{equation}
X=0.3\sum_{j=1}^5 Z_j+\epsilon_1;\quad Y=\gamma X+0.3\sum_{j=1}^5 Z_j+\epsilon_2,
\label{equ:data:gen}
\end{equation}
where $\epsilon_1,\epsilon_2\sim\mathcal{N}(0,1)$ are independent noises. The model-X CRT first fits lasso for $X\sim Z$ and $Y\sim Z$ to derive the coefficients:
\begin{align*}
\widehat\gamma_{xz}=&{\rm argmin}_{\gamma_{xz}}(2N)^{-1}\|\x-\Z\gamma_{xz}\|_2^2+\lambda_{x}\|\gamma_{xz}\|_1;\\
\widehat\gamma_{yz}=&{\rm argmin}_{\gamma_{xz}}(2n)^{-1}\|\y-\Z\gamma_{yz}\|_2^2+\lambda_{y}\|\gamma_{yz}\|_1.
\end{align*}
Then it constructs the d$_0$ statistic as $|T\subcrt|$ where $T\subcrt=(\x-\Z\widehat\gamma_{xz})\trans(\y-\Z\widehat\gamma_{yz})$, and
\[
\EE{X\mid Z}=\Z\trans\gamma_{xz}^*,\quad\bepsilon_x=\x-\Z\trans\gamma_{xz}^*,\quad\EE{Y\mid Z}=\Z\trans\gamma_{yz}^*,\quad\bepsilon_y=\y-\gamma\bepsilon_x-\Z\trans\gamma_{yz}^*.
\]
In the Maxway CRT, $X$'s predictor $\Z\widehat\gamma_{xz}$ is actually replaced by 
\[
\Z\widehat\gamma_{xz}+\EEhat{R\mid g(\Z)}=\Z\widehat\gamma_{xz}+\EEhat{X-Z\trans\widehat\gamma_{xz}\mid g(\Z)}=\Z(\widehat\gamma_{xz}+\widehat\gamma_{rz}),
\]
where $g(\Z)=(\bZ\widehat{\gamma}_{yz},\bZ_{\bullet,\mathrm{top}(k)})$ and $\widehat\gamma_{rz}$ satisfies $\Z\widehat\gamma_{rz}=\EEhat{X-Z\trans\widehat\gamma_{xz}\mid g(\Z)}$. So the d$_0$ statistic of the Maxway CRT is taken as $|T\submax|$ where $T\submax=(\x-\Z\widehat\gamma_{xz}-\Z\widehat\gamma_{rz})\trans(\y-\Z\widehat\gamma_{yz})$; see Implementation example \ref{example:1} for more details. In Figure \ref{fig:hist}, we plot the histograms of the test statistics extracted in the model-X and the Maxway procedures side by side under the above-mentioned setting with $\gamma=0$. Based on these, we now demonstrate the cause of power discrepancy. Heuristically, inspecting the expansion of $T\subcrt$:
\begin{align*}
T\subcrt=\bepsilon_x\trans\bepsilon_y+\gamma\|\bepsilon_x\|_2^2+(\bepsilon_y+\gamma\bepsilon_x)\trans\Z(\gamma_{xz}^*-\widehat\gamma_{xz})+\bepsilon_x\trans\Z(\gamma_{yz}^*-\widehat\gamma_{yz})+(\gamma_{xz}^*-\widehat\gamma_{xz})\trans\Z\trans\Z(\gamma_{yz}^*-\widehat\gamma_{yz}),
\end{align*}
we could neglect $(\bepsilon_y+\gamma\bepsilon_x)\trans\Z(\gamma_{xz}^*-\widehat\gamma_{xz})$ since $(\bepsilon_y+\gamma\bepsilon_x)\trans\Z/n$ concentrates to $\bzero$ and $\gamma_{xz}^*-\widehat\gamma_{xz}$ is relatively small because $\widehat\gamma_{xz}$ is estimated with significantly larger than $n$ samples in the SSL scenario. Then the CRT samples $\x\supm=\Z\widehat\gamma_{xz}+\bepsilon_x\supm$ where $\bepsilon_x\supm\sim\Nsc(\bzero,\|\x-\Z\widehat\gamma_{xz}\|_2^2/n)$ and extract the $p$-value as:
\[
2\Phi\left(-\frac{\bepsilon_x\trans\bepsilon_y+\gamma\|\bepsilon_x\|_2^2+\bepsilon_x\trans\Z(\gamma_{yz}^*-\widehat\gamma_{yz})+(\gamma_{xz}^*-\widehat\gamma_{xz})\trans\Z\trans\Z(\gamma_{yz}^*-\widehat\gamma_{yz})}{\|\y-\Z\widehat\gamma_{yz}\|_2^2}\right),
\]
where $\Phi(\cdot)$ represents the cumulative density of $\Nsc(0,1)$. Note that the terms $\bepsilon_x\trans\bepsilon_y/\|\y-\Z\widehat\gamma_{yz}\|_2^2$ and $\bepsilon_x\trans\Z(\gamma_{yz}^*-\widehat\gamma_{yz})/\|\y-\Z\widehat\gamma_{yz}\|_2^2$ center around zero due to the orthogonality between $\bepsilon_x$ and $\{\bepsilon_y,\Z\}$ and have very similar variances between the model-X and the Maxway CRT. This is verified by the histograms in Figure \ref{fig:hist}, in which the test statistics extracted by the two methods empirically show very similar variance. Thus, the power discrepancy between the two methods is mainly driven by the mean shifting terms ``$\gamma\|\bepsilon_x\|_2^2/\|\y-\Z\widehat\gamma_{yz}\|_2^2$" and ``$(\gamma_{xz}^*-\widehat\gamma_{xz})\trans\Z\trans\Z(\gamma_{yz}^*-\widehat\gamma_{yz})/\|\y-\Z\widehat\gamma_{yz}\|_2^2$".

Recall that the estimator $\widehat\gamma_{yz}$ remains to be the same between the model-X and Maxway procedures. So there is no difference between the two methods in terms of $\gamma\|\bepsilon_x\|_2^2$ and $\|\y-\Z\widehat\gamma_{yz}\|_2^2$. Therefore, we only need to compare the term $(\gamma_{xz}^*-\widehat\gamma_{xz})\trans\Z\trans\Z(\gamma_{yz}^*-\widehat\gamma_{yz})/n$ (the model-X) with $(\gamma_{xz}^*-\widehat\gamma_{xz}-\widehat\gamma_{rz})\trans\Z\trans\Z(\gamma_{yz}^*-\widehat\gamma_{yz})/n$ (the Maxway). Under model (\ref{equ:data:gen}), $\gamma_{xz}^*=0.3(\mathbf{1}_5\trans,\mathbf{0}_{p-5}\trans)\trans$ and $\gamma_{yz}^*=0.3(1+\gamma)(\mathbf{1}_5\trans,\mathbf{0}_{p-5}\trans)\trans$ where $|\gamma|<1$. It is known that the lasso estimators $\widehat\gamma_{xz}$ and $\widehat\gamma_{yz}$ tend to shrink to $0$ compared with $\gamma_{xz}^*$ and $\gamma_{yz}^*$ due to regularization. Consistent with this, we found in our simulation that the first five entries in $\gamma_{xz}^*-\widehat\gamma_{xz}$ and $\gamma_{yz}^*-\widehat\gamma_{yz}$ lie between $[0,0.3)$ and $[0,0.3(1+\gamma))$ respectively at most times and the remaining entries are $0$ or very close to $0$. This leads to a non-negligible and positive $(\gamma_{xz}^*-\widehat\gamma_{xz})\trans\Z\trans\Z(\gamma_{yz}^*-\widehat\gamma_{yz})/n$. Compared to this, in the Maxway CRT, the corresponding term ``$(\gamma_{xz}^*-\widehat\gamma_{xz}-\widehat\gamma_{rz})\trans\Z\trans\Z(\gamma_{yz}^*-\widehat\gamma_{yz})/n$" tends to be relatively small and centered around zero because our method picks (some of) the leading confounding covariates in $g(\Z)=(\bZ\widehat{\gamma}_{yz},\bZ_{\bullet,\mathrm{top}(k)})$ and use low dimensional regression against $g(\Z)$ to adjust for the bias caused by shrinkage.

Consequently, one could find in Figure \ref{fig:hist} that under the null model with $\gamma\|\bepsilon_x\|_2^2=0$, the test statistics of the model-X CRT still show a positive mean shifting while the Maxway test statistics center closely to zero. This partially explains why our method achieves better type-I error control under the null model. More importantly, under the alternative model with $\gamma\neq0$, this positive bias will make the effect spuriously larger when $\gamma>0$ (i.e., the same sign with the mean shifting) and smaller when $\gamma<0$ (i.e., the opposite sign). This explains the phenomenon that compared to our method, the model-X CRT has a larger power for positive $\gamma$, a smaller power for negative $\gamma$, and a similar overall power. From this discussion, one could conclude that in terms of power, the Maxway CRT is also more preferable to the model-X since it is more symmetric and balanced between positive and negative effects. 

\begin{figure}[htb!]
    \centering
    \includegraphics[width=0.47\textwidth]{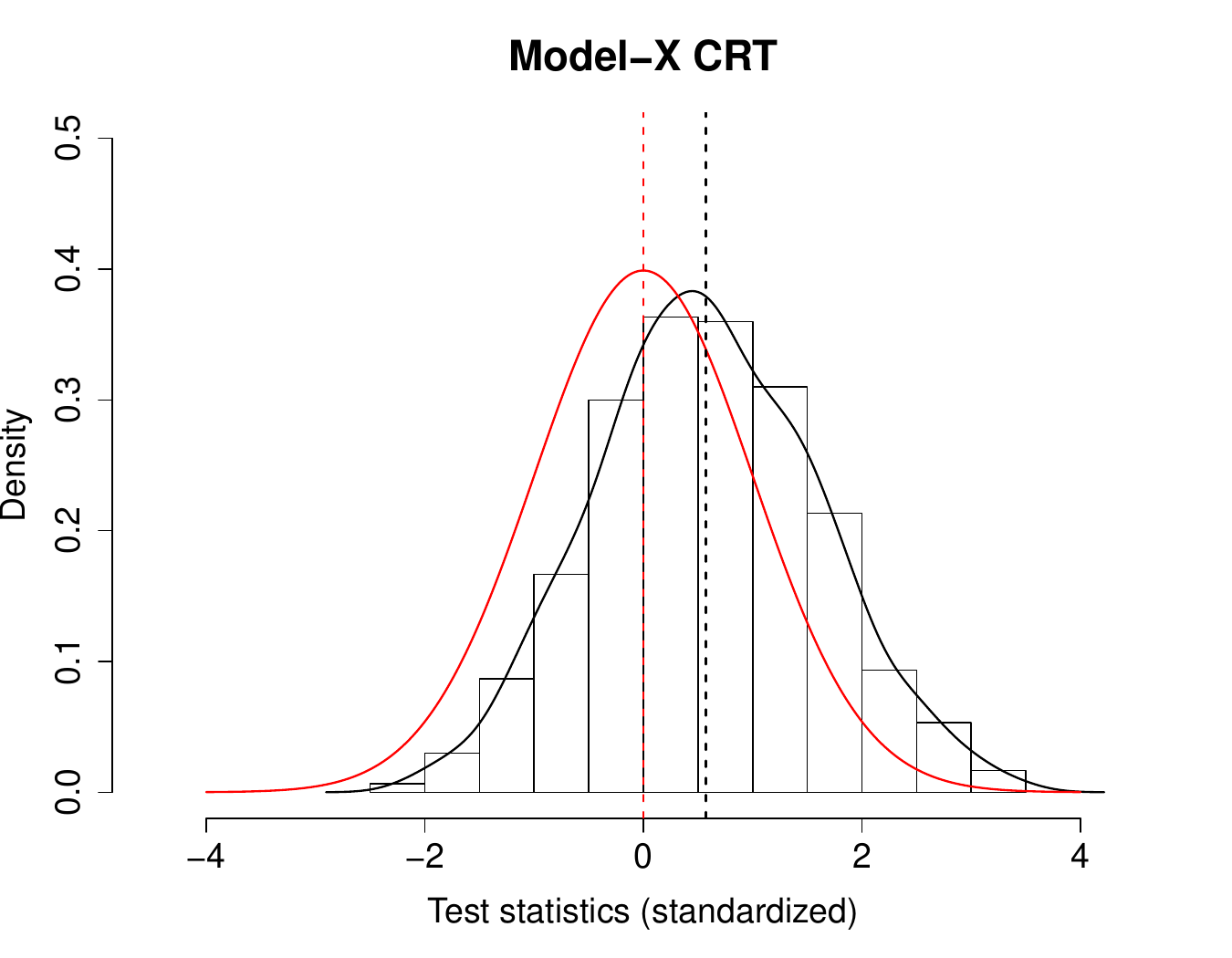}
    \includegraphics[width=0.47\textwidth]{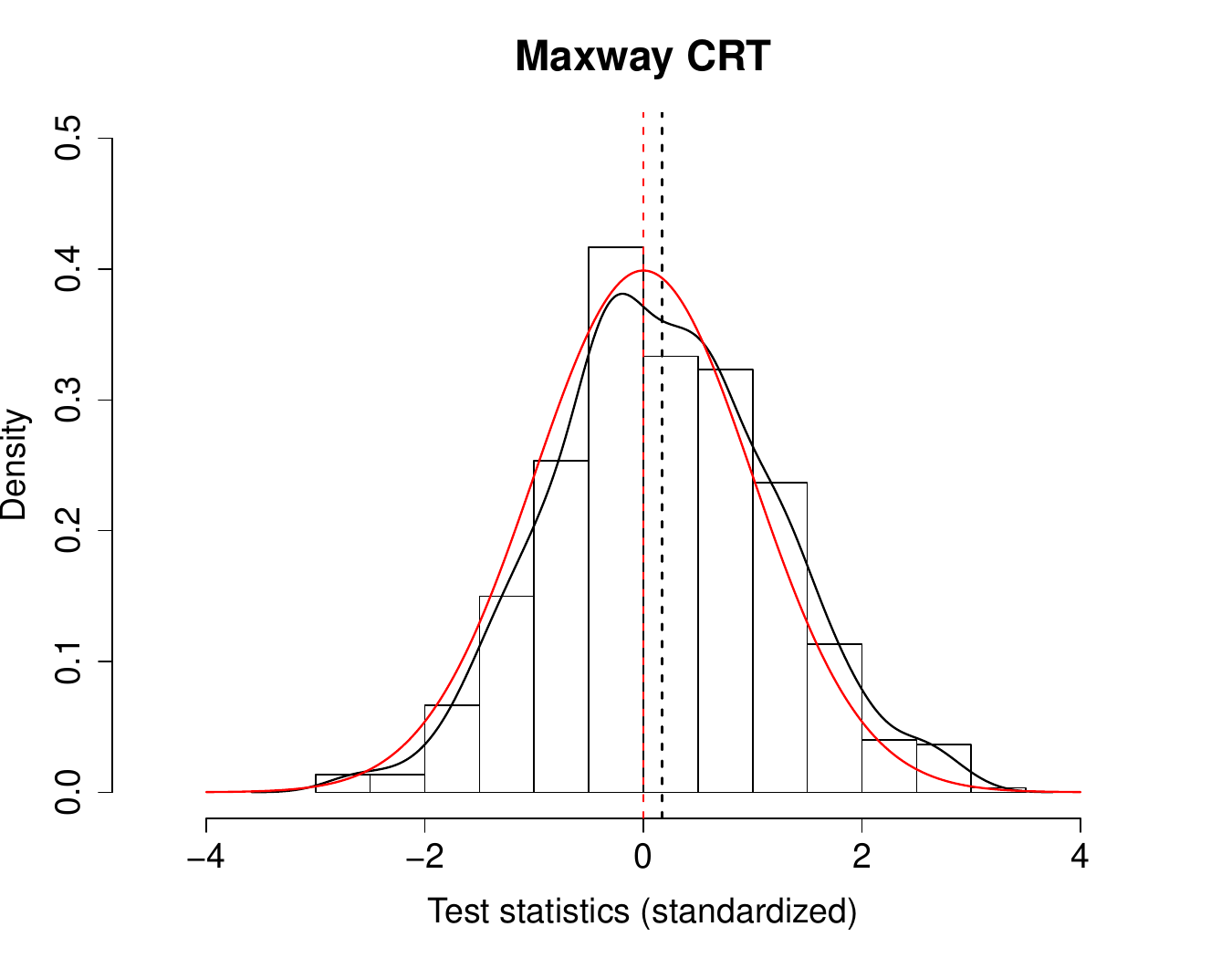}
   \caption{Histograms of the model-X and the Maxway's d$_0$ test statistics under Configuration (SS.I) with the data generation parameters $\eta=0$ and $\gamma=0$. Red curves represent the density function of the standard normal distribution. The dot lines indicate the mean values of the plotted distributions. The replication number is $1000$.} 
    \label{fig:hist}
\end{figure}

For further demonstration, we also conduct an additional simulation study with exactly the same setup of data generation and methods implementation. The only difference is that the test statistics are instead taken as $|T\subcrt-\mbox{mean}(T\subcrt)|$ and $|T\submax-\mbox{mean}(T\submax)|$ where $\mbox{mean}(T\subcrt)$ and $\mbox{mean}(T\submax)$ represent the mean of the observed test statistics under the null model estimated through simulation. In this way, the mean shifting of $\mbox{mean}(T\subcrt)$ and $\mbox{mean}(T\submax)$ discussed above can be effectively removed. The resulting average power of the model-X and Maxway CRT are plotted in Figure \ref{fig:bias:adj}. One can see that after removing the mean shifting term incurred by lasso shrinkage, the two methods turn out to have very close power. This demonstrates that overall, our proposed method has no power loss compared to the model-X CRT.

\begin{figure}[htb!]
    \centering
    \includegraphics[width=0.65\textwidth]{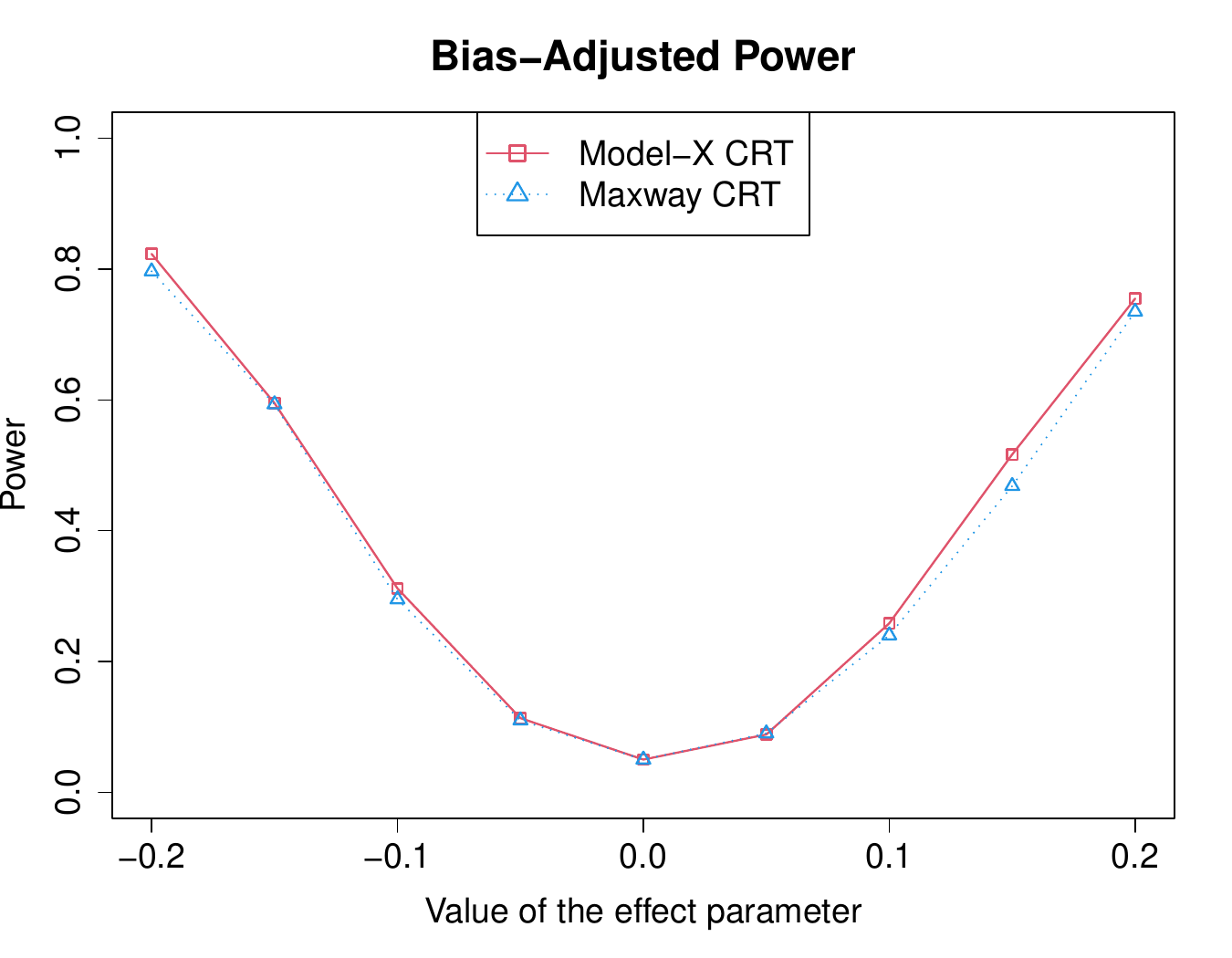}
    \caption{Average power of the mean-shifting-adjusted d$_0$ statistic under Configuration (SS.I) with $\eta=0$ and the effect $\gamma$ varying from $-0.2$ to $0.2$. The replication number is $1000$.}
    \label{fig:bias:adj}
\end{figure}

}

\newpage

\section{Additional numerical results}\label{sec:app:num}

\begin{figure}[htbp!]
    \centering
    \includegraphics[width=0.34\textwidth]{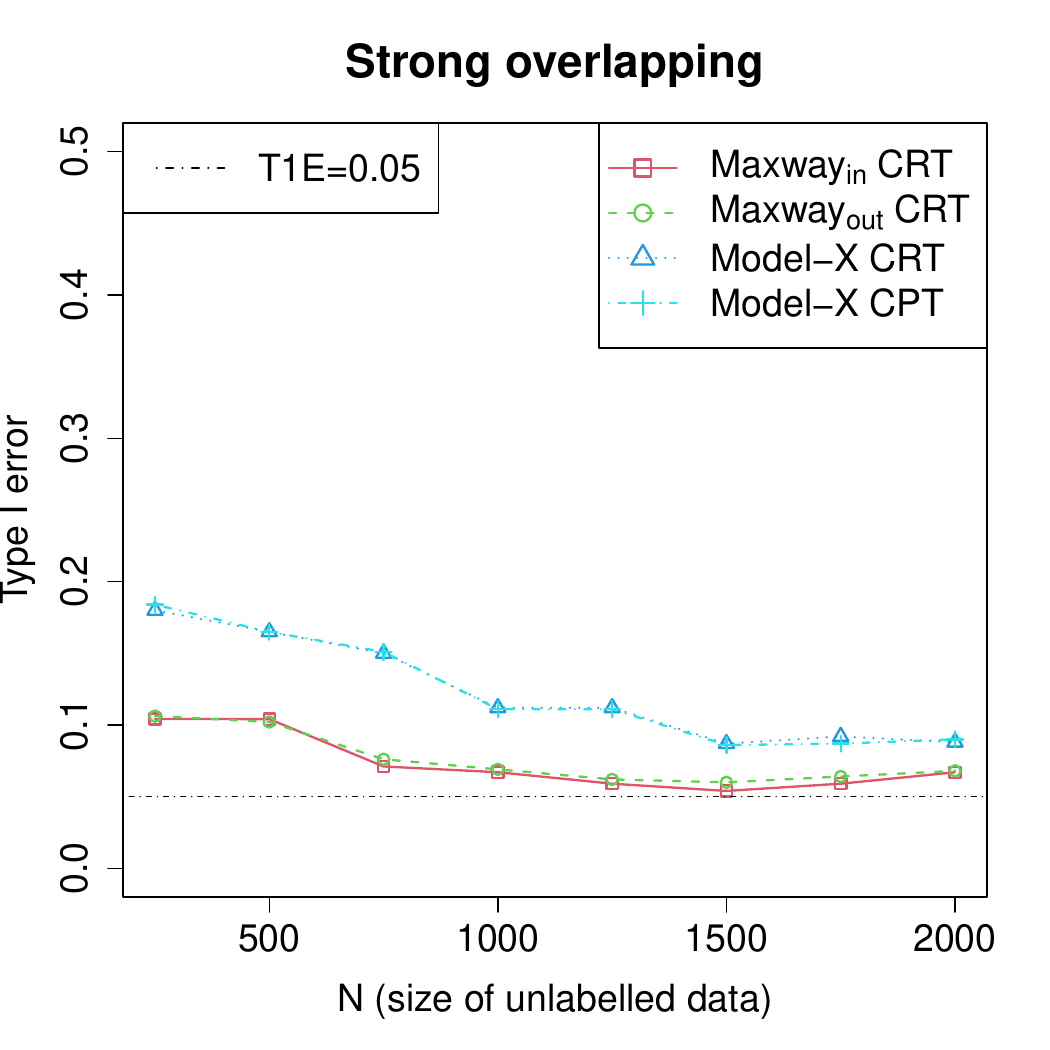}
    \includegraphics[width=0.34\textwidth]{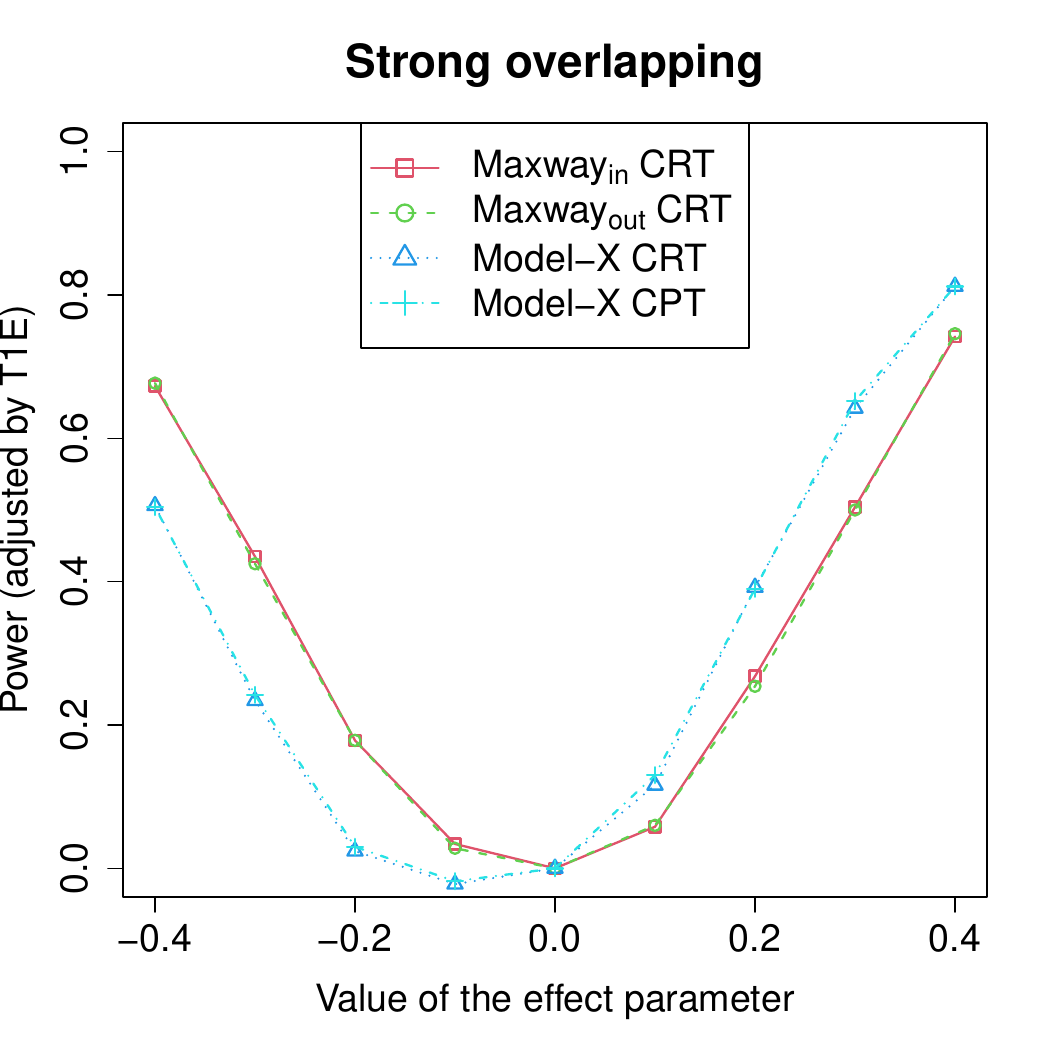}
    \includegraphics[width=0.34\textwidth]{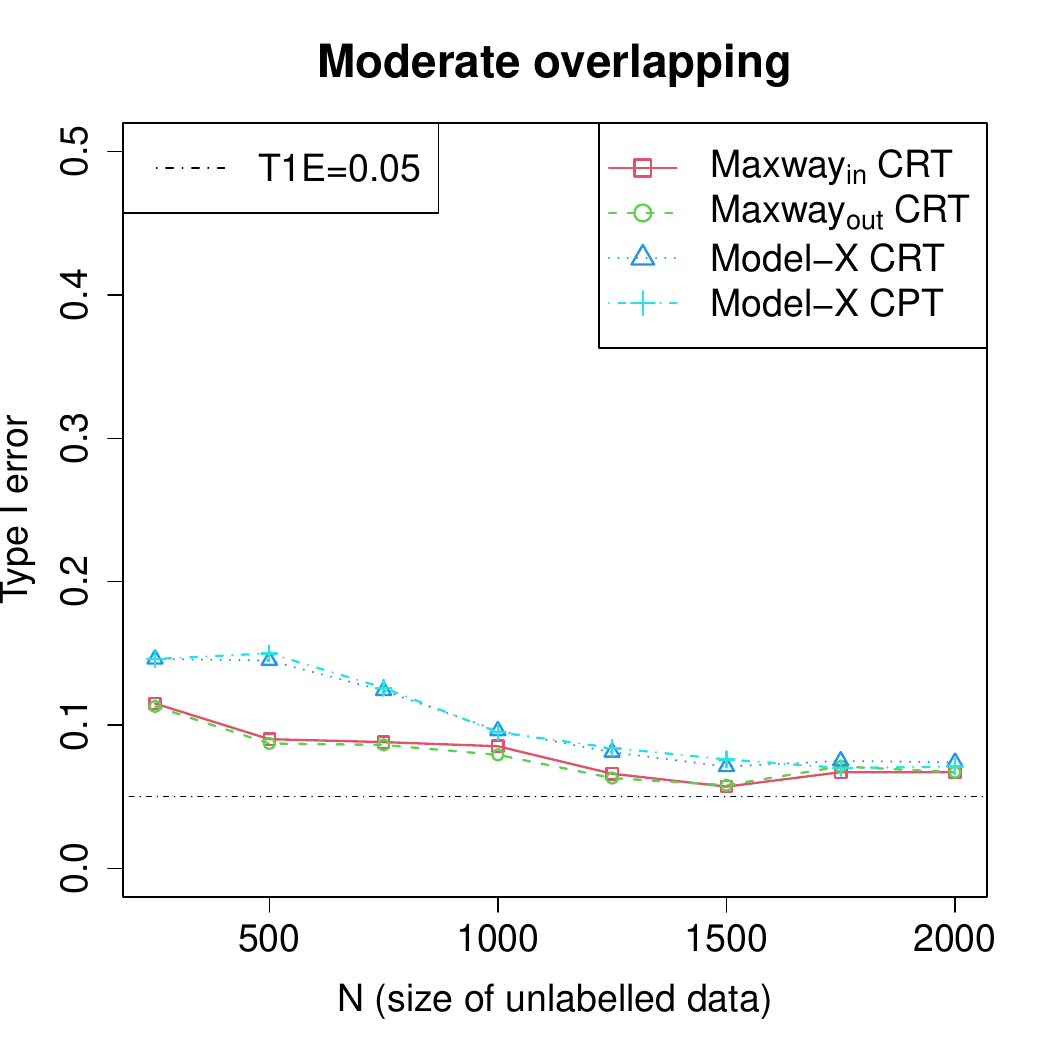}
    \includegraphics[width=0.34\textwidth]{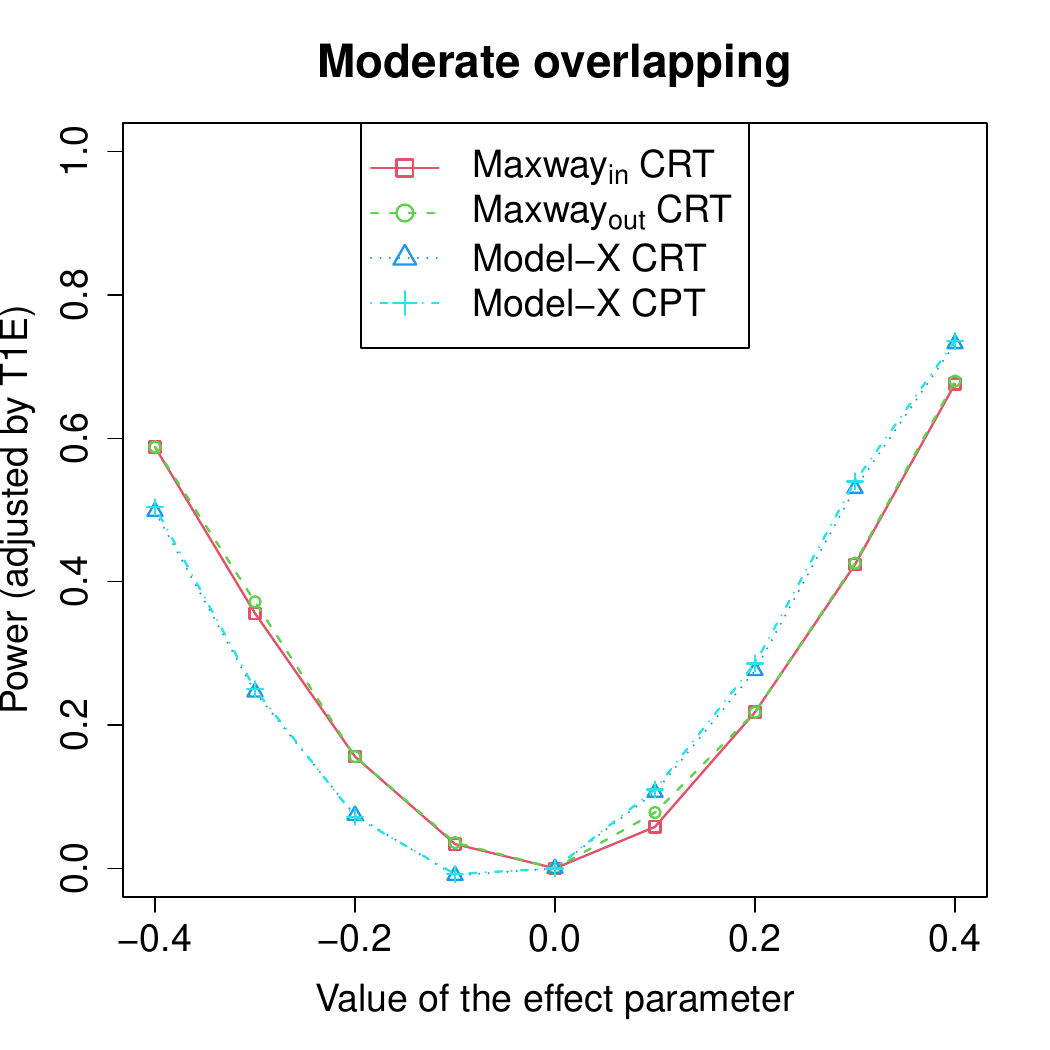}
    \includegraphics[width=0.34\textwidth]{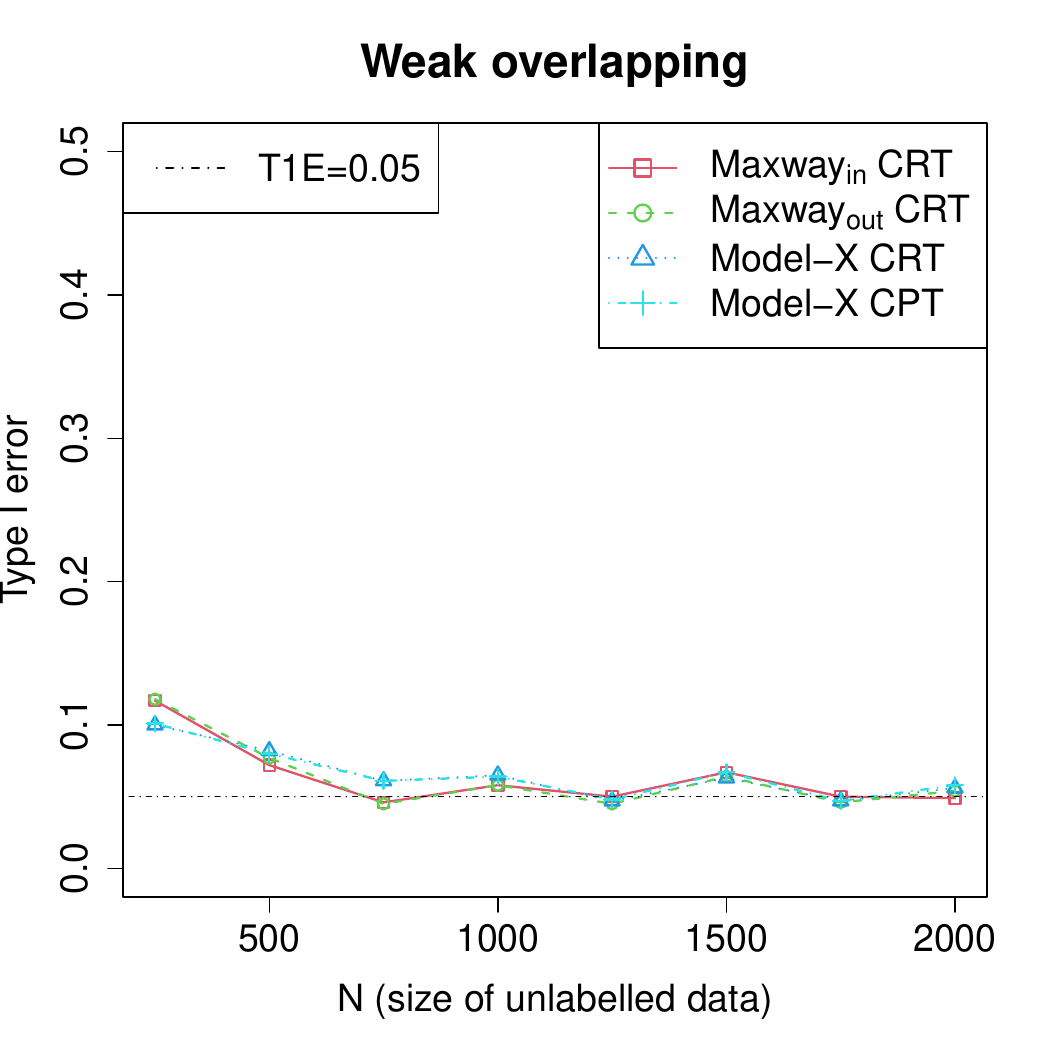}
    \includegraphics[width=0.34\textwidth]{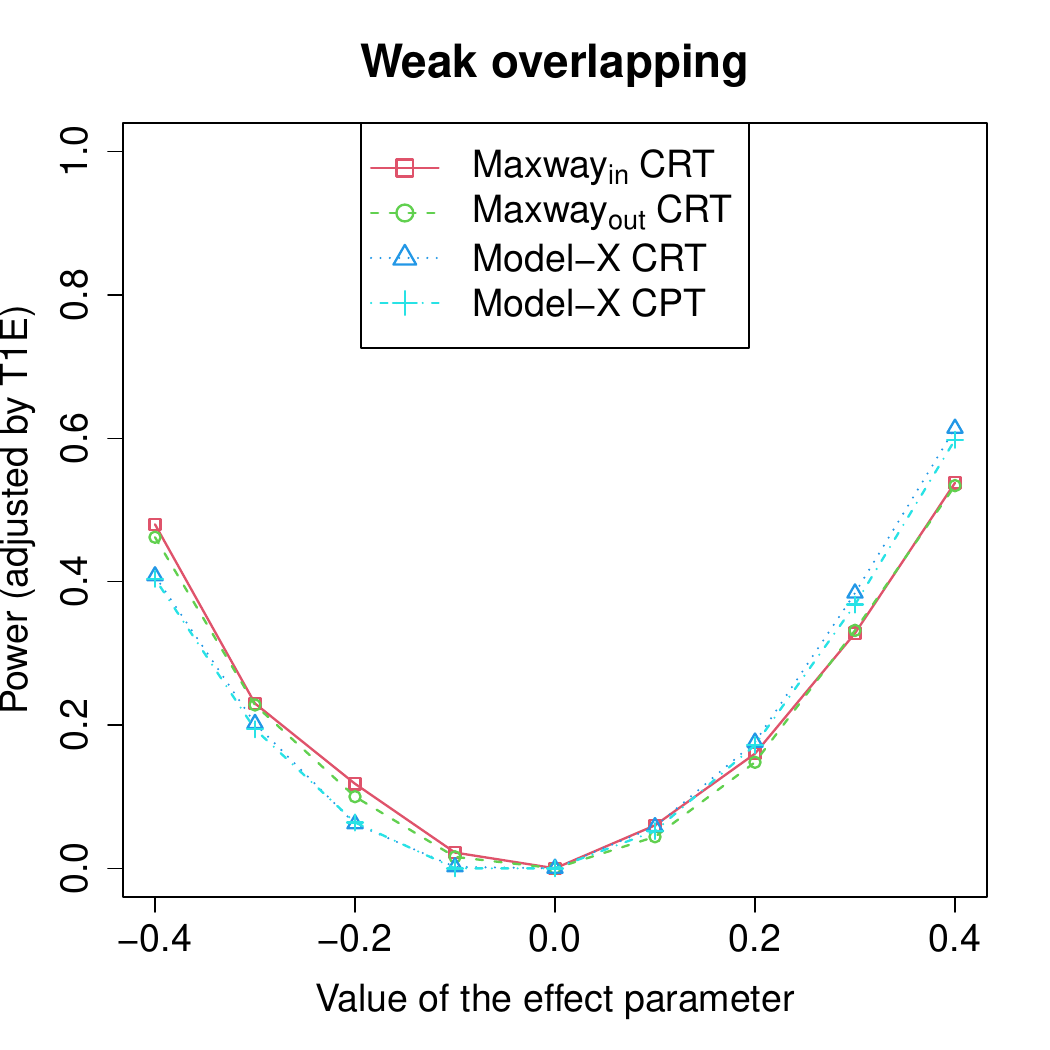}
    \caption{Type-I error and average power (adjusted by type-I error) under the three overlapping scenarios (i.e. $\eta=0,0.1,0.2$) of Configuration (SS.II) {\bf logistic linear $X\mid Z$ and gaussian linear $Y\mid Z$} with $h(X,Z)=X$ and the d$_0$ statistic used for testing, as introduced in Section \ref{sec:sim:ss}. The replication number is $500$ and all standard errors are below $0.01$.} 
    \label{fig:binary}
\end{figure}

\begin{figure}[htb!]
    \centering
    \includegraphics[width=0.41\textwidth]{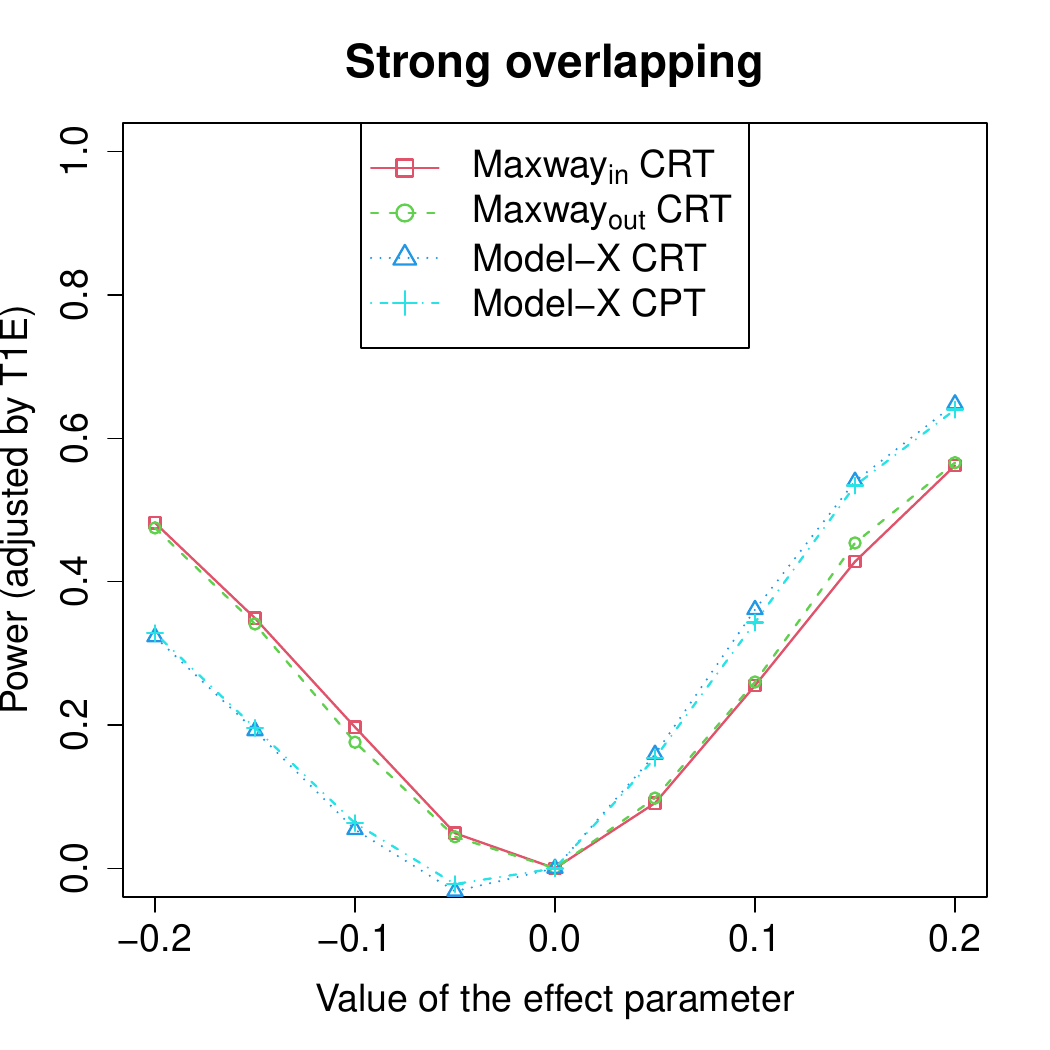}
    \includegraphics[width=0.41\textwidth]{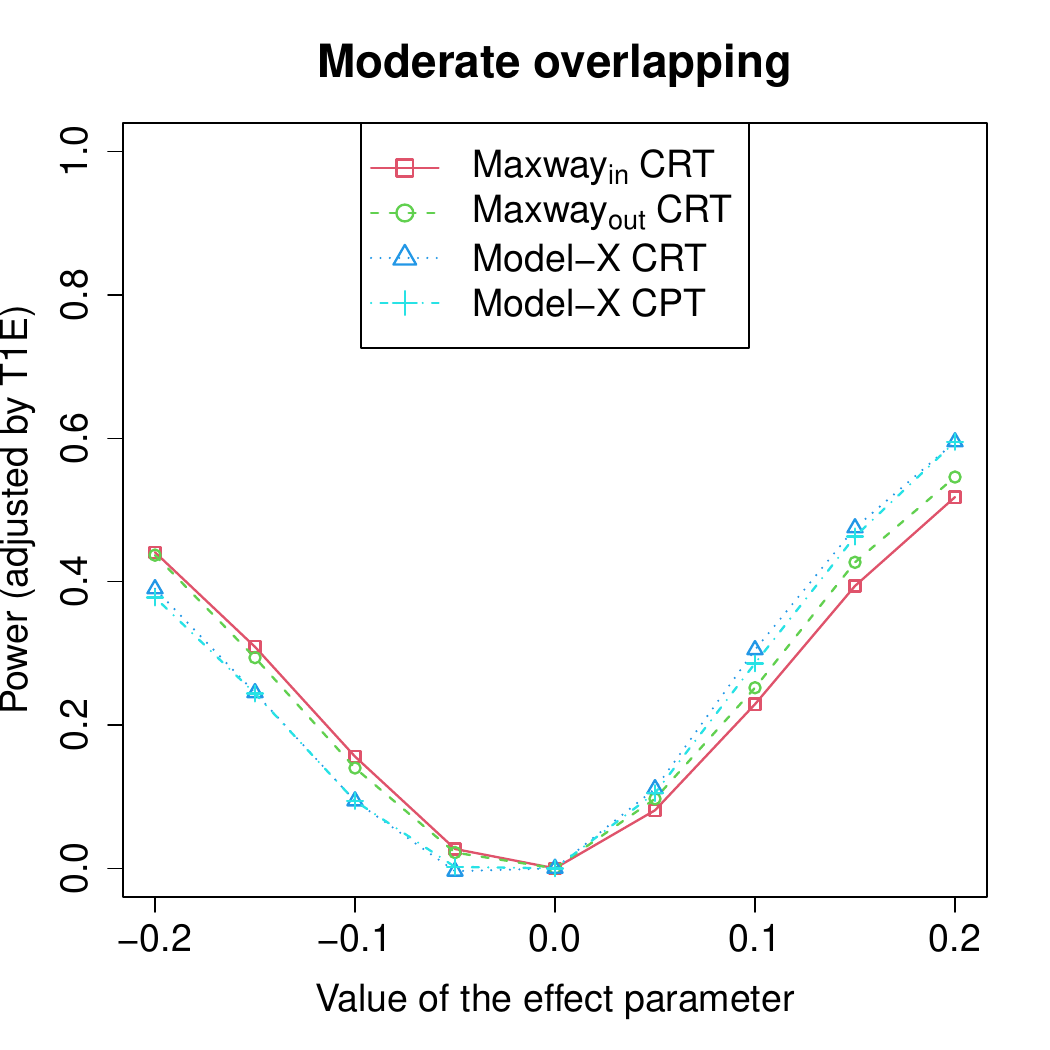}
    \includegraphics[width=0.41\textwidth]{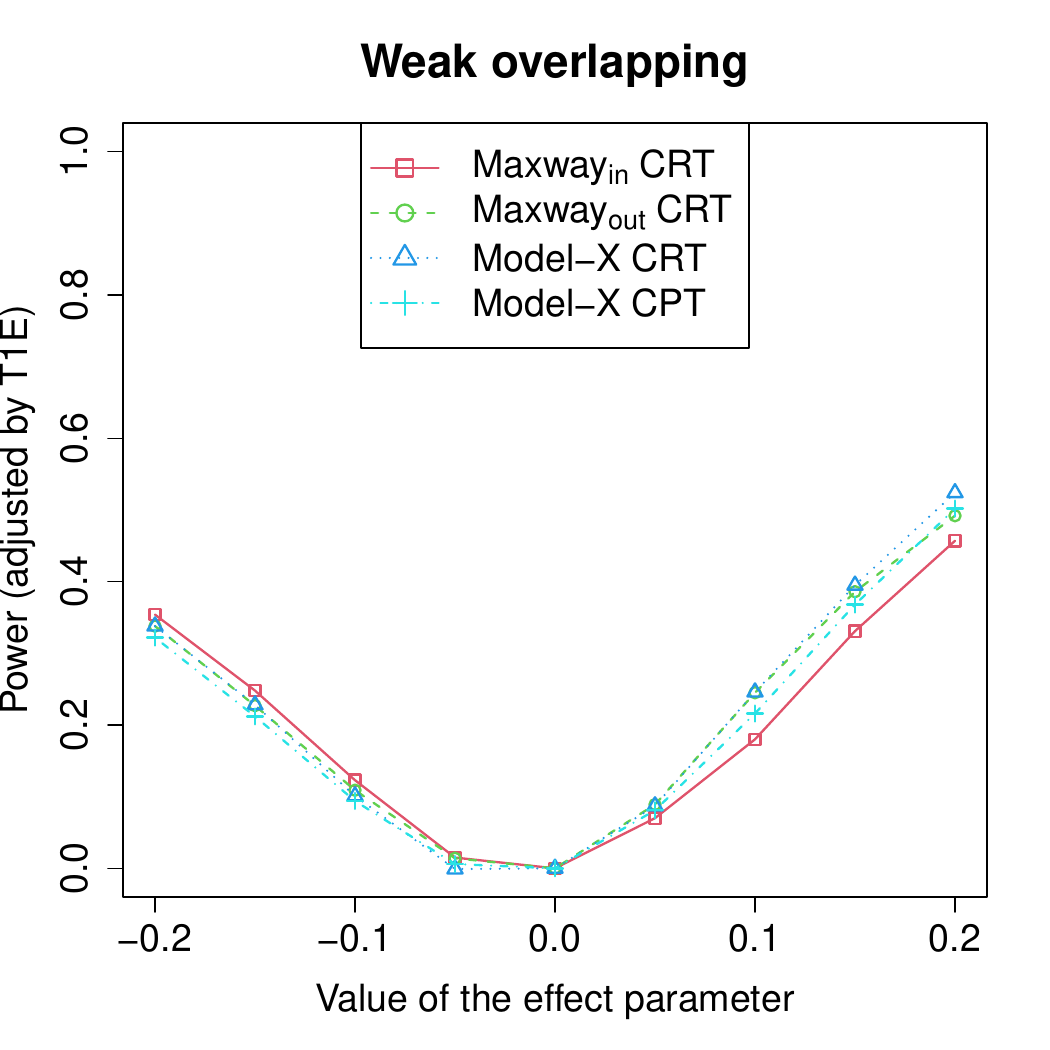}
    \caption{Average power (adjusted by type-I error) under the three overlapping scenarios (i.e. $\eta=0,0.1,0.2$) of Configuration (SS.I) with the effect of $X$ containing interaction: $h(X,\Z)=X+X\sum_{j=1}^5Z_j$ and the d$_0$ statistic used for testing, as introduced in Section \ref{sec:sim}. The type-I error has been presented in the left panel of Figure \ref{fig:gauss}. The replication number is $500$ and all standard errors are below $0.01$.} 
    \label{fig:gauss:add:1}
\end{figure}

\begin{figure}[htbp]
    \centering
    \includegraphics[width=0.41\textwidth]{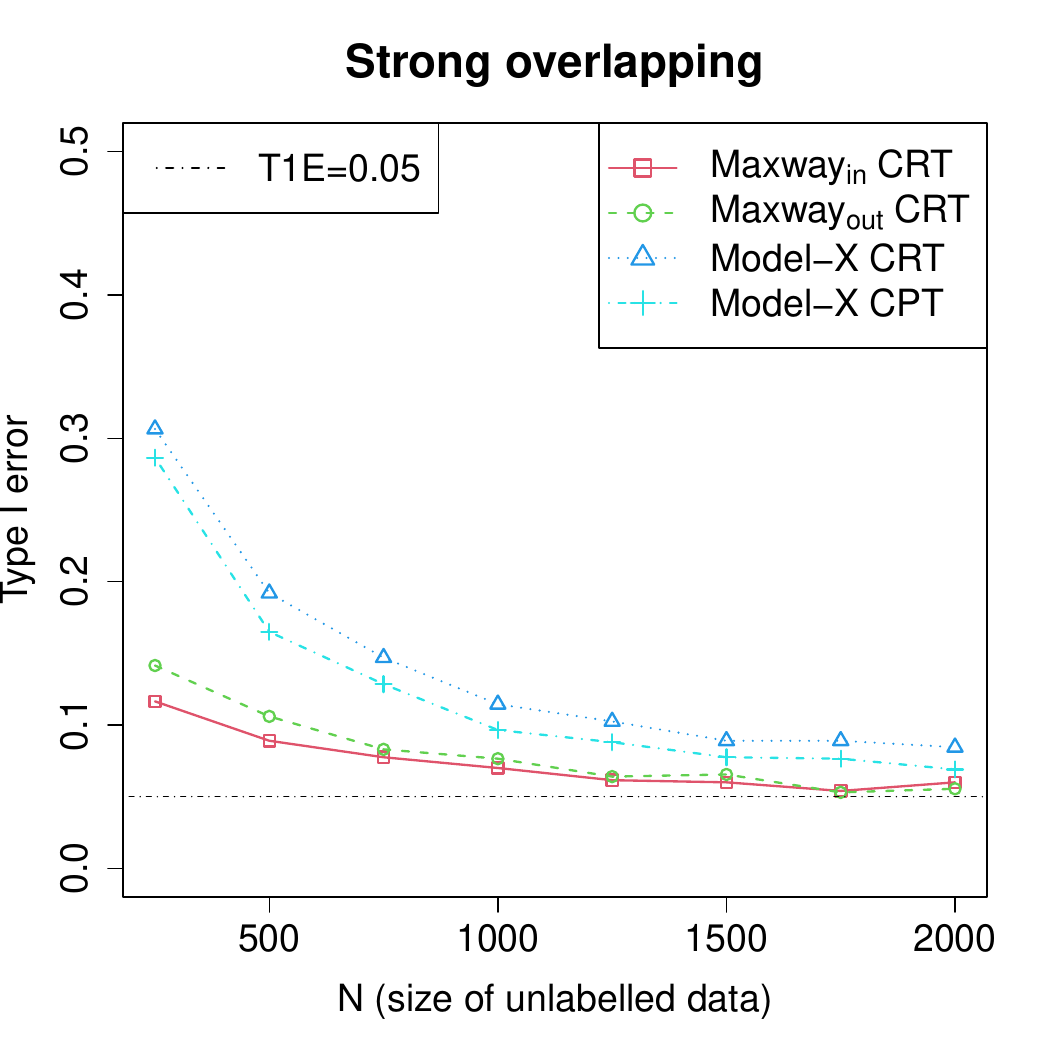}
    \includegraphics[width=0.41\textwidth]{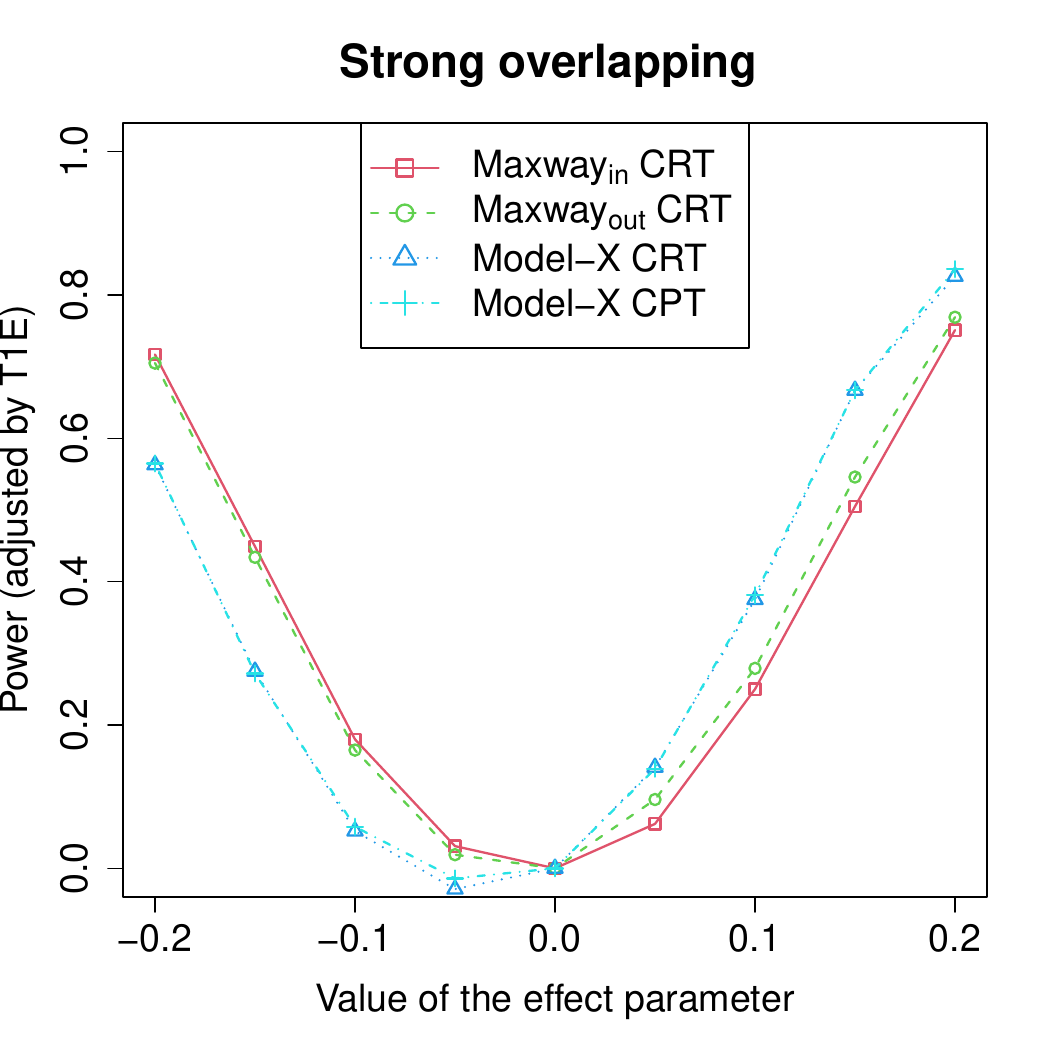}
    \includegraphics[width=0.41\textwidth]{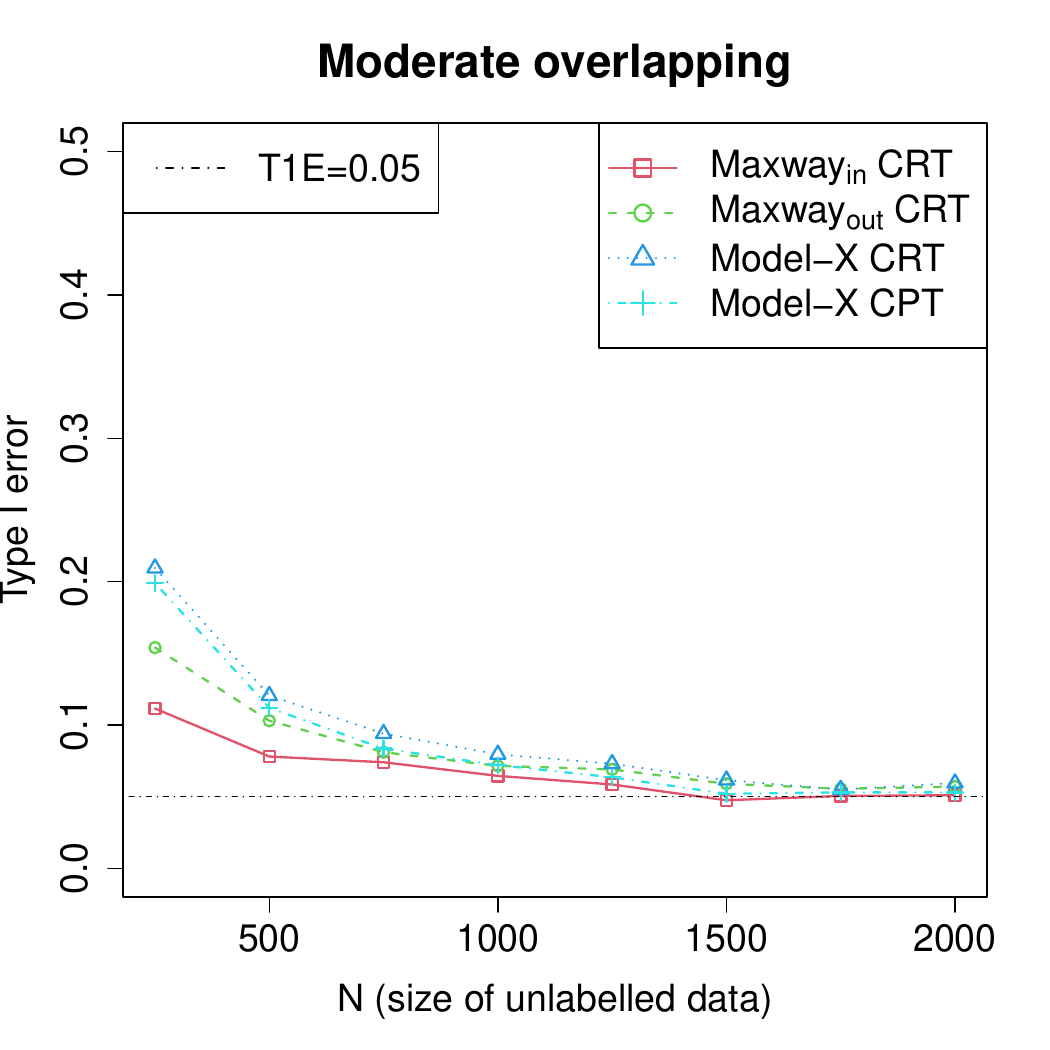}
    \includegraphics[width=0.41\textwidth]{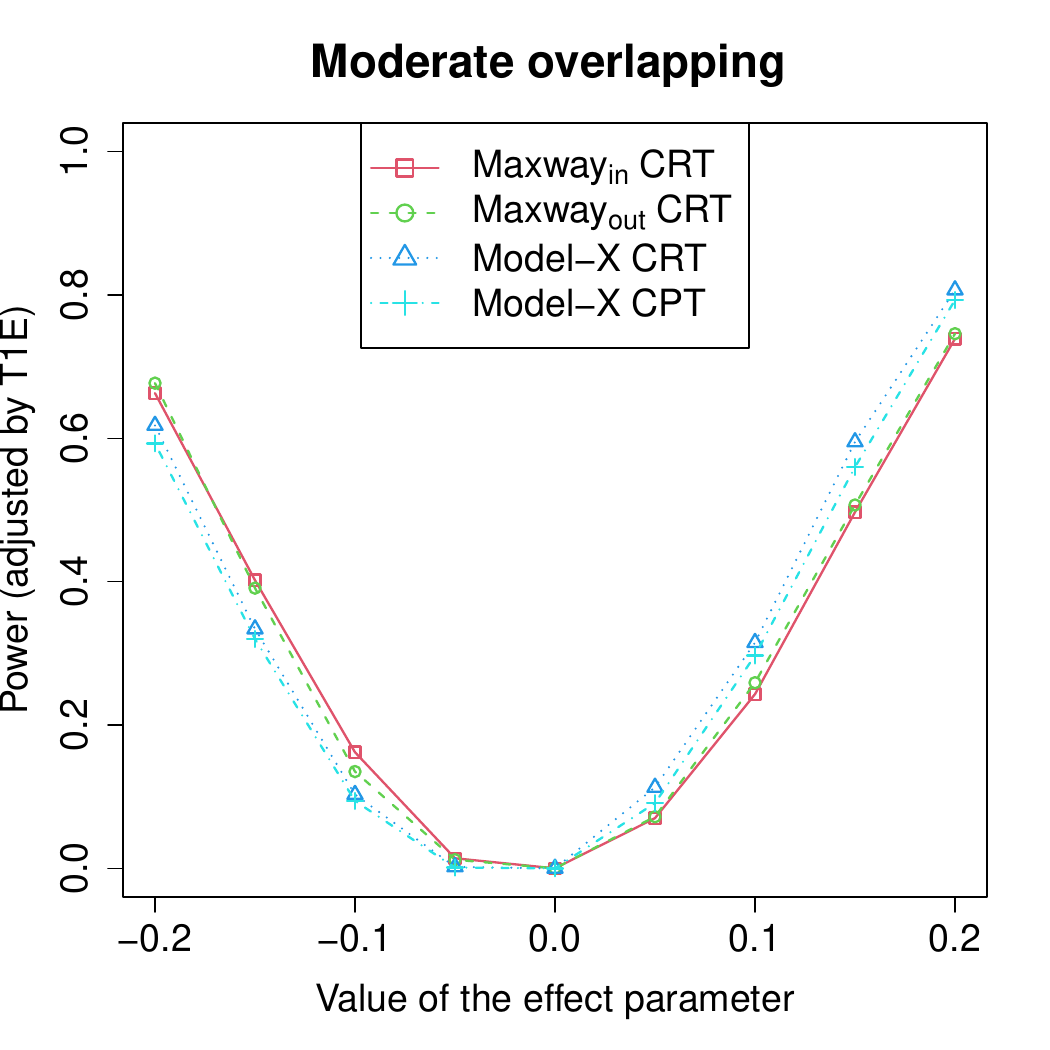}
    \includegraphics[width=0.41\textwidth]{figures/T1E_gaussian_linear_model_magnxy02.pdf}
    \includegraphics[width=0.41\textwidth]{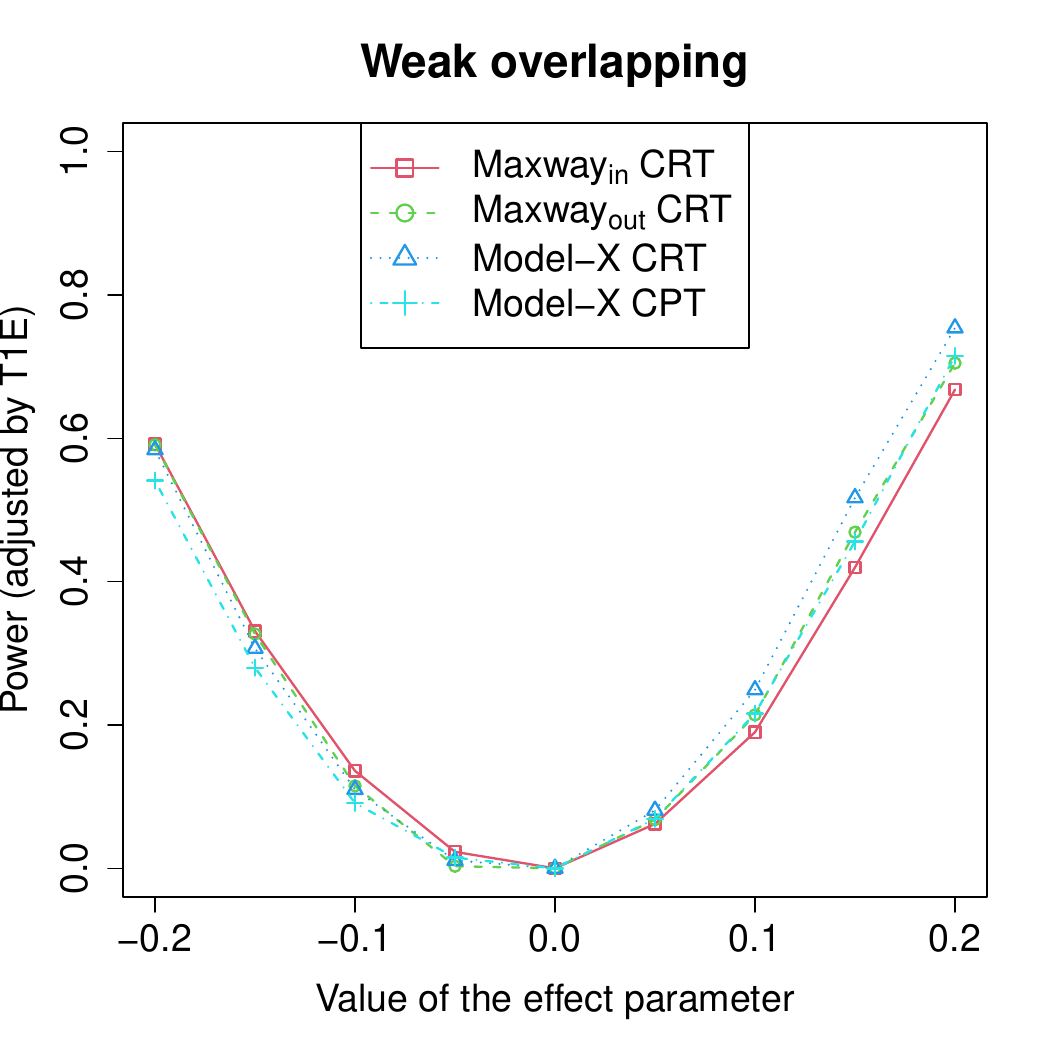}
    \caption{Type-I error and average power (adjusted by type-I error) under the three overlapping scenarios (i.e. $\eta=0,0.1,0.2$) of Configuration (SS.I) with $h(X,Z)=X$ and the d$_{\mathrm{I}}$ statistic used for testing, as introduced in Section \ref{sec:sim}. The replication number is $500$ and all standard errors are below $0.01$.} 
    \label{fig:gauss:add:2}
\end{figure}

\begin{figure}[htbp]
    \centering
    \includegraphics[width=0.41\textwidth]{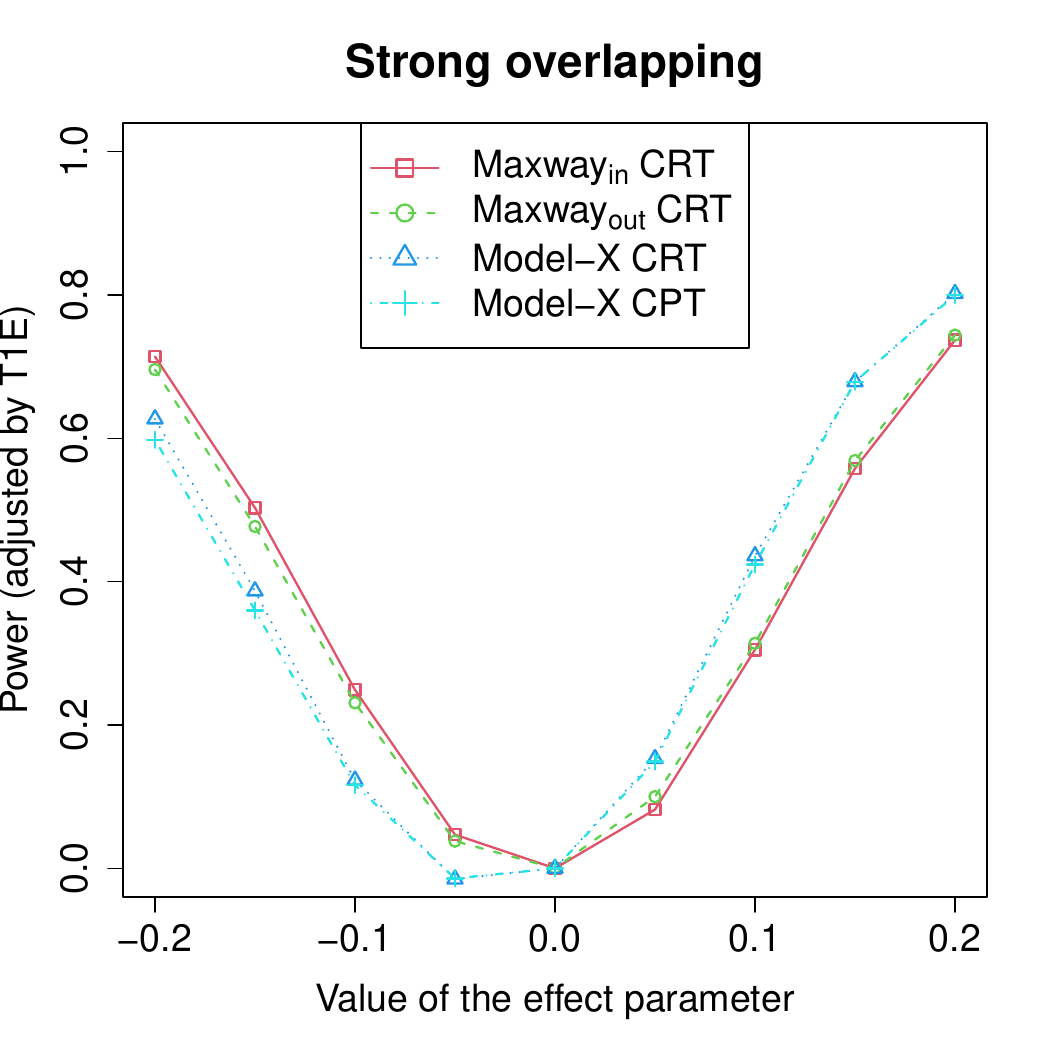}
    \includegraphics[width=0.41\textwidth]{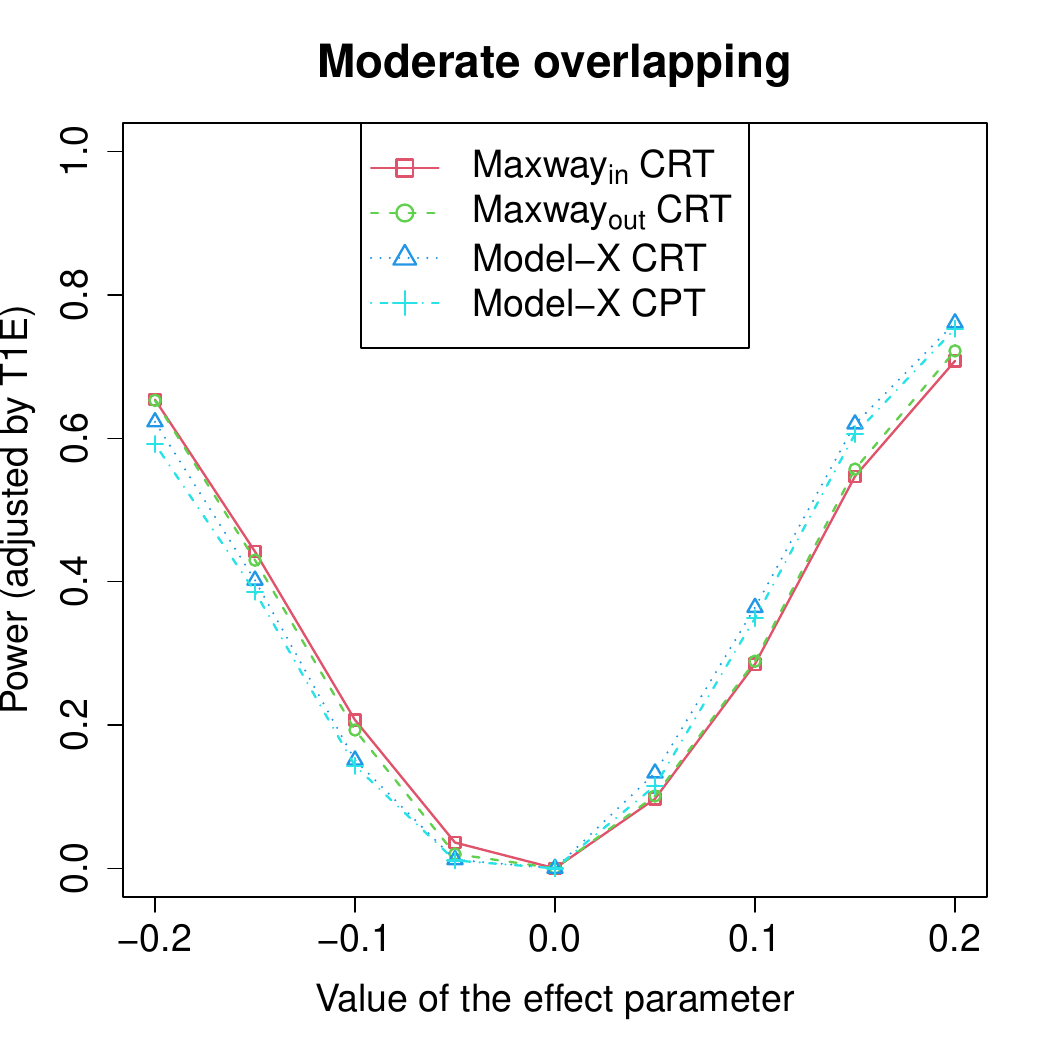}
    \includegraphics[width=0.41\textwidth]{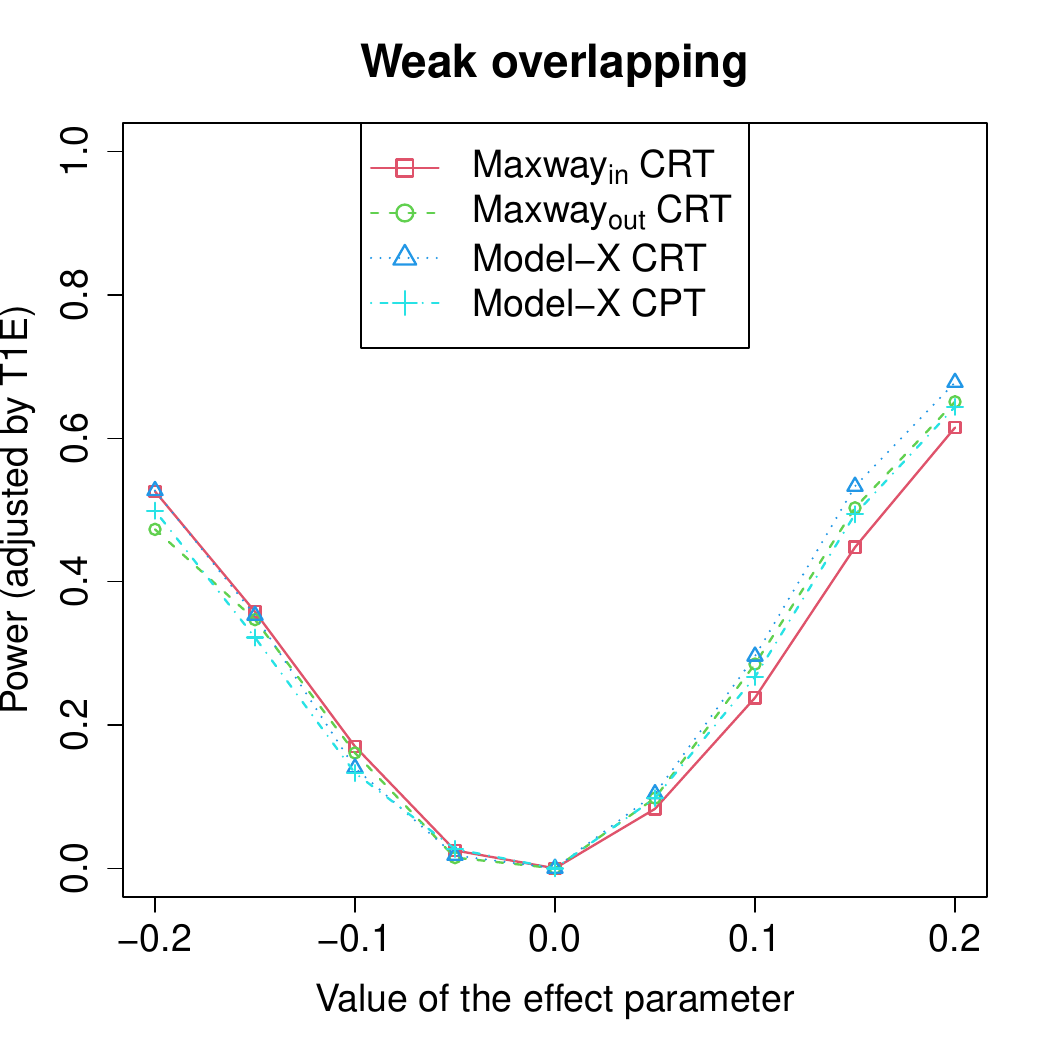}
    \caption{Type-I error and average power (adjusted by type-I error) under the three overlapping scenarios (i.e. $\eta=0,0.1,0.2$) of Configuration (SS.I) with the effect of $X$ containing interaction: $h(X,\Z)=X+X\sum_{j=1}^5Z_j$ and the d$_{\mathrm{I}}$ statistic used for testing, as introduced in Section \ref{sec:sim}. The type-I error has been presented in the left panel of Figure \ref{fig:gauss:add:2}. The replication number is $500$ and all standard errors are below $0.01$.} 
    \label{fig:gauss:add:3}
\end{figure}

\begin{figure}[htbp]
    \centering
    \includegraphics[width=0.41\textwidth]{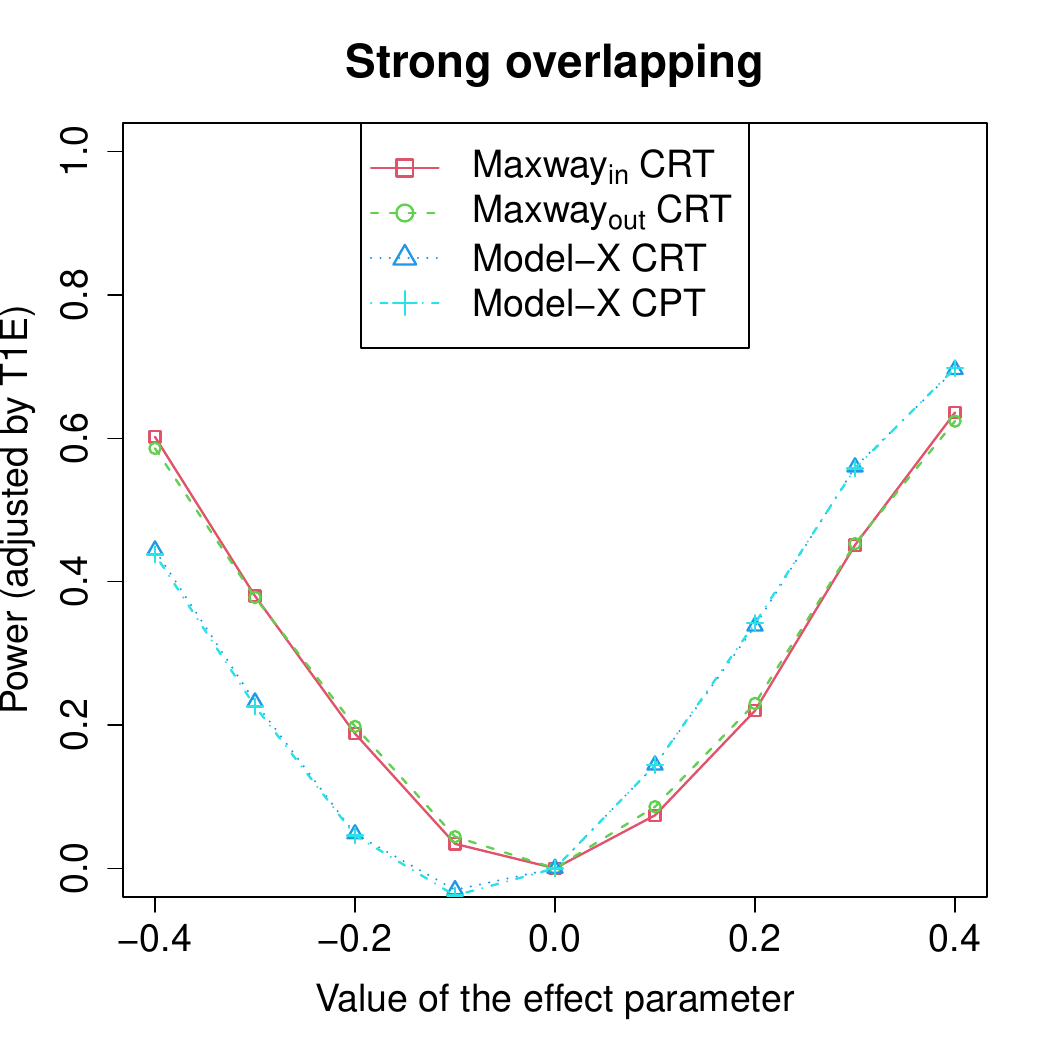}
    \includegraphics[width=0.41\textwidth]{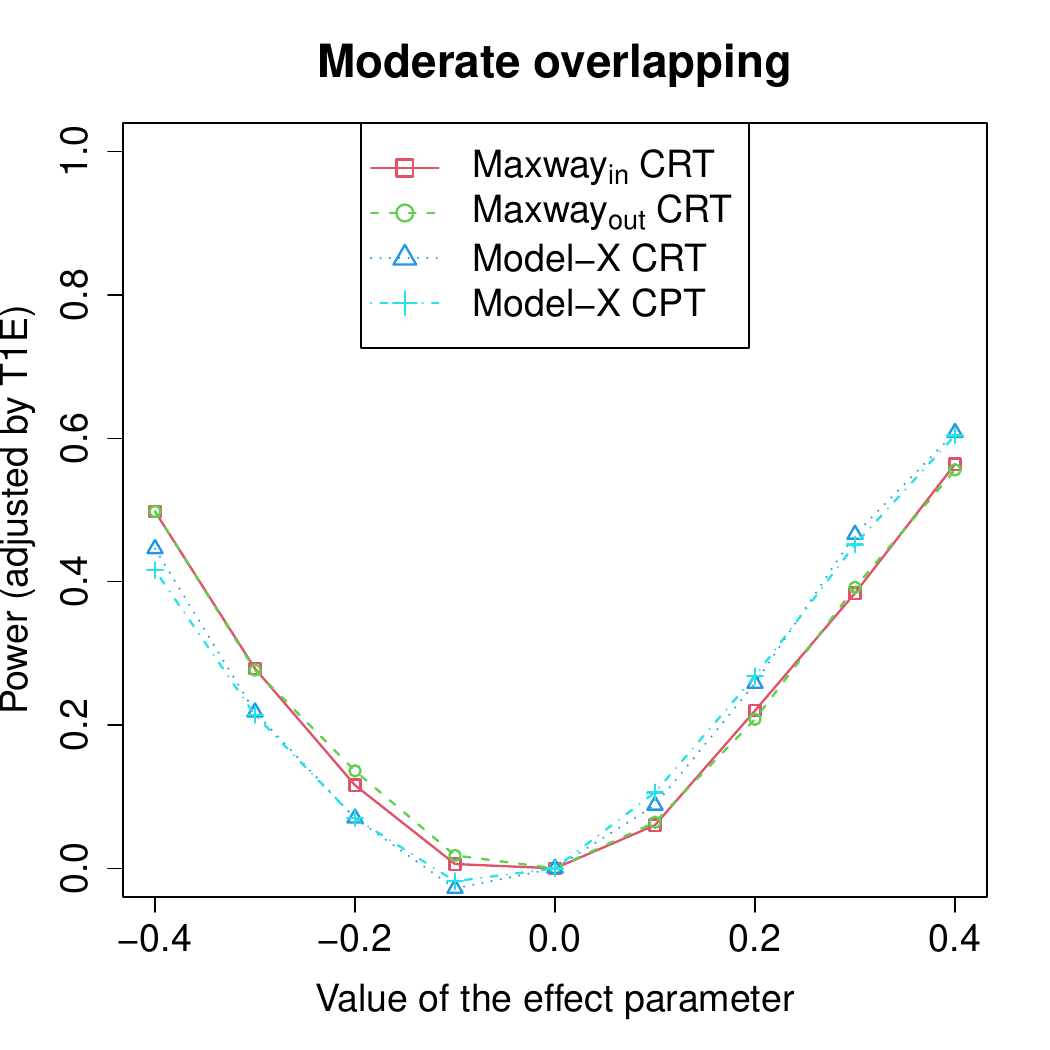}
    \includegraphics[width=0.41\textwidth]{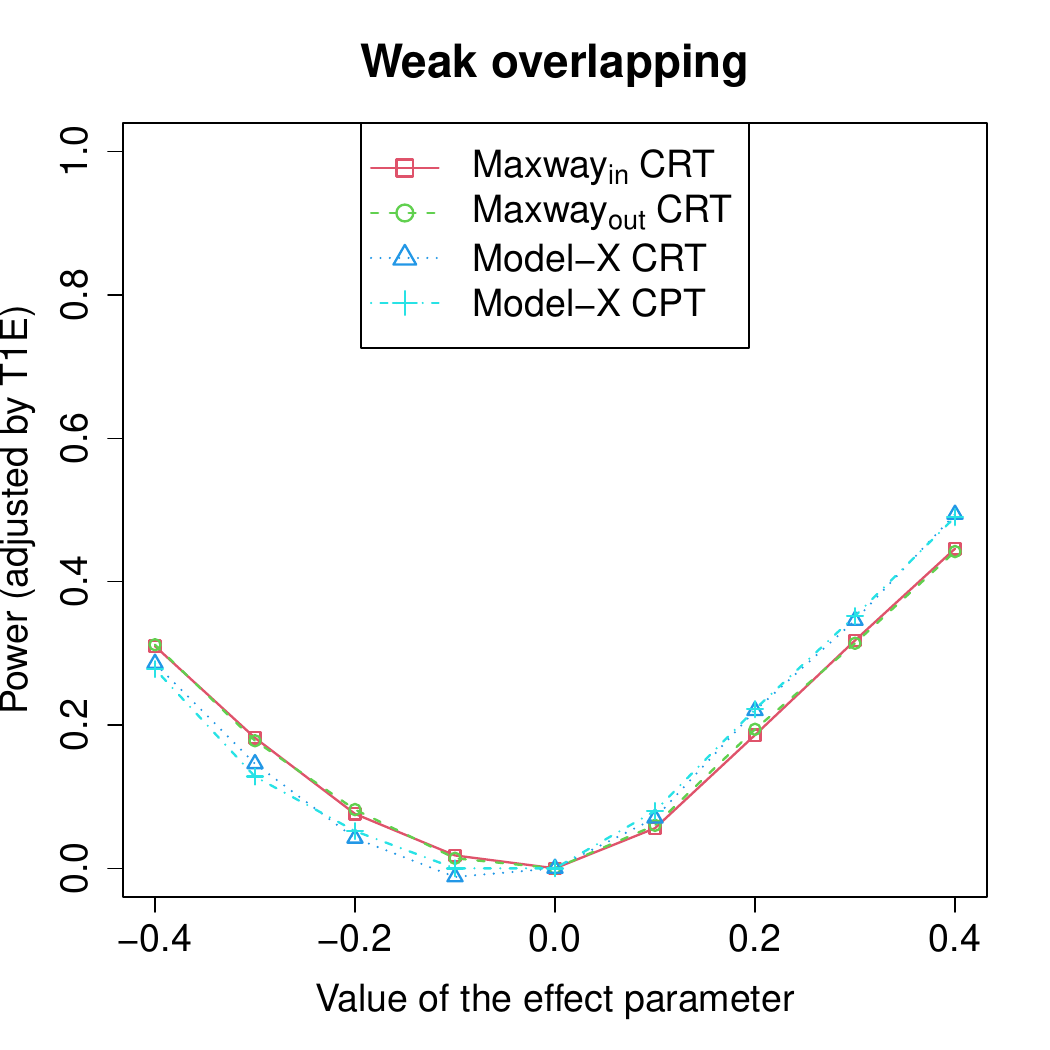}
    \caption{Average power (adjusted by type-I error) under the three overlapping scenarios (i.e. $\eta=0,0.1,0.2$) of Configuration (SS.II) with the effect of $X$ containing interaction: $h(X,\Z)=X+X\sum_{j=1}^5Z_j$ and the d$_0$ statistic used for testing, as introduced in Section \ref{sec:sim}. The type-I error has been presented in the left panel of Figure \ref{fig:binary}. The replication number is $500$ and all standard errors are below $0.01$.}
    \label{fig:binary:add:1}
\end{figure}

\begin{figure}[htbp]
    \centering
    \includegraphics[width=0.41\textwidth]{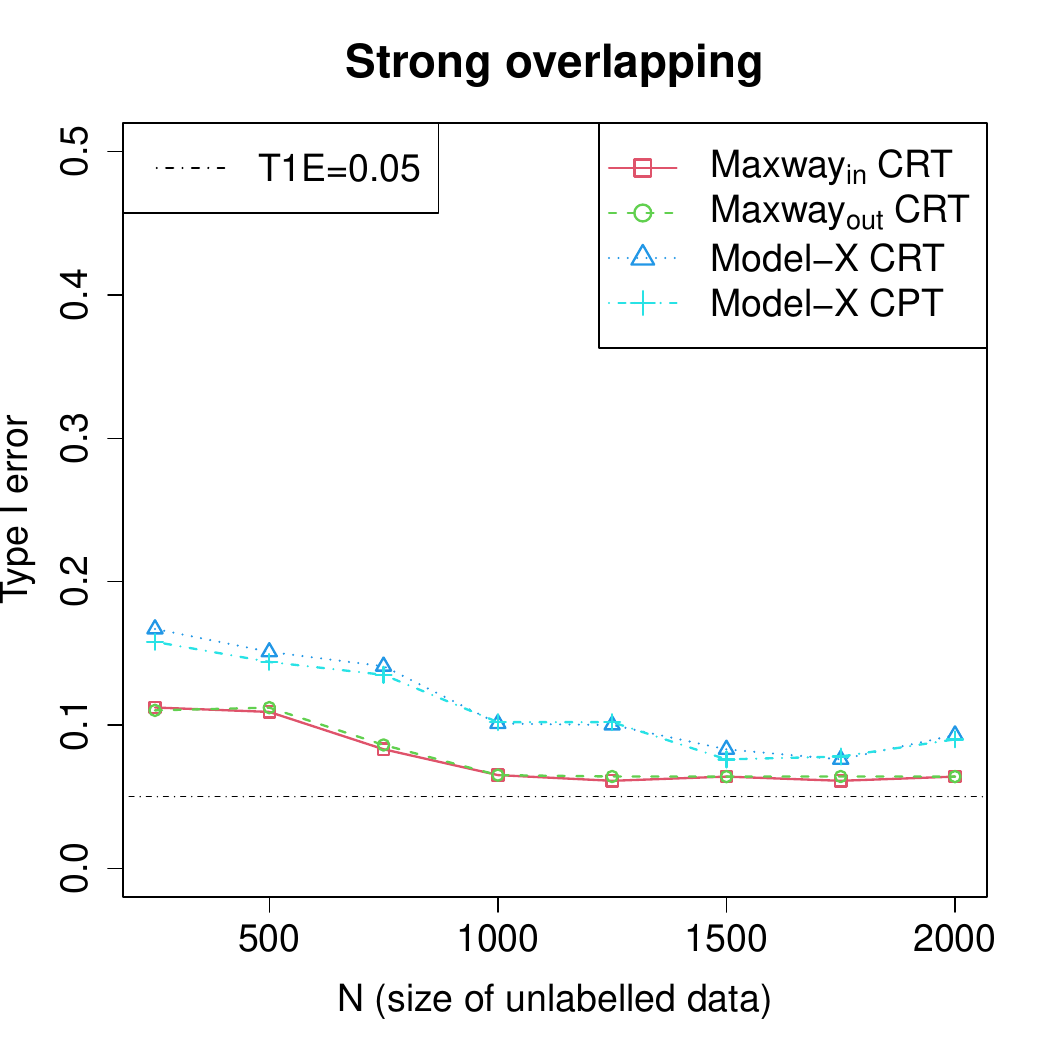}
    \includegraphics[width=0.41\textwidth]{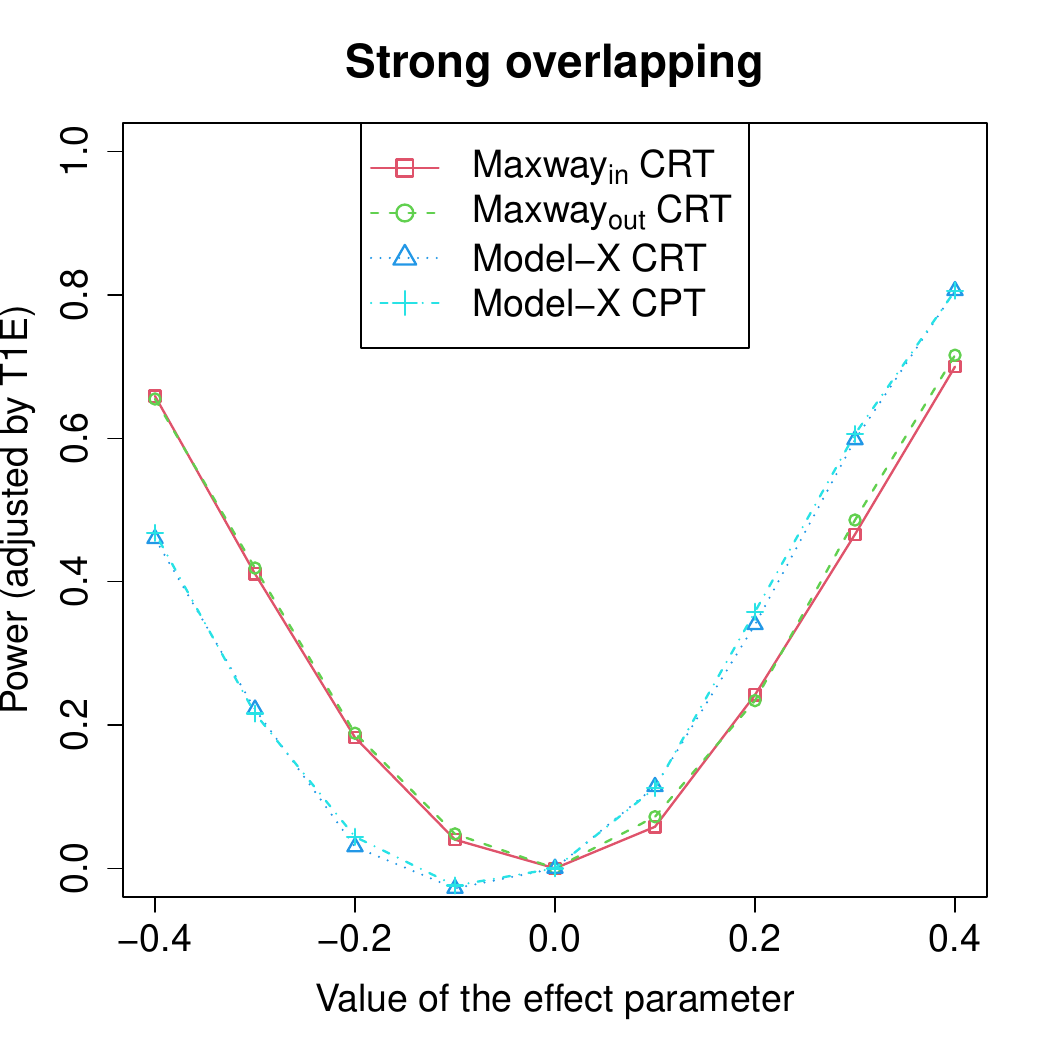}
    \includegraphics[width=0.41\textwidth]{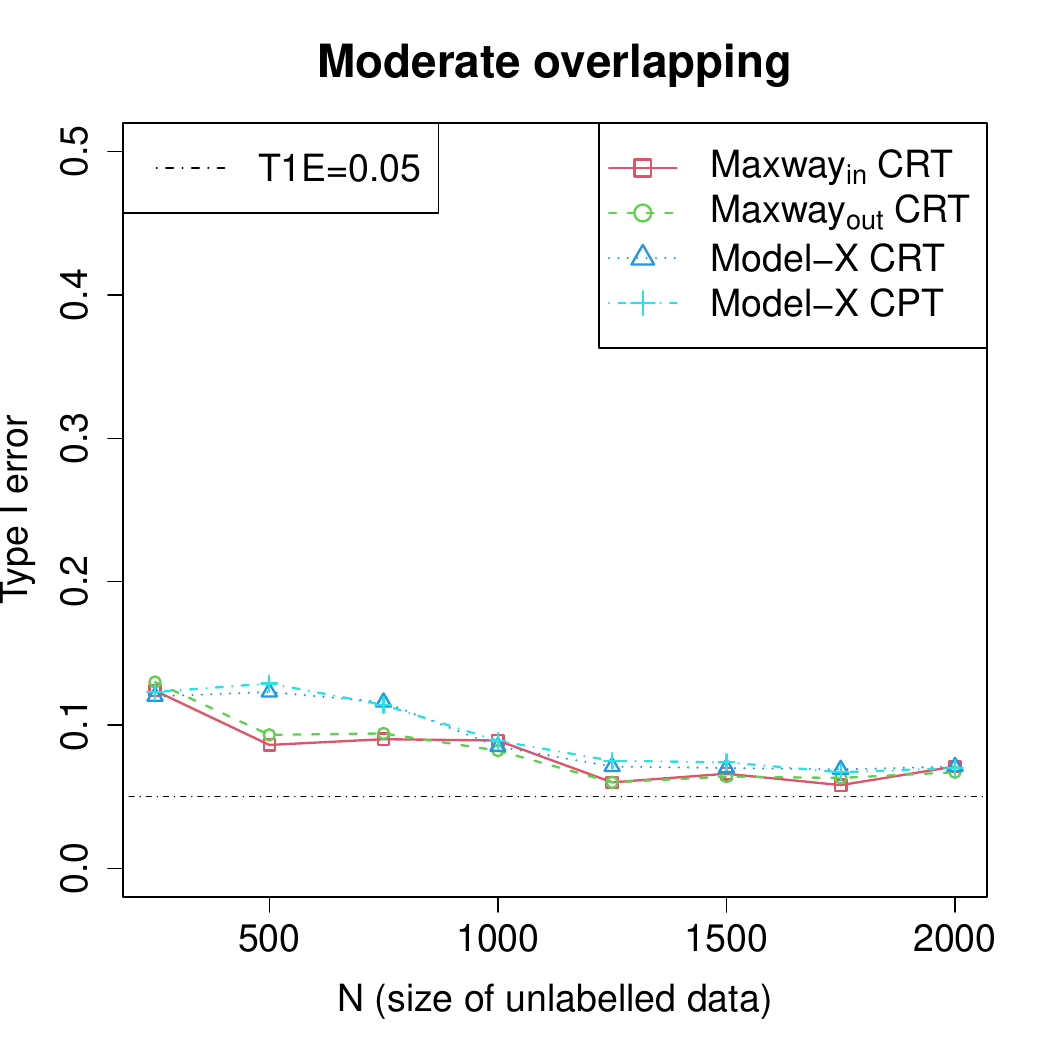}
    \includegraphics[width=0.41\textwidth]{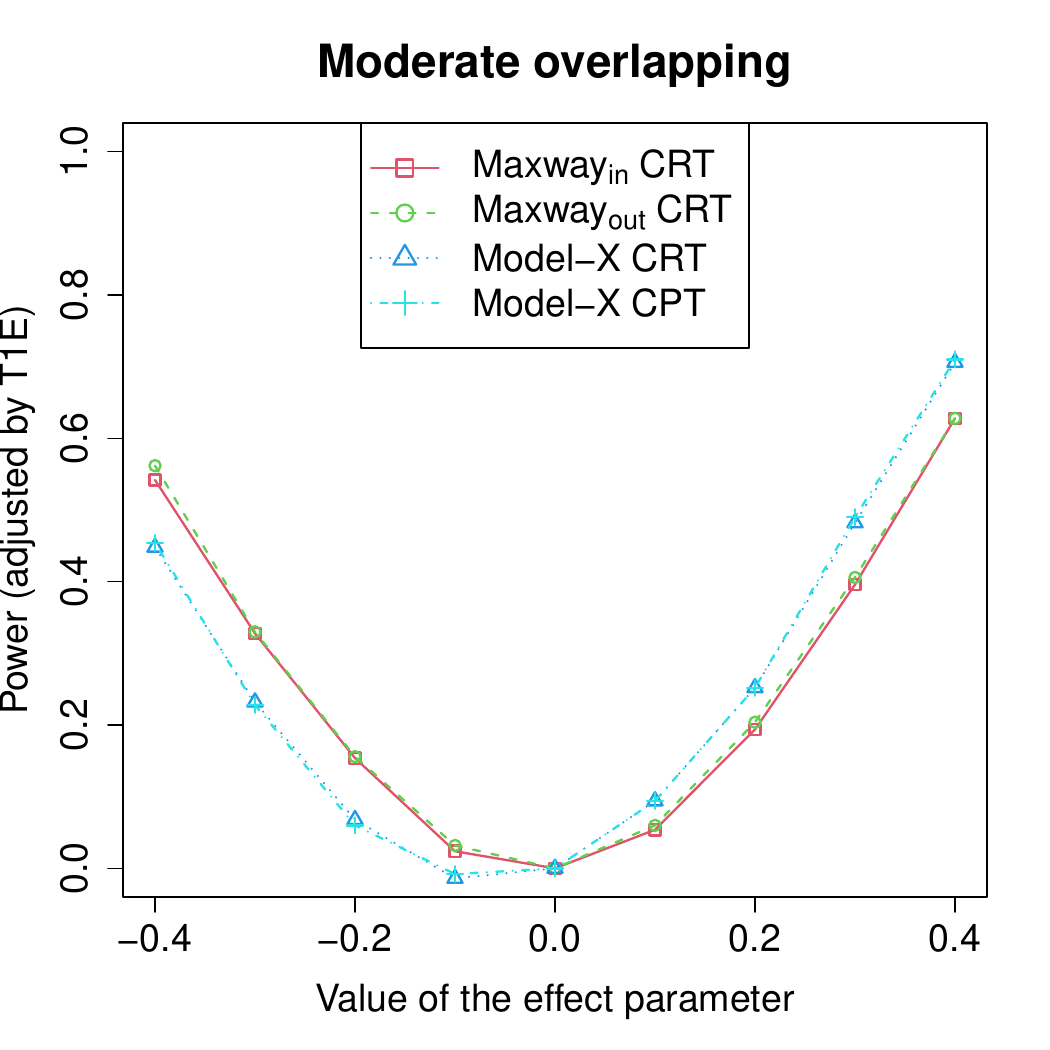}
    \includegraphics[width=0.41\textwidth]{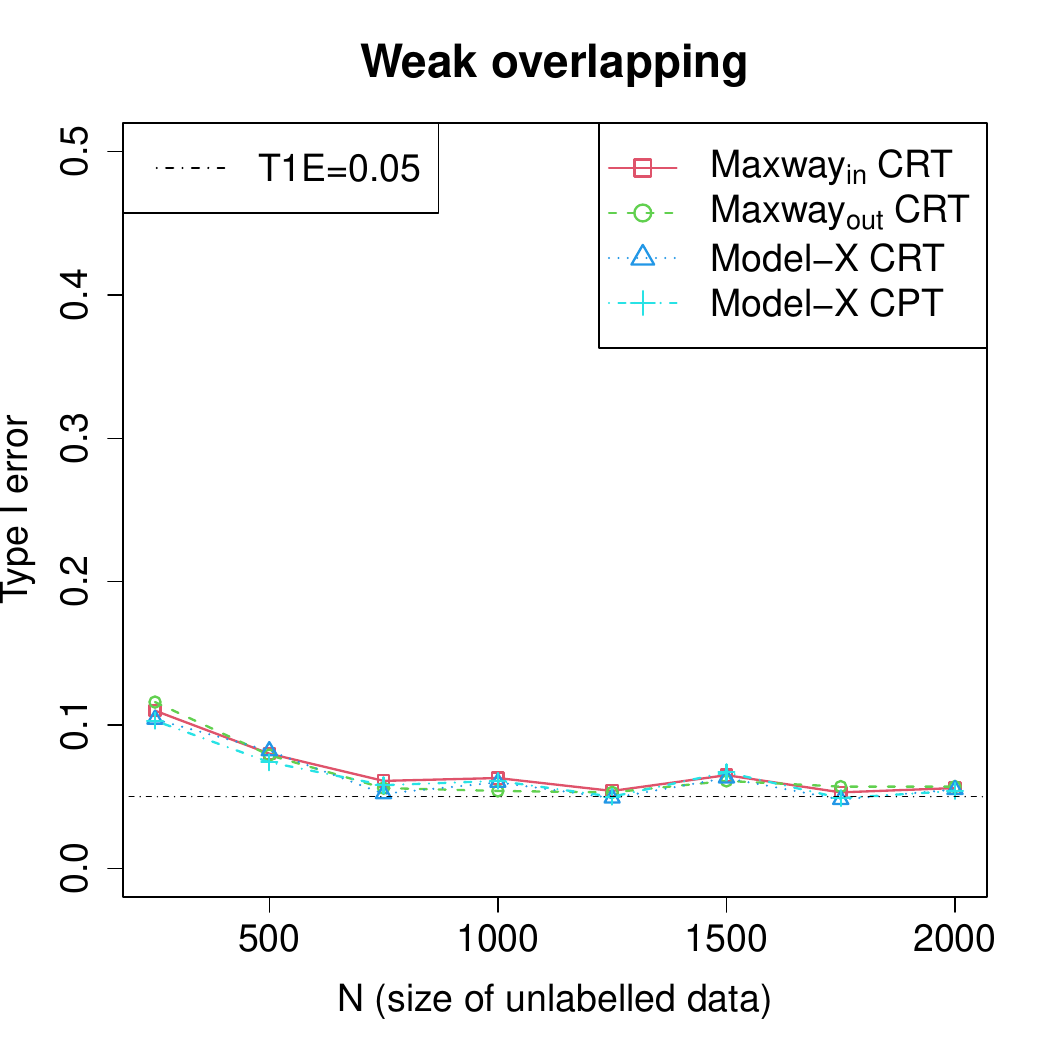}
    \includegraphics[width=0.41\textwidth]{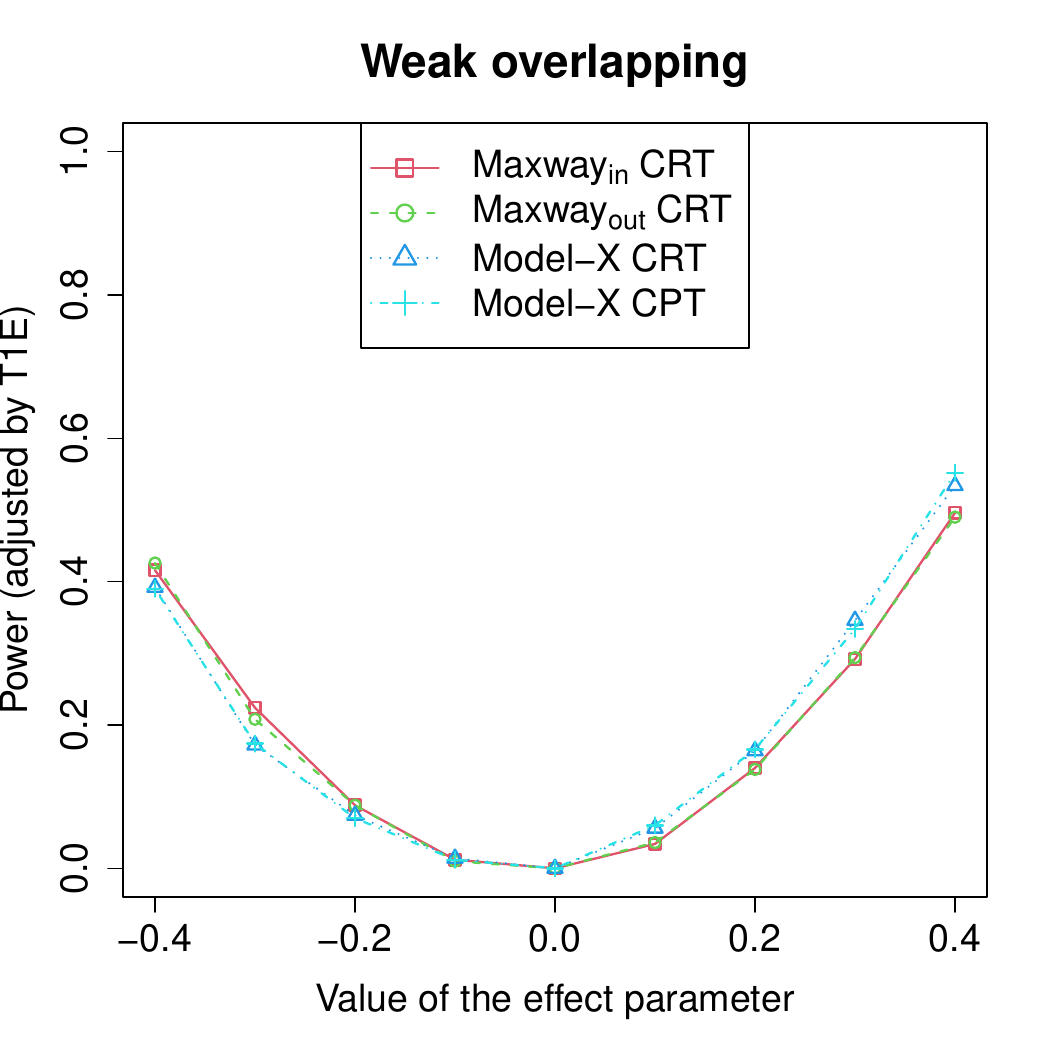}
    \caption{Type-I error and average power (adjusted by type-I error) under the three overlapping scenarios (i.e. $\eta=0,0.1,0.2$) of Configuration (SS.II) with $h(X,Z)=X$ and the d$_{\mathrm{I}}$ statistic used for testing, as introduced in Section \ref{sec:sim}. The replication number is $500$ and all standard errors are below $0.01$.} 
    \label{fig:binary:add:2}
\end{figure}

\begin{figure}[htbp]
    \centering
    \includegraphics[width=0.41\textwidth]{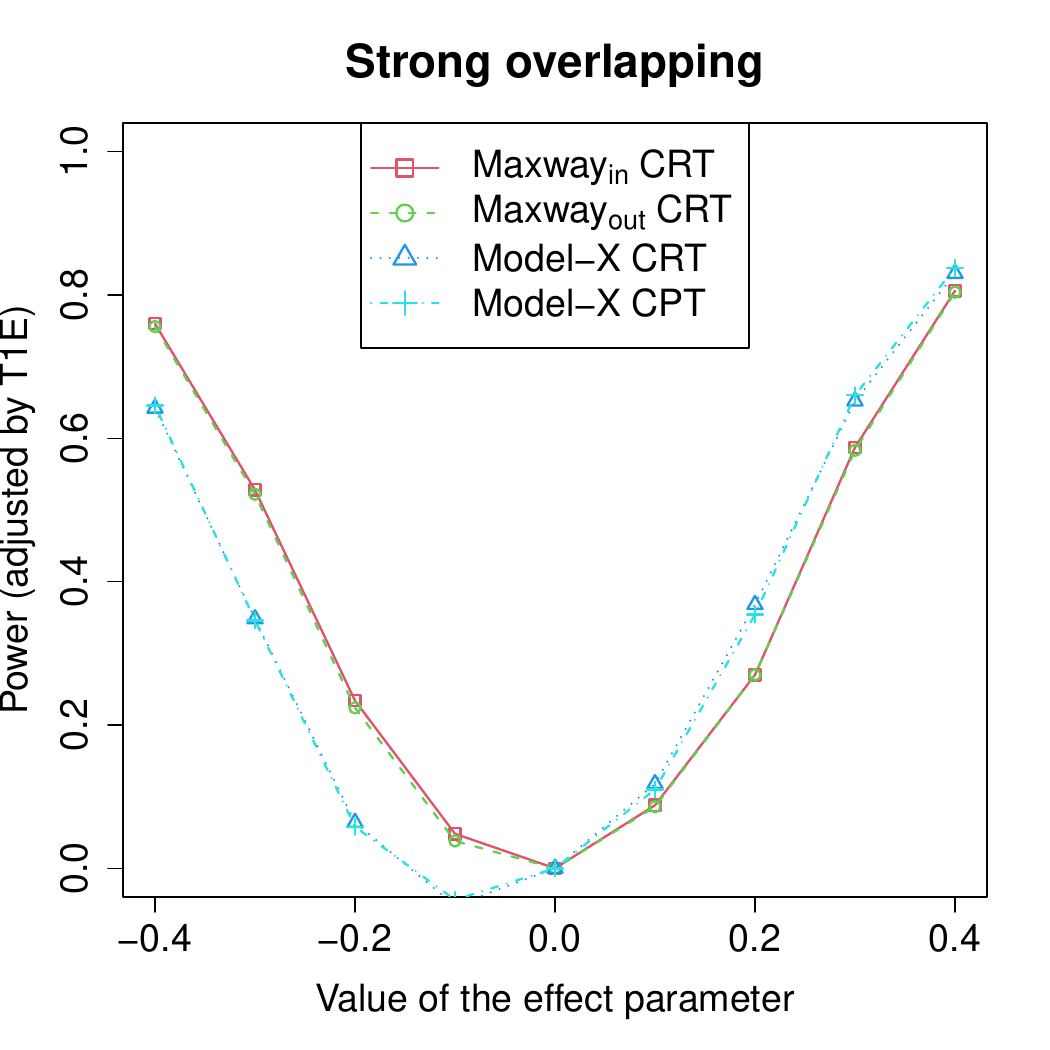}
    \includegraphics[width=0.41\textwidth]{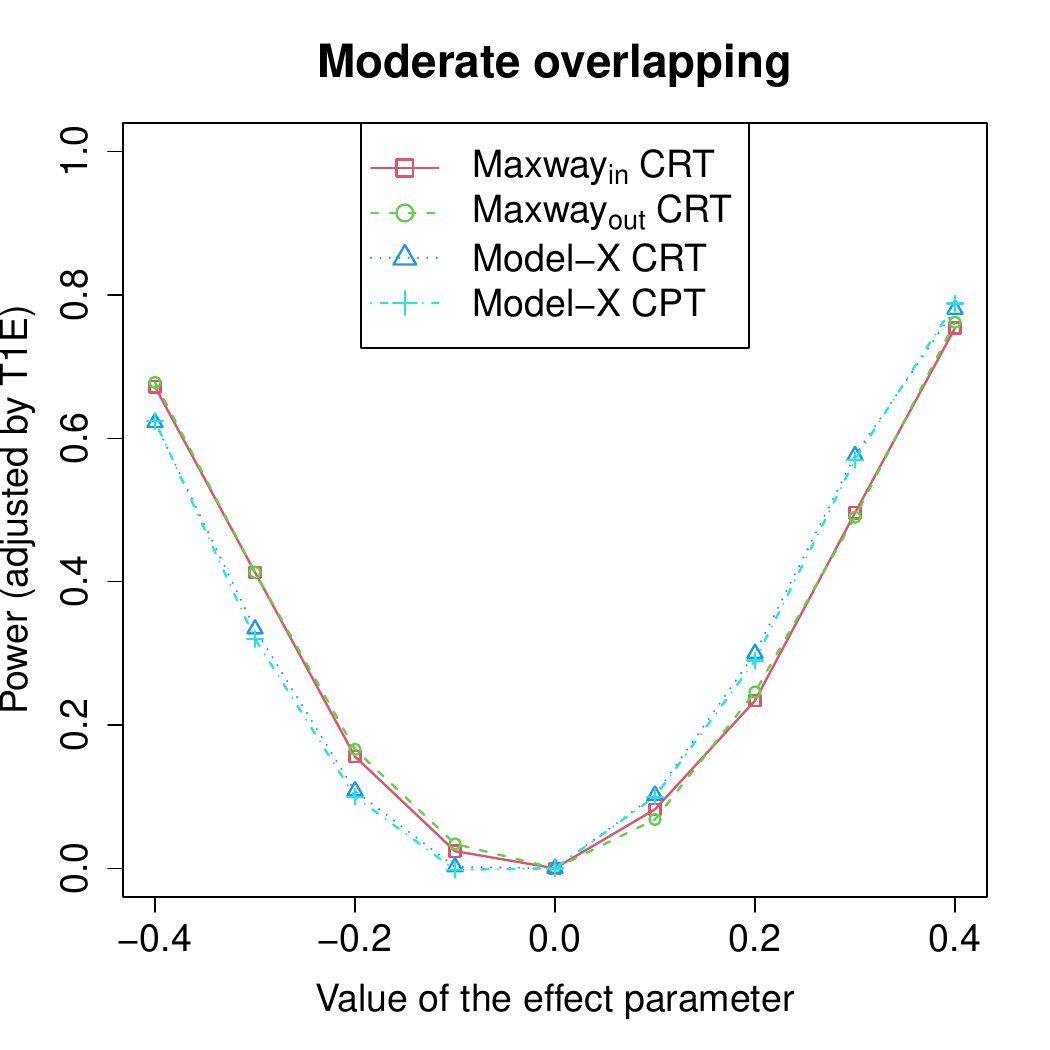}
    \includegraphics[width=0.41\textwidth]{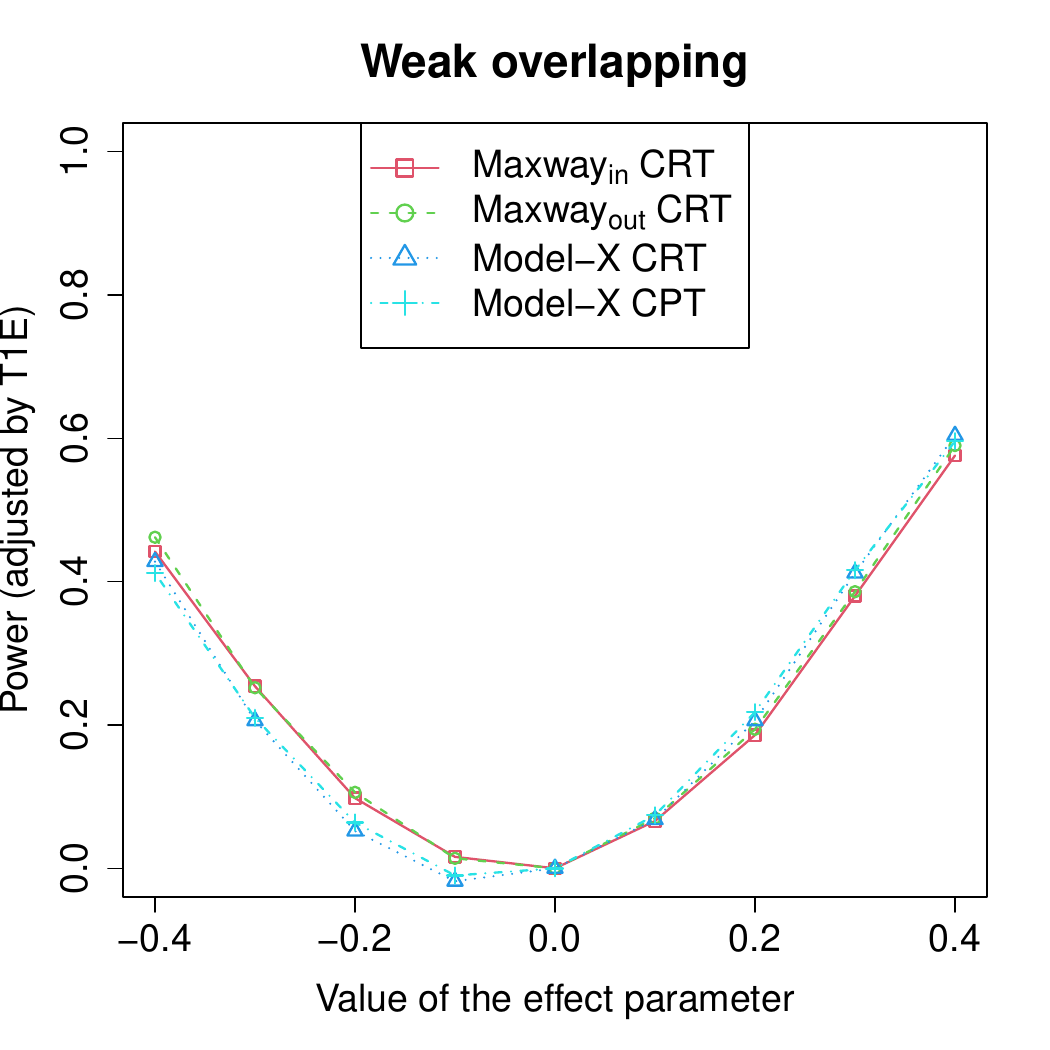}
    \caption{Type-I error and average power (adjusted by type-I error) under the three overlapping scenarios (i.e. $\eta=0,0.1,0.2$) of Configuration (SS.II) with the effect of $X$ containing interaction: $h(X,\Z)=X+X\sum_{j=1}^5Z_j$ and the d$_{\mathrm{I}}$ statistic used for testing, as introduced in Section \ref{sec:sim}. The type-I error has been presented in the left panel of Figure \ref{fig:binary:add:2}. The replication number is $500$ and all standard errors are below $0.01$.} 
    \label{fig:binary:add:3}
\end{figure}

\end{document}